%% file: main.tex
\documentclass{lmcs}
\pdfoutput=1

\usepackage{lastpage}
\renewcommand{\dOi}{14(3:2)2018}
\lmcsheading{}{1--\pageref{LastPage}}{}{}%
{Mar.~31,~2017}{Jul.~20,~2018}{}

\usepackage{amsmath}
\usepackage{xspace}
\usepackage{comment}
\usepackage{mathrsfs}  
\usepackage{amssymb}
\usepackage{amsthm}
\usepackage{stmaryrd}
\usepackage{bussproofs}
\usepackage{wrapfig}
\usepackage{hyperref}
\usepackage{mathdots}
\usepackage{upgreek}
\usepackage{tikz}
\usepackage{proof}
\usetikzlibrary{arrows,shapes,positioning}
\usetikzlibrary{calc,decorations.markings}
\usepackage{float}
  \floatstyle{boxed} 
  \restylefloat{figure}
\include{include/macros}

\newcounter{prob}
\newcounter{temp}
\newcounter{sec}

\theoremstyle{plain}
\newtheorem{lemma}[thm]{Lemma}
\newtheorem{proposition}[thm]{Proposition}
\newtheorem{example}[thm]{Example}
\newtheorem{remark}[thm]{Remark}

\newtheorem{notation}[thm]{Notation}
\newtheorem{theorem}[thm]{Theorem}
\newtheorem{definition}[thm]{Definition}
\newtheorem*{convention}{Convention}
\newenvironment{corollary}{\begin{cor}}{\end{cor}}

\newtheorem{problem}[prob]{Problem}
\usepackage%
   {todonotes}

\begin{document}
\title{Relational Graph Models at Work}

\author{Flavien Breuvart}
\author{Giulio Manzonetto}
\author{Domenico Ruoppolo}
\address{Universit\'e Paris~13, Laboratoire LIPN, CNRS UMR 7030, France}
  \email{\{flavien.breuvart,giulio.manzonetto,domenico.ruoppolo\}@lipn.univ-paris13.fr}

\maketitle

\begin{abstract}
We study the relational graph models that constitute a natural subclass of relational models of \lam-calculus.
We prove that among the \lam-theories induced by such models there exists a minimal one, and that the corresponding relational graph model is very natural and easy to construct. 
We then study relational graph models that are fully abstract, in the sense that they capture some observational equivalence between \lam-terms.
We focus on the two main observational equivalences in the \lam-calculus, the $\lambda$-theory $\Hpl$ generated by taking as observables the $\beta$-normal forms, and $\Hst$ generated by considering as observables the head normal forms.
On the one hand we introduce a notion of \emph{\lam-K\"onig} model and prove that a relational graph model is fully abstract for $\Hpl$ if and only if it is extensional and \lam-K\"onig.
On the other hand we show that the dual notion of \emph{hyperimmune} model, together with extensionality, captures the full abstraction for $\Hst$.
\end{abstract}


\section*{Introduction}
\input{include/intro}
\section{Preliminaries}\label{sec:prelim}
\input{include/prelim}
\section{The Lambda Theories $\Hpl$ and $\Hst$}\label{sec:lamtheories}
\input{include/theories}
\section{The Relational Graph Models}\label{sec:rgm}
\input{include/rgms}
\section{Quantitative Properties and Approximation Theorem}\label{sec:logical}
\input{include/logical}
\section{The Minimal Relational Graph Theory}\label{sec:minimalth}
\input{include/minimalth}
\section{Characterizing Fully Abstract Relational Models of $\Hpl$}\label{sec:Hpl}
\input{include/fullabsHpl}
\section{Characterizing Fully Abstract Relational Models of $\Hst$}\label{sec:Hst}
\input{include/fullabsHst}
\section{Conclusions}
\input{include/conclusions}
\medskip

{\bf Acknowledgements.} We would like to thank Henk Barendregt, Mariangiola Dezani, Thomas Ehrhard, Jean-Jacques L\'evy, Michele Pagani, Andrew Polonsky and Simona Ronchi Della Rocca for many interesting discussions on relational models and the \lam-theory $\Hpl$.
We also wish to thank the anonymous reviewers for their valuable comments and suggestions to improve the quality of the paper. 

\bibliographystyle{plain}
\bibliography{include/bib}
\end{document}

%% file: include/macros.tex
\newcommand{\derof}{\blacktriangleright}

\newcommand{\Lam}{\ensuremath{\Lambda}}
\newcommand{\Lamo}{\ensuremath{\Lambda^o}}

\newcommand{\ler}{\sqle_{\sf r}}

\newcommand{\rg}[1]{\mathrm{rg}(#1)}

\newcommand{\bU}{\bold{U}}
\newcommand{\bV}{\bold{V}}
\newcommand{\bJ}{\bold{J}}

\newcommand{\bI}{\bold{I}}
\newcommand{\bK}{\bold{K}}
\newcommand{\onec}{\mathbf{1}}
\newcommand{\bY}{\bold{Y}}
\newcommand{\bO}{\boldsymbol{\Omega}}
\newcommand{\cnc}{\mathbf{c}}

\newcommand{\jc}{\bold{J}}

\newcommand{\blam}{\boldsymbol{\lambda}}

\newcommand{\Hplus}{\mathcal{H^+}}
\newcommand{\Hstar}{\mathcal{H^*}}
\newcommand{\Hpl}{\Hplus}
\newcommand{\Hst}{\Hstar}

\newcommand{\setof}[1]{\{{#1}\}}
\newcommand{\seqof}[1]{\langle{#1}\rangle}

\newcommand{\subt}{\restr}

\newcommand{\con}[2]{{#1}.{#2}}
\newcommand{\concat}[2]{{#1}\cdot {#2}}
\newcommand\sqle{\sqsubseteq}

\newcommand{\pathsof}[1]{\Pi(#1)}
\newcommand{\BTlabs}{\mathscr{L}}

\newcommand{\Wit}[2][\cD]{\mathsf{W}_{#1}(#2)}
\newcommand{\emptyseq}{\varepsilon}
\newcommand{\Seq}{\nat^{<\omega}}

\newcommand{\Trees}[1][\textrm{\sf rec}]{\mathbb{T}^{\infty}_{#1}}

\newcommand{\treeof}[1]{#1}

\newcommand{\funshift}[2]{#1^{\ge{#2}}}

\newcommand{\nak}[1]{\lfloor#1\rfloor}
\newcommand{\Compl}[1]{\overline{#1}}

\renewcommand{\bold}[1]{{\bf #1}}

\newcommand{\subst}[2]{\{#2/#1\}}

\newcommand{\nf}[1][\beta]{\mathsf{nf}_{\red{#1}}} 
 
\newcommand{\phnf}{\mathrm{phnf}} 

\newcommand{\NF}[1][]{\mathrm{NF}_{\red{#1}}}
\newcommand{\cdz}{\mathrm{cdz}}
\newcommand{\FV}{\mathsf{FV}}

\newcommand{\catC}{\bold{C}}
\newcommand{\One}{\bold{1}}

\newcommand{\Rel}{\bold{Rel}}
\newcommand{\MRel}{\bold{M\!Rel}}
\newcommand{\Termobj}{\top}
\newcommand{\eval}{\mathrm{ev}}

\newcommand{\pairfun}[2]{\langle#1,#2\rangle}
\newcommand{\Id}[1]{\mathrm{Id}_{#1}}

\newcommand{\Pair}[2]{\langle{#1},{#2}\rangle}
\newcommand{\Curry}{\Uplambda}

\newcommand{\Abs}{\uplambda}
\newcommand{\App}{\NF[\beta\bot]}

\newcommand{\Lint}[2][]{\llbracket #2\rrbracket^{#1}}
\newcommand{\Int}[2][]{\vert #2\vert^{#1}}

\newcommand{\With}[2]{{#1}\hspace{-2pt}\mathrel{\&}\hspace{-3pt}{#2}}

\newcommand{\Proj}[1]{\pi_{#1}}

\newcommand{\Funint}[2]{{#1}\hspace{-2pt}\Rightarrow\hspace{-2pt}{#2}}

\newcommand{\comp}{\hspace{-1pt}\circ\hspace{-1pt}}
\newcommand{\Th}{\mathrm{Th}}
\newcommand{\Thle}{\mathrm{Th_\sqsubseteq}}

\newcommand\emptymset{\omega}
\newcommand{\Mfin}[1]{\mathcal{M}_{\mathrm{f}}( #1 )}
\newcommand{\Mset}[1]{[{#1}]}

\newcommand{\Powf}[1]{\mathcal{P}_{\mathrm{f}}(#1)}

\newcommand{\st}{\ \vert\ }
\newcommand{\nat}{\mathbb{N}}

\newcommand{\Pow}[1]{\mathcal{P}(#1)}
\newcommand{\dom}{\mathrm{dom}}
\newcommand{\Var}{\mathrm{Var}}

\newcommand{\seq}[1]{\vec #1}

\renewcommand{\phi}{\varphi}
\newcommand*{\msto}[1][\beta]{\twoheadrightarrow_{\red{#1}}}

\newcommand{\code}[1]{\lceil #1\rceil}

\newcommand{\rgm}{\textrm{rgm}}

\newcommand{\BTth}{\mathcal{B}}
\newcommand{\BTset}{\mathbb{BT}}
\newcommand{\BTeth}{\mathcal{H}^+}
\newcommand{\BT}[2][]{\mathrm{BT}^{#1}(#2)}
\newcommand{\BTle}{\le_\bot\!}

\newcommand{\To}{\Rightarrow}
\newcommand{\cA}{\mathcal{A}}
\newcommand{\cB}{\mathcal{B}}

\newcommand{\cD}{\mathcal{D}}
\newcommand{\cE}{\mathcal{E}}

\newcommand{\cH}{\mathcal{H}}

\newcommand{\cO}{\mathcal{O}}

\newcommand{\cT}{\mathcal{T}}

\newcommand{\obsle}[1][\mathrm{nf}]{\sqsubseteq^{#1}}
\newcommand{\obseq}[1][\mathrm{nf}]{\equiv^{#1}}

\newcommand{\bnf}{\mathrel{::=}} 
\newcommand*{\mystackrel}[2]{\stackrel{\raise-2pt\hbox{$\scriptstyle\!#1$}}{#2}}
\newcommand*{\mystackrelrev}[2]{\stackrel{\raise-2pt\hbox{$\scriptstyle\,#1$}}{#2}}
\newcommand*{\mystackrelequi}[2]{\stackrel{\raise-2pt\hbox{$\scriptstyle#1$}}{#2}}

\newcommand*{\red}[1]{\mathtt{#1}}  
\newcommand*{\redto}[1][]{\rightarrow_{\red{#1}}}

\newcommand*{\mstoinv}[1][\beta]{\,_{#1}\!\!\twoheadleftarrow}

\newcommand{\leeta}[1][]{\le^\eta_{#1}}
\newcommand{\geeta}[1][]{\ge^\eta_{#1}}

\newcommand{\app}{\#\textrm{app}}
\newcommand{\appl}{\textrm{App}}
\newcommand{\LLT}{\lam\cT}
\newcommand\Lamb{\Lambda_\bot} 

\newcommand{\rel}[1]{\ \mathtt{#1}\ } 
\newcommand{\da}[1]{\mathsf{da}(#1)} 
\newcommand{\rR}{\mathtt{R}} 

\newcommand\hole[1]{[#1]} 
\newcommand{\lam}{\ensuremath{\lambda}}
\newcommand{\lamb}{\ensuremath{\lambda\bot}}

\newcommand{\size}[1]{{\# #1}}

\newcommand{\restr}[2]{#1\hspace{-3pt}\upharpoonright_{\!#2}}
\newcommand{\fpath}[2]{\langle #1|#2 \rangle}

\newcommand{\rednum}[1]{\#\mathrm{red_{\beta}}(#1)} 

\newcommand{\supp}{\mathrm{supp}}

\newcommand{\Lamrule}[2]{\mathtt{Lam}\big(#1,#2\big)}

\newcommand{\Apprule}[2]{\mathtt{App}\big(#1,#2\big)}
\newcommand{\Env}[1]{\mathsf{Env}_{#1}}
\newcommand{\EnvD}{\mathsf{Env}_{\mathcal{D}}}
\newcommand{\justlike}{\simeq}

{\begin{list}%
       {$-$}%
       {\setlength{\itemsep}{0pt}
     \setlength{\parsep}{0.5pt}
     
     \setlength{\partopsep}{0pt}
     \setlength{\leftmargin}{2em}
     \setlength{\labelwidth}{1em}
     \setlength{\labelsep}{0.3em}}}%
{\end{list}}


%% file: include/intro.tex
The untyped \lam-calculus is a paradigmatic programming language introduced by Church in~\cite{Church41}. It has a prominent role in theoretical computer science~\cite{Bare} and, despite its very simple syntax, it is Turing-complete \cite{Kleene36, Turing37}.
Its denotational models have been fruitfully used for proving the consistency of extensions of $\beta$-convertibility, called \emph{\lam-theories}, and for exposing operational features of \lam-terms.
The first model of \lam-calculus, $\mathscr{D}_\infty$, was defined by Scott in the pioneering article~\cite{Scott72}. 
Subsequently, a wealth of models have been introduced in various categories of domains and classified into \emph{semantics} according to the nature of their representable functions.
Scott's \emph{continuous semantics}~\cite{Scott76} corresponds to the category whose objects are complete partial orders and morphisms are continuous functions. 
The \emph{stable semantics}~\cite{Berry78} and the \emph{strongly stable semantics}~\cite{BucciarelliE91} are refinements of the continuous semantics, introduced to capture the notion of ``sequential'' continuous function.
In each of these semantics all the models come equipped with a partial order, and some of them, called \emph{webbed models}, are built from lower level structures called ``webs''~\cite{Berline00}. 
The simplest class of webbed models is the class of \emph{graph models}~\cite{Longo83}, which was isolated in the seventies by Plotkin, Scott and Engeler within the continuous semantics~\cite{Engeler81,Plotkin71,Scott76}.

In each of the aforementioned semantics there exists a continuum of models inducing pairwise distinct \lam-theories. Nevertheless, certain models are particularly important because they allow to capture operational properties of the \lam-terms. 
For instance, two \lam-terms have the same interpretation in Engeler's graph model $\mathscr{E}$ exactly when they are equal in the \lam-theory $\cB$, which equates all \lam-terms having the same B\"ohm tree.
The main technical tool for proving such a result is the Approximation Theorem~\cite{HonsellR92}, stating that the interpretation of a \lam-term is given by the supremum of the interpretations of the finite approximants of its B\"ohm tree.
Other models are significant because they are \emph{fully abstract}, which means that the induced \lam-theory captures some observational equivalence between \lam-terms. 
A celebrated result by Hyland~\cite{Hyland76} and Wadsworth~\cite{Wadsworth78} shows that Scott's $\mathscr{D}_\infty$ is fully abstract for the \lam-theory $\Hst$, that corresponds to the observational equivalence where the observables are the head normal forms.
In~\cite{CoppoDZ87}, Coppo, Dezani-Ciancaglini and Zacchi constructed a \emph{filter model} $\mathscr{D}_\cdz$ and proved that it is fully abstract for $\Hpl$, the observational equivalence where the observables are the $\beta$-normal forms. 
The \lam-theory $\Hpl$ is the original extensional observational theory defined by Morris in his thesis~\cite{Morristh}.
It is maybe less ubiquitously studied in the literature than $\Hst$ but we believe is equally important. 
For instance, its notion of observables is central in the B\"ohm Theorem~\cite{Bohm68} and in other separability results~\cite{CoppoDR78}.

\subsection*{Graph Models in the Relational Semantics.} 
In the present paper we focus on the \emph{relational semantics} of \lam-calculus, that has been introduced by Girard as a quantitative model of linear logic in~\cite{Girard88}.
The first concrete examples of relational models of \lam-calculus were built in~\cite{BucciarelliEM07,HylandNPR06}.
Recently, Manzonetto and Ruoppolo individuated the subclass of \emph{relational graph models} encompassing all previous examples~\cite{ManzonettoR14}.
The definition of a relational graph model (Definition~\ref{def:graphrel}) really is the relational analogue of the definition of a graph model living in the continuous semantics.
In particular, relational graph models can be built by free completion and by forcing like the continuous ones.
However, from the point of view of the induced \lam-theories, they share more similarities with filter models.
For instance, the relational graph model $\cD_\omega$ built in \cite{BucciarelliEM07} has the same theory as Scott's $\mathscr{D}_\infty$, namely it is fully abstract for $\Hst$~\cite{Manzonetto09}.
Similarly, the model $\cD_\star$ from~\cite{ManzonettoR14} is fully abstract for $\Hpl$, like the filter model $\mathscr{D}_\cdz$.
On the other hand, no graph model living in the continuous semantics can represent such \lam-theories because no graph model is extensional.
When comparing relational graph models with filter models inducing the same \lam-theory, one can see that the former are in general simpler because their elements are not partially ordered. 
Moreover, an element $\sigma$ in the relational interpretation of a \lam-term $M$ carries information concerning intensional properties of $M$.
In particular, from $\sigma$ it is possible to compute a bound to the number of head-reduction steps towards its normal form and infer the amount of resources consumed by $M$ during such a reduction sequence~\cite{deCarvalho09,CarvalhoPF11}.

\subsection*{Relational Graph Models as Type Systems}
The Stone duality between filter models and intersection type systems has been widely studied in the literature, e.g.,~\cite{Abramsky91,CoppoDHL84,HonsellR92,BareTypes,SalvatiMGB12}. 
(We refer to~ Ronchi Della Rocca and Paolini's book for a thoughtful discussion \cite[Ch.~13]{RonchiP04}.)
Such a correspondence shows that some interesting classes of domain-based models can be described in logical form. 
The intuition is that a functional intersection type $\alpha_1\land\cdots\land\alpha_n\to\beta$ can be seen as a  continuous step function sending the set $\{\alpha_1,\dots,\alpha_n\}$ to the element $\beta$.
Types come equipped with inference rules reflecting the structure of the underlying domain.
In~\cite{PaoliniPR15}, Paolini et al.\ introduce the \emph{strongly linear} relational models (a class encompassing relational graph models, but included in the \emph{linear} relational models of~\cite{Manzonetto12}) and show that they can be represented as relevant (i.e., without weakening) intersection type systems where the intersection is a non-idempotent operation (it is actually a linear logic tensor~$\otimes$).
The idea, already present in~\cite{deCarvalho09}, is that in the absence of idempotency and partial orders the type $\alpha_1\land\cdots\land\alpha_n\to\beta$ can be seen as a relation associating the multiset $[\alpha_1,\dots,\alpha_n]$ with the element $\beta$.
As a consequence of the work in~\cite{PaoliniPR15}, all relational graph models can be presented in logical form, that is, as non-idempotent intersection type systems.
We use this kind of representation to expose and exploit their quantitative features.

\subsection*{The Approximation Theorem.}
Besides soundness (Theorem~\ref{thm:soundness}), one of the main properties enjoyed by relational graph models is the Approximation Theorem (Theorem~\ref{thm:app}).
Typically such a theorem is proved by exploiting Wadsworth's stratified refinements of the $\beta$-reduction~\cite{Wadsworth78}, that also work in the relational framework as shown in~\cite{Manzonetto09}.
Other techniques are based on Tait and Girard's \emph{reducibility candidates}~\cite{Tait67,Girard72}, that are  widespread in logic and the theory of programming languages~\cite{Reynolds83,Krivine93,KrivineClass,BareTypes}, but notoriously give rise to proofs based on impredicative principles. 
Thanks to its quantitative nature, in the context of the relational semantics it is possible to get rid of the traditional methods and provide a combinatorial proof. 
This is the case of the proof given in~\cite{ManzonettoR14} by relying on these facts:
\begin{itemize}
\item 
	Relational graph models are also models of Ehrhard's \emph{differen\-tial \lam-calculus}~\cite{lambdadiff} and Tranquilli's resource calculus~\cite{phdtranquilli}. 
	This follows from the fact that they all are \emph{linear} reflexive objects in the Cartesian closed differential category $\MRel$~\cite{Manzonetto12}.
\item 
	An easy induction shows that the interpretation of a \lam-term $M$ in a relational graph model is equal to the interpretation of its \emph{Taylor expansion}~\cite{lambdadiff}, which is a representation of $M$ as a power series of resource approximants (replacing in a way the finite approximants of its B\"ohm tree).
\item The usual Approximation Theorem follows from the above result by applying a theorem due to Ehrhard and Regnier~\cite{bohmtaylor} stating that the normal form of the Taylor expansion of $M$ coincides with the Taylor expansion of its B\"ohm tree.
\end{itemize}
In Section~\ref{subsec:AT} we provide a new combinatorial proof of the Approximation Theorem by exploiting the logical presentation discussed above. 
We are going to associate a measure with the derivation tree $\pi$ of $\Gamma\vdash M : \alpha$ and show that when $M$ $\beta$-reduces to $N$ by contracting a redex $R$ two cases are possible: either there exists a derivation of $\Gamma\vdash N : \alpha$ having a strictly smaller measure, or $\pi$ is a derivation of $\,\Gamma\vdash M\subst{R}{\bot} : \alpha$, where $M\subst{R}{\bot}$ denotes the approximant obtained by substituting a constant $\bot$ for the redex $R$ in $M$. 
In both cases, either the measure of the derivation or the number of redexes in $M$ has decreased.
Therefore the Approximation Theorem follows by a simple induction over the ordinal $\omega^2$.

\subsection*{The Minimal Relational Graph Theory.} 
Every relational graph model induces a \lam-theory through the kernel of its interpretation function. 
We call \emph{relational graph theories} those \lam-theories induced by some relational graph models. 
A natural question that arises is what \lam-theories are in addition relational graph theories.
We do not provide a characterization, but we show that the \lam-theories $\cB$, $\Hpl$ and $\Hst$ are. 
Another question is whether there exists a minimal relational graph theory: for instance, in~\cite{BucciarelliS08} Bucciarelli and Salibra proved that the minimal \lam-theory among the ones represented by usual graph models exists, but their construction of the minimal model is complicated and what \lam-terms are actually equated in the minimal theory remains a mystery. 
In Section~\ref{sec:minimalth} we show not only that a minimal relational graph theory exists, but also that such a \lam-theory  is actually $\cB$. Also, the corresponding model~$\mathcal{E}$ is very simple to define (its construction is actually analogous to the one of Engeler's graph model). 
Moreover, we prove that even the preorder induced by $\cE$ on \lam-terms is minimal among representable inequational theories.
Our model $\cE$ shares many properties with Ronchi Della Rocca's filter model defined in~\cite{Ronchi82}.

\subsection*{Characterizing Fully Abstract Models}
In the literature there are many full abstraction theorems, namely results showing that some observational equivalence arises as the theory of a suitable denotational model.
However, until recently, researchers were only able to prove full abstraction results for individual models~\cite{Hyland76,Wadsworth78,CoppoDZ87}, or at best to provide sufficient conditions for models living in some class to be fully abstract \cite{Manzonetto09,ManzonettoR14,GouyTh}. 
For instance, Manzonetto showed in \cite{Manzonetto09} that a model of \lam-calculus living in a cpo-enriched Cartesian closed category is fully abstract for $\Hst$ whenever it is a ``well stratified $\bot$-model\footnote{Using the terminology of \cite{Manzonetto09}.}''.
More recently, he proved in collaboration with Ruoppolo that every extensional relational graph model preserving the polarities of the empty multiset (in a technical sense) is fully abstract for~$\Hpl$~\cite{ManzonettoR14}.
A substantial advance in the study of full abstraction was made in~\cite{Breuvart14}, where  Breuvart proposed a notion of \emph{hyperimmune} model of \lam-calculus, and showed that a $K$-model\footnote{%
The class of \emph{$K$-models}, which contains all graph models, was isolated by Krivine in~\cite{Krivine93}.} living in the continuous semantics is fully abstract for $\Hst$ if and only if it is extensional and hyperimmune, thus providing a characterization.
In Section~\ref{sec:Hpl} we define the dual notion of \emph{\lam-K\"onig} model and prove that a relational graph model is fully abstract for $\Hpl$ exactly when it is extensional and \lam-K\"onig.
In Section~\ref{sec:Hst} we show that the notion of  hyperimmune continuous model has a natural counterpart in the relational semantics and that, also in the latter case, together with extensionality gives a characterization of all relational graph models fully abstract for~$\Hst$.

\subsection*{Related Works} 
The primary goal of this article is to provide a uniform and self-contained treatment of relational graph models and their properties. In particular, we present some semantical results recently appeared in the conference papers~\cite{ManzonettoR14,BreuvartMPR16}. 
(The syntactic results in~\cite{BreuvartMPR16} will be the subject of a different paper~\cite{IntrigilaMP18} in connection with the $\omega$-rule.)
Besides giving more detailed proofs and examples, we provide several original results, like a quantitative proof of the Approximation Theorem, the characterization of the minimal representable theory, and the characterization of all relational graph models that are fully abstract for $\Hst$.
A natural comparison is with the article~\cite{PaoliniPR15}, where Paolini \emph{et al.}\ introduce the notion of strongly linear relational model and show that such models correspond to their notion of \emph{essential} type systems.
We remark that their work rather focuses on the properties enjoyed by those systems, like (weighted) subject reduction/expansion and adequacy, while we focus on the representable \lam-theories and provide general full abstraction results.

The relational semantics, being very versatile, can also be used to model the call-by-value and the call-by-push-value \lam-calculus~\cite{Levy06}, as well as non-deterministic \cite{deLiguoroP95,DezaniLP98,AschieriZ13} and resource sensitive extensions~\cite{lambdadiff,bohmtaylor} of \lam-calculus. 
We refer the reader to~\cite{Ehrhard12,CarraroG14} for a relational semantics of the call-by-value \lam-calculus and to~\cite{EhrhardG16} for the call-by-push-value.
For non-deterministic calculi, see~\cite{BucciarelliEM12} in the call-by-name setting and \cite{Diaz-CaroMP13} in the call-by-value one.
Relational models of differential and resource calculi have been studied in~\cite{PaganiR10,Manzonetto12,Breuvart13}.
The relational semantics has been generalized by considering multisets with infinite multiplicities to build models that are not sensible~\cite{CarraroES10}, and by replacing relations with matrices of scalars to provide quantitative models of non-deterministic PCF~\cite{LairdMMP13}.
An even more abstract perspective, where the categorical notion of profunctors takes the role of relations, was contemplated in~\cite{jdm096,Hyland10}.

\subsection*{Outline}
In Section~\ref{sec:prelim} we present some notions and notations, mainly concerning \lam-calculus, that are useful in the rest of the article. 
In Section~\ref{sec:lamtheories} we review some literature concerning the observational equivalences corresponding to $\Hpl$ and $\Hst$, and their characterizations in terms of extensional equivalences on B\"ohm trees.
In Section~\ref{sec:rgm} we define the class of relational graph models, show how to build them via free completion of a partial pair, and prove the soundness.
In Section~\ref{sec:logical} we provide the presentation of relational graph models in logical form, exhibit their quantitative features and prove the Approximation Theorem.
Section~\ref{sec:minimalth} is devoted to present the Engeler-style relational graph model and prove that the induced (inequational) theory is minimal.
In Section~\ref{sec:Hpl} and Section~\ref{sec:Hst} we provide the characterizations of all relational graph models that are fully abstract for $\Hpl$ and $\Hst$.


%% file: include/prelim.tex
\subsection{Coinduction.}
In this article we will often use coinductive structures and coinductive reasoning. 
We recall here some basic facts about coinduction and discuss some terminology, but we mainly take them for granted.
Nice tutorials on the subject are~\cite{JacobsR97,KozenS17}; see also~\cite{abelPTS13}.

A coinductive structure, corresponding to a coinductive datatype, is just the greatest fixed point over a grammar, or equivalently the terminal coalgebra over the corresponding functor.
We also consider coinductive propositions (and relations), that are the greatest propositions over such coinductive structures respecting the structural constraints. Coinductive propositions are proven by infinite derivation trees that are coinductive structures.

As opposed to the inductive principle, the coinductive one is unable to destruct a coinductive object, it is only able to construct one. More precisely, we can prove by coinduction a statement of the form $\forall x{\in} X\, \exists y{\in} Y\, \phi(x,y)$ where $Y$ and $\phi$ are coinductively defined. 
Consider for now the simpler case where we want to prove $\forall x\!\in\! X \, \phi(x)$ for some coinductive proposition $\phi$. In order to apply the coinduction principle, it is sufficient to show
$$
	\textstyle\forall x\in X\;\, \exists \, n\in\nat\;\, \exists \, \seq x\in X^n\; \Big(\bigwedge_{i\le n} \phi(x_i)\;\, \stackrel{p} {\Rightarrow }\;\, \phi(x) \Big)
$$
where the ``$p$'' in $\stackrel{p} {\Rightarrow }$ refers to the productivity of the implication. In this case ``productive'' means that in the proof of the implication some patterns of the grammar coinductively defining $\phi$ are actually applied to the various $\phi(x_i)$. 
Usually $x$ is a certain (inductive or coinductive) structure and $ \seq x$ can be seen as an ``unfolding'' of $x$. 

In the general case, namely $\forall x{\in} X\,\exists y{\in} Y\, \phi(x,y)$, it is sufficient to show that
$$
	\textstyle\forall x\in X \;\,\exists \, n\in\nat\;\,\exists\,\seq x\in X^n\;\,\forall \seq y\in Y^n\;\,\exists\,^p y \in Y\;\, \Big(\bigwedge_{i\le n} \phi(x_i,y_i)\;\, \stackrel{p} {\Rightarrow }\;\,\phi(x,y)\Big)  .
$$
The ``productivity'' of the existential $\exists\,^p$ simply requires that $y = \mathsf{p}(y_1,\dots, y_n)$ for some pattern $\mathsf{p}$ of the grammar coinductively generating $Y$.
Notice that the negation of a coinductive statement is inductive and \emph{vice versa}, a fact that we use in the proof of Lemma~\ref{lemma:bourdelique}.

Since structural coinduction has been around for decades and many efforts have been made in the community to explain why it should be used as innocently as structural induction, in our proofs we will not reassert the coinduction principle every time it is used. 
Borrowing the terminology from~\cite{KozenS17}, we say that we apply the ``coinductive hypothesis'' whenever the coinduction principle is applied.
We believe that this mathematical writing greatly improves the readability of our proofs without compromising their correctness: the suspicious reader can study~\cite{KozenS17} where it is explained how this informal terminology actually corresponds to a formal application of the coinduction principle. 

\subsection{Sets, Functions and Multisets}\label{subs:msets}
We denote by $\nat$ the set of natural numbers and by $\nat^+$ the set of strictly positive natural numbers.

Let $A,B$ be two sets. We write $\Pow{A}$ for the powerset of $A$, and $\Powf A$ for the set of all finite subsets of $A$.
Given a function $f : A \to B$ we write $\dom(f)$ for its \emph{domain}.

A \emph{multiset} over $A$ is a partial function $a : A\to\nat^+$. 
Given $\alpha\in A$ and a multiset $a$ over $A$, the \emph{multiplicity of $\alpha$ in $a$} is given by $a(\alpha)$.
A multiset $a$ is called \emph{finite} if  $\dom(a)$, which is called its \emph{support}, is finite. 
A finite multiset $a$ will be represented as an unordered list of its elements 
$[\alpha_1,\dots,\alpha_n]$, possibly with repetitions. The empty multiset is denoted by~$\emptymset$. 
We write $\Mfin{A}$ for the set of all finite multisets over $A$.
Given $a_1,a_2\in\Mfin{A}$, their \emph{multiset union} is denoted by $a_1 + a_2$ and defined as a pointwise sum.

\subsection{Sequences and Trees}\label{ssec:notations}

We denote by $\Seq$ the set of all finite sequences over $\nat$.
An arbitrary sequence is of the form $\sigma =\seqof{n_1,\dots,n_k}$. The empty sequence is denoted by $\emptyseq$.

Let $\sigma = \seqof{n_1,\dots,n_k}$ and $\tau = \seqof{m_1,\dots,m_{k'}}$ 
be two sequences and let $n\in\nat$. 
We write: 
\begin{itemize}
\item $\con{\sigma}{n}$ for the sequence $\seqof{n_1,\dots,n_k,n}$,
\item $\concat{\sigma}{\tau}$ for the concatenation of $\sigma$ and $\tau$,
that is for the sequence
$\seqof{n_1,\dots,n_k,m_1,\dots,m_{k'}}$. 
\end{itemize}
Given a function $f : \nat \to \nat$, its \emph{prefix of length $n$} is the sequence $\fpath{f}{n} = \seqof{f(0), \dots, f(n-1)}$.

\begin{definition}[Trees and subtrees]\label{def:tree}\
\begin{itemize}
\item
 A \emph{tree} is a partial function $T : \Seq\to\nat$ such that 
$\dom(T)$ is closed under prefixes and for all $\sigma\in\dom(T)$ and $n\in\nat$ we have $\con{\sigma}n \in\dom(T)$ if and only if $n < T(\sigma)$. 
\item We write $\mathbb{T}$ for the set of all trees.
\item
The \emph{subtree of $T$ at $\sigma$} is the tree $\subt T\sigma$ defined by $\subt T\sigma\!\!(\tau) = T(\concat\sigma\tau)$ for all $\tau\in\Seq$.
\end{itemize}
\end{definition}

The elements of $\dom(T)$ are called \emph{positions}. 
For all $\sigma\in\dom(T)$, $T(\sigma)$ gives the number of children of the node in position $\sigma$.
Hence $T(\sigma) = 0$ when $\sigma$ corresponds to a leaf.

\begin{definition}
A tree $T$ is called: 
\begin{itemize}
\item \emph{recursive} if the function $T$ is partial recursive (after coding); 
\item \emph{finite} if $\dom(T)$ is finite; 
\item \emph{infinite} if it is not finite.
\end{itemize}
\end{definition}
We denote by $\Trees[]$ (resp.~$\Trees$) the set of all infinite (resp.\ recursive infinite) trees.

\begin{definition}[Infinite paths]\label{def:inf_paths}
A function $f :\nat\to\nat$ is \emph{an infinite path of $T$} if $\fpath f n\in\dom(T)$ for all $n\in\nat$. 
We denote by $\pathsof{T}$ the set of all infinite paths of $T$.
\end{definition}
By K\"onig's lemma, a tree $T$ is infinite if and only if $\pathsof{T}\neq\emptyset$.

\subsection{Category Theory}\label{subsec:categories} 
Concerning category theory we mainly use the notations from~\cite{AspertiL91}.

Let $\catC$ be a category and $A,B,C$ be arbitrary objects of $\catC$. 
We write $\catC(A,B)$ for the homset of morphisms from $A$ to $B$.
When there is no chance of confusion we simply write $f:A\to B$ instead of $f\in\catC(A,B)$.
Given two morphisms $f: A\to B$ and $g : B \to C$, their composition is indicated by $g \comp f : A\to C$.

When the category $\catC$ is Cartesian, we denote by $\Termobj$ the \emph{terminal object}, by $A\times B$ the \emph{categorical product} of $A$ and $B$, by $\Proj1: A\times B\to A$, $\Proj2: A\times B\to B$ the associated \emph{projections} and, given a pair of arrows $f:C\to A$ and $g:C\to B$, by $\pairfun{f}{g}: C\to A\times B$ the unique arrow such that $\Proj1 \comp \pairfun{f}{g} =f$ and $\Proj2 \comp \pairfun{f}{g} = g$.
We write $f\times g$ for the \emph{product map of $f$ and $g$} which is defined by 
$f\times g = \Pair{f\comp\Proj 1}{g\comp\Proj 2}$.

When $\catC$ is in addition \emph{Cartesian closed} we write $\Funint{A}{B}$ for the \emph{exponential object} 
and $\eval_{AB}:(\Funint AB)\times A\to B$ for the \emph{evaluation morphism}.
Moreover, for any object $C$ and arrow $f:C\times A\to B$, $\Curry(f):C\to(\Funint AB)$ stands for the (unique) morphism such that $\eval_{AB}\comp(\Curry(f)\times \Id{A}) = f$.

A Cartesian closed category $\catC$ is \emph{well-pointed} if, for all objects $A,B$ and morphisms $f,g: A\to B$, whenever $f\neq g$, there exists a morphism $h : \Termobj \to A$ such that $f\comp h \neq g\comp h$.
Similarly, an object $A$ is \emph{well-pointed} if the property above holds for all $f,g: A \to A$.

We say that $\cD = (D,\appl,\Abs)$ is a \emph{reflexive object (living in $\catC$)} if $D$ is an object of $\catC$ and $\appl: D\to (D\To D),\Abs : (D\To D)\to D$ are morphisms such that $\appl\comp \Abs = \Id{D\To D}$.
A reflexive object $\cD$ is called \emph{extensional} whenever $\Abs\comp \appl = \Id{D}$.


\subsection{The Lambda Calculus}\label{subsec:lamcal} 

We generally use the notation of Barendregt's book \cite{Bare} for \lam-calculus.
The set $\Lambda$ of \emph{\lam-terms} over a denumerable set $\Var$ of variables is 
defined by:
$$
	\Lam :\quad M,N,P,Q\ \bnf\ x\ \mid\ \lam x.M\ \mid\ MN \qquad
	\textrm{for all }x\in\Var.
$$
We assume that the application associates to the left and has a higher precedence than \lam-abstraction.
For instance, we write $\lam xyz.xyz$ for the \lam-term $\lam x.(\lam y.(\lam z.((xy)z)))$.
Moreover, we often write $\seq x$ for the sequence $\seqof{x_1,\dots,x_n}$ and $\lam\seq x.M$ for $\lam x_1\dots \lam x_n.M$.

The set $\FV(M)$ of \emph{free variables} of $M$ and the $\alpha$-conversion are defined as in~\cite[Ch.~1\S2]{Bare}.
Hereafter, we consider \lam-terms up to $\alpha$-conversion. 

\begin{definition} %
A \lam-term~$M$ is \emph{closed} whenever $\FV(M)=\emptyset$ and 
in this case it is also called a \emph{combinator}. 
The set of all combinators is denoted by~$\Lam^o$.
\end{definition}

We often consider relations on \lam-terms that have the property of being ``context closed''.
Intuitively a context $C\hole-$ is a \lam-term with a \emph{hole} denoted by~$\hole-$. 
Formally, the hole is an algebraic variable and contexts are defined as follows.

\begin{definition} \label{def:contexts}\ 
\begin{itemize}
\item A \emph{context} $C\hole-$ is generated by the grammar (for $x\in\Var$):
$$
	C\hole-\ ::= \ \hole-\ \mid \ x \ \mid \ \lam x.C\hole- \ \mid \ (C_1\hole-)(C_2\hole-)
$$
\item 
	A context $C'\hole-$ is called \emph{single hole} if it has a unique occurrence of the algebraic variable~$\hole-$.
	Single hole contexts are generated by (for $M\in\Lam$):
	$$
		C'\hole- \ ::= \ \hole- \ \mid\ \lam x.C'\hole -  \ \mid\ M(C' \hole-) \ \mid\ (C'\hole -) M  
	$$
\item 
	A context $H\hole-$ is a \emph{head context} if it has the shape $(\lam x_1\dots x_n.\hole-)M_1\cdots M_k$ for some $n,k\ge0$ and $M_1,\dots,M_k\in\Lam$.	
\end{itemize}
Given a context $C\hole-$, we write $C\hole M$ for the \lam-term obtained from $C\hole-$ by substituting $M$ for the hole $\hole-$ possibly with capture of free variables in $M$. 
\end{definition}

A relation $\rR\subseteq\Lam\times\Lam$ is \emph{context closed} whenever $M\rel R N$ entails $C\hole M\rel R C\hole N$ for all single hole contexts $C\hole -$. 
The \emph{context closure} of a relation $\red R\subseteq\Lam\times\Lam$ is the smallest context closed relation $\rR'$ containing $\rR$.
\smallskip

{\bf Reductions}. The \lam-calculus is a higher-order term rewriting system and several notions of reduction can be considered. 
As a matter of notation, given a reduction $\red{R}$, we write $\redto[R]$ for its context closure, $\msto[R]$ for the transitive and reflexive closure of $\redto[R]$, finally $=_{\red R}$ for the corresponding \emph{$\red R$-conversion}, that is the transitive, reflexive and symmetric closure of~$\redto[R]$.
We denote by $\nf[R](M)$ the $\red{R}$-normal form (\emph{$\red{R}$-nf}, for short) of $M$ (if it exists) and by $\NF[R]$ the set of all $\red{R}$-normal forms.
Given two reductions $\red R_1,\red R_2$ we denote their union by simple juxtaposition, i.e.\ $\redto[\rR_1\rR_2]$ represents the relation $\redto[\rR_1]\cup\redto[\rR_2]$.

The main notion of reduction is the \emph{$\beta$-reduction}, which is the context closure of:
\[
 (\lam x.M)N \redto M\subst{x}{N}	\tag{$\beta$}
\]
where $M\subst{x}{N}$ denotes the capture-free simultaneous substitution of $N$ for all free occurrences of $x$ in $M$. 
The term on the left hand-side of the arrow is called \emph{$\beta$-redex}, while the term on the right hand-side is its \emph{$\beta$-contractum}.

A \lam-term $M$ is in $\beta$-normal form if and only if $M = \lam\seq x.x_iM_1\cdots M_k$ and each $M_i$ is in $\beta$-normal form. We say that $M$ \emph{has a $\beta$-normal form} whenever $\nf(M)$ exists.

It is well known that the \lam-calculus is an intensional language --- there are $\beta$-different \lam-terms that are extensionally equal. 
This justifies the definition of the $\eta$-reduction:
\[
\tag{$\eta$} \lam x.Mx \redto M\textrm{ provided }x\notin\FV(M).
\]
Notice however that, when $M$ is a \lam-abstraction, the $\eta$-reduction is actually a $\beta$-step.
\smallskip

{\bf Useful combinators and encodings.}\label{combinators}
Concerning specific combinators, we fix the following notations:
$$  \arraycolsep=0.4cm\def\arraystretch{1.2}
	\begin{array}{r@{\:=\:}l r@{\:=\:}l r@{\:=\:}l}
	\bI&\onec_0=\lam x.x		&\bK&\lam xy.x		&\onec_{n+1} & \lam xy.x(\onec_ny)\\
        \bold{F} & \lam xy.y
        &\Delta&\lam x.xx
        &\bY &\lam f.(\lam x.f(xx))(\lam x.f(xx))\\
        \bO & \Delta\Delta
	&\cnc_n & \lam f z. f^n(z)
        &\bJ & \bY(\lam jxy.x(jy))\\
        
	\end{array}
$$
where $f^n(z) = f(\cdots f(f(z))\cdots)$, $n$ times. We will simply denote by $\onec$ the combinator $\onec_1=_\beta\lam xy.xy$.
It is easy to check that $\bI$ is the identity, $\One_n$ is a $\beta\eta$-expansion of the identity, $\bK$ and $\bold{F}$ are respectively the first and second projection, $\bO$ is the paradigmatic looping combinator, $\bY$ is Curry's fixed point combinator, $\cnc_n$ is the $n$-th Church's numeral 
and $\bJ$ is Wadsworth's ``infinite $\eta$-expansion of the identity'' (see Section~\ref{sec:lamtheories}).

Given two \lam-terms $M$ and $N$ their \emph{composition} is defined by $M\comp N = \lam x.M(Nx)$ and their \emph{pairing} by $[M,N] = \lam x.xMN$ (for $x\notin\FV(MN)$).
Moreover, it is possible to \lam-define a given enumeration $(M_n)_{n\in\nat}$ of closed \lam-terms whenever such an enumeration is effective.  

\begin{definition}
An enumeration of closed \lam-terms $(M_n)_{n\in\nat}$ is called \emph{effective} (or \emph{uniform} in~\cite[\S8.2]{Bare}) if there is a combinator $F\in\Lamo$ such that $F \cnc_n =_\beta M_n$ for all $n\in\nat$. 
\end{definition}

As shown in \cite[Def.~8.2.3]{Bare}, when the enumeration is effective, the sequence $[M]_{n\in\nat}$ can be expressed (using the fixed point combinator $\bY$) as a single \lam-term satisfying 
$$
	[M_n]_{n\in\nat} =_\beta [M_0 , [M_{n+1}]_{n\in\nat}] =_\beta [M_0 , [M_1,[M_{n+2}]_{n\in\nat}]]] =_\beta \cdots 
$$
Such infinite sequences will be mainly used in the following sections to build examples.

\smallskip
{\bf Solvability.}
Lambda terms are classified into solvable and unsolvable, depending on their capability of interaction with the environment, represented here by a context.

\begin{definition}\ 
\begin{itemize}
\item A \lam-term $M$ is {\em solvable} if there exists a head context $H\hole-$ such that ${H\hole M =_\beta \bI}$. Otherwise $M$ is called \emph{unsolvable}. 
\item Two \lam-terms $M,N$ are called \emph{separable} if there exists a context $C\hole-$ such that $C\hole M =_\beta x$ and $C\hole N =_\beta y$ for some variables $x\neq y$.
\item Two \lam-terms $M,N$ are called \emph{semi-separable} if there exists a context $C\hole-$ such that $C\hole M$ is solvable while $C\hole N$ is unsolvable or \emph{vice versa}.
\end{itemize}
\end{definition}

Solvability has been characterized by Wadsworth in terms of \emph{head normalization} in~\cite{Wadsworth76}.
We recall that a \lam-term $M$ is in \emph{head normal form} (\emph{hnf}, for short) if it is of the form $M = \lam x_1\dots x_n.x_iM_1\cdots M_k$ for $n,k\ge 0$.
Remark that in our notation the head variable $x_i$ can be either bound or free. 
A \lam-term $M$ \emph{has a hnf} if it is $\beta$-convertible to a hnf.
The \emph{principal hnf} of a \lam-term $M$, denoted $\phnf(M)$, is the hnf obtained from $M$ by \emph{head reduction} $\to_h$, that is by repeatedly contracting the head redex $\lam\seq y.\underline{(\lam x.M)N}\seq P$ in $M$.
We refer to~\cite[Def.~8.3.10]{Bare} for a formal definition.
\smallskip

\begin{theorem}[Wadsworth~\cite{Wadsworth76}]\label{thm:Wads}~\\
A \lam-term $M$ is solvable if and only if $M$ has a head normal form.
\end{theorem}

We say that $M,N$ have \emph{similar hnf's} if $\phnf(M) = \lam x_1\cdots x_n.x_iM_1\cdots M_k$ and $\phnf(N) = \lam x_1\cdots x_{n'}.x_iN_1\cdots N_{k'}$ with $k-n = k'-n'$ and either $x_i$ is free or $i\le\min\{n,n'\}$.

\subsection{B\"ohm Trees}\label{sec:BT}

The B\"ohm trees, introduced by Barendregt in 1977~\cite{Barendregt77} and named after Corrado B\"ohm, are possibly infinite labelled trees representing the execution of a \lam-term.
The following coinductive definition is taken from~\cite{Lassen99} (see also~\cite{JacobsR97}).

\begin{definition}\input{include/fig1}
The \emph{B\"{o}hm tree} $\BT{M}$ of a \lam-term $M$ is defined coinductively as follows:
\begin{itemize}
\item 
	if $M$ is unsolvable then $\BT{M} = \bot$;
\item
	if $M$ is solvable and $\phnf(M)=\lambda x_{1}\ldots x_{n}.x_iM_{1}\cdots M_{k}$ then:
 
\begin{tikzpicture}
\node (root) at (-130pt,8pt) {~};
\node (BT) at (0,0) {$\BT{M}=$};
\node[right of = BT, node distance =2cm] (head) {$\lam x_{1} \ldots x_{n}.x_i$};
\node (BTN1) at ($(head)+(-.25,-.75)$) {$\BT{M_1}$};
\draw ($(head.south east)+(-0.5,0.1)$) -- ($(BTN1.north)-(0.,0)$);
\node (BTNk) at ($(head)+(2.,-.75)$) {$\BT{M_k}$};
\draw ($(head.south east)+(-0.1,0.1)$) -- ($(BTNk.north)-(0.,0)$);
\node at ($(head)+(.9,-.75)$) {$\cdots$};
\end{tikzpicture}
\end{itemize}
\end{definition}

In Figure~\ref{fig:Boehm} we provide some notable examples of B\"ohm trees.
Notice that in general $\FV(\BT{M})\subseteq \FV(M)$ but the converse might not hold. 
Indeed, given a \lam-term $M$ satisfying $M\msto[\beta] \lam zx.x(Mz)$, we have that $z\notin\FV(\BT{Mz})$ because $z$ is ``pushed to infinity''.

We also present, as an auxiliary notion, the ``B\"ohm-like trees'', which are labelled trees that look like B\"ohm trees but may not arise as a B\"ohm tree of a \lam-term.

\begin{definition} A \emph{B\"ohm-like tree} is a labelled tree over $\BTlabs = \setof{\bot} \cup \setof{\lambda \vec x.y \mid  \seq x,y \in \Var}$, that is a function $U : \Seq\to\BTlabs\times \nat$ such that $\pi_2\comp U$ is a tree and $(\pi_1\comp U)(\sigma) = \bot$ entails $(\pi_2\comp U)(\sigma)=0$. 
The tree $\pi_2\comp U$ is called \emph{the underlying tree of $U$} and is denoted by $\nak{U}$.
\end{definition}

A B\"ohm-like tree $U$ is \emph{\lam-definable} if $U = \BT M$ for some $M\in\Lam$.
In~\cite[Thm. 10.1.23]{Bare}, Barendregt gives the following characterization of \lam-definable B\"ohm-like trees.

\begin{theorem}\label{thm:BTMiffeffectiveT}
Given a B\"ohm-like tree $U$, there exists a \lam-term $M$ such that $\BT{M} = U$ if and only if the restriction $\restr U{\dom_\bot(U)}$ where $\dom_\bot(U) = \{ \sigma\in\dom(U)\st (\pi_1\comp U)(\sigma) \neq\bot\}$ is computable (after coding)
and $\FV(U)$ is finite.
\end{theorem}

 The B\"ohm-like tree $\bot$ represents the absolute lack of information, therefore it makes sense to say that $\bot$ is ``less defined'' than any B\"ohm-like tree $U$. 
This is the consideration behind the order $\BTle$ on B\"ohm-like trees defined below.

\begin{definition}\label{def:BTle}
Given two B\"ohm-like trees $U,V$ we say that \emph{$U$ is an approximant of $V$}, written $U\BTle V$, whenever $U$ results from $V$ by replacing some subtrees by $\bot$. 
\end{definition}

{\bf Approximations of B\"ohm trees.}
A B\"ohm tree can be also seen as the least upper bound of its finite approximants, and finite approximants can be seen as the normal forms of a \lam-calculus extended with a constant $\bot$ and an additional reduction $\redto[\bot]$.

A \emph{$\lamb$-term} $M$ is a \lam-term possibly containing occurrences of the constant $\bot$.
The set $\Lamb$ of all \emph{$\lamb$-terms} is generated by the grammar:
$$
	\Lamb :\quad M,N\bnf\ x\ \mid\ \lam x.M\ \mid\ MN \ \mid\ \bot
$$
Similarly a \emph{(single hole) $\lamb$-context }is a (single hole) context $C\hole-$ possibly containing occurrences of $\bot$.
The \emph{$\bot$-reduction} $\redto[\bot]$ is defined as the $\lam\bot$-contextual closure of the rules:
\[
\tag*{($\bot$)} x.\bot\to \bot\qquad\qquad\qquad \bot M \to \bot
\]
The $\beta$- and $\eta$- reductions are extended to \lamb-terms in the obvious way.
We write $\App$ for the set of $\lamb$-terms in $\beta\bot$-normal forms and we denote its elements by $s,t,u,\dots$

The following characterization of $\beta\bot$-normal forms is well known.

\begin{lemma}
Let $M\in\Lamb$. We have $M\in\App$ if and only if either $M=\bot$ or $M$ has shape $\lam x_1\dots x_n.x_iM_1\cdots M_k$ (for some $n,k\ge 0$) and each $M_j$ is $\beta\bot$-normal.
\end{lemma}

\begin{definition}\label{def:sizet}
 The \emph{size} of $t\in\App$, written $\size t$, is defined by induction:
$$
	\size\bot = 0,\qquad
	\size(\lam x.t) = 1 + \size t,\qquad
	\size(xt_1\cdots t_k) = 1 + \textstyle\sum_{i=1}^k \size t_i.
$$
\end{definition}

The preorder $\BTle$ is defined on \lamb-terms as the $\lam\bot$-contextual closure of $\bot \le M$. 
It is easy to check that, for all $t\in\App$, this definition and Definition~\ref{def:BTle}  coincide.

The set of all finite approximants of the B\"ohm tree of $M$ can be obtained by calculating the direct approximants of all \lam-terms $\beta$-convertible with $M$.

\begin{definition}\label{def:omegastuff} Let $M\in\Lamb$. 
\begin{enumerate}
\item\label{def:omegastuff1}
	The \emph{direct approximant of $M$}, written $\da{M}$, is the $\lamb$-term defined as:
	\begin{itemize}
	\item $\da{M} = \bot$ if $M = \lam x_1\dots x_k.\bot M_1\cdots M_k$,
	\item $\da{M} = \bot$ if $M = \lam x_1\dots x_k.(\lam y.M')NM_1\cdots M_k$,
	\item $\da{M} = \lam x_1\dots x_n.x_i\da{M_1}\cdots \da{M_k}$ if $M = \lam x_1\dots x_n.x_iM_1\cdots M_k$,
	\end{itemize}
\item\label{def:omegastuff2} 
	The \emph{set of finite approximants of $M$} is defined by:
	$$
			\BT[*]{M} = \big\{\da{M'} \st M =_\beta M'\big\}.
	$$
\end{enumerate}
\end{definition}
As shown in~\cite[\S2.3]{AmadioC98}, the set $\BT[*]M$ is directed with respect to $\BTle$.
For $M\in\Lam$, a finite approximant $t\in\App$ belongs to $\BT[*]M$ exactly when $t\BTle \BT{M}$, therefore $ \BT{M} = \bigvee \BT[*]M$. 
Moreover, the following syntactic continuity property holds.


\begin{lemma}\label{lemma:xDomenico} Let $M\in\Lamb$ and let $C\hole-$ be a \lamb-context. If $M = C[(\lam x.P)Q]$ then $\BT[*]{C[\bot]}\subseteq\BT[*]{M}$.
\end{lemma}

\begin{proof} Assume that $t\in\BT[*]{C[\bot]}$. By definition, there exists $C'[-]$ such that $C[y] =_\beta C'[y]$, for some fresh variable $y$, and $t = \da{C'[\bot]}$. Then $C[(\lam x.P)Q] =_\beta C'[(\lam x.P)Q]$ and $t = \da{C'[(\lam x.P)Q]}$, which implies that $t\in\BT[*]{M}$.
\end{proof}

\subsection{Inequational and Lambda Theories} 

Inequational theories and \lam-theories become the main object of study when considering the computational equivalence more important than the process of computation.


\begin{definition}
An \emph{inequational theory} is any context closed preorder on $\Lam$ containing the $\beta$-conversion.
A {\em \lam-theory} is any context closed equivalence on $\Lam$ containing the $\beta$-conversion. 
\end{definition}

Given an inequational theory $\cT$ we write $M\sqle_\cT N$ or $\cT\vdash M\sqle N$ for~$(M,N)\in\cT$.
Similarly, given a \lam-theory $\cT$, we write $M=_\cT N$ or $\cT\vdash M=N$ whenever $(M,N)\in\cT$.

The set of all \lam-theories, ordered by set theoretical inclusion, forms a complete lattice $\LLT$ that has a rich mathematical structure, as shown by Salibra and his coauthors in their works~\cite{Salibra01,LusinS04,ManzonettoS06}.

\begin{definition}
A \lam-theory (or an inequational theory) is called:
\begin{itemize}
\item
{\em consistent} if it does not equate all \lam-terms; 
\item
{\em extensional} if it contains the $\eta$-conversion; 
\item
{\em sensible} if it equates all unsolvables.
\end{itemize}
\end{definition}

We denote by 
$\blam$ the least \lam-theory, 
by $\blam\eta$ the least extensional \lam-theory,
by $\cH$ the least sensible \lam-theory, and 
by $\BTth$ the (sensible) \lam-theory equating all \lam-terms having the same B\"ohm tree.
Given a \lam-theory $\cT$, we write $\cT\eta$ for the least \lam-theory containing $\cT\cup\blam\eta$.

Inequational theories are less ubiquitously studied in the literature, except when they capture some observational preorder as explained in the next section.
However, they have been studied in full generality in connection with denotational models (see, e.g., \cite{BerlineMS09}).


%% file: include/fig1.tex
\begin{figure}[t!]
\begin{tikzpicture}
\node (myspot) at (-5.9,0) {~};
\node (BTxyo) at (-6.5,3.5) {$\BT{\lam x.y\Omega}$};
\node[below of = BTxyo, node distance=8pt] (eq) {$\shortparallel$};
\node[below of = eq, node distance=10pt] (l1) {$\lam x.y$};
\node[below of = l1, node distance=18pt,xshift=6pt] (l2) {$\bot$};
\draw ($(l1.south)+(6pt,0.05)$) -- ($(l2.north)-(0,0.05)$);
\node (BTJ) at (-4.5,3.5) {$\BT{\One_3}$};
\node[below of = BTJ, node distance=8pt] (eq) {$\shortparallel$};
\node[below of = eq, node distance=10pt] (l1) {$\lam xz_0.x~~$};
\node[below of = l1, node distance=18pt] (l2) {$\lam z_1.z_0$};
\node[below of = l2, node distance=18pt] (l3) {$\lam z_2.z_1$};
\node[below of = l3, node distance=18pt] (l4) {$\phantom{\lam z_3.}z_2$};
\draw ($(l1.south)+(8pt,0.05)$) -- ($(l2.north)+(8pt,-0.15)$);
\draw ($(l2.south)+(8pt,0.05)$) -- ($(l3.north)+(8pt,-0.15)$);
\draw ($(l3.south)+(8pt,0.05)$) -- ($(l4.north)+(8pt,-0.15)$);
\node (BTJ) at (-2.75,3.5) {$\BT{\bJ}$};
\node[below of = BTJ, node distance=8pt] (eq) {$\shortparallel$};
\node[below of = eq, node distance=10pt] (l1) {$\lam xz_0.x~~$};
\node[below of = l1, node distance=18pt] (l2) {$\lam z_1.z_0$};
\node[below of = l2, node distance=18pt] (l3) {$\lam z_2.z_1$};
\node[below of = l3, node distance=18pt] (l4) {$\lam z_3.z_2$};
\node[below of = l4, xshift=8pt, node distance=10pt] (l5) {$\vdots$};
\draw ($(l1.south)+(8pt,0.05)$) -- ($(l2.north)+(8pt,-0.15)$);
\draw ($(l2.south)+(8pt,0.05)$) -- ($(l3.north)+(8pt,-0.15)$);
\draw ($(l3.south)+(8pt,0.05)$) -- ($(l4.north)+(8pt,-0.15)$);
\node (BTOf) at (-1,3.5) {$\BT{\bY}$};
\node[below of = BTOf, node distance=8pt] (eq) {$\shortparallel$};
\node[below of = eq, node distance=10pt] (l1) {$\lam f.f$};
\node[below of = l1, xshift=6pt, node distance=18pt] (l2) {$f$};
\node[below of = l2, node distance=18pt] (l3) {$f$};
\node[below of = l3, node distance=18pt] (l4) {$f$};
\node[below of = l4, node distance=10pt] (l5) {$\vdots$};
\draw ($(l1.south)+(6pt,0.05)$) -- ($(l2.north)+(0pt,-0.05)$);
\draw ($(l2.south)+(0,0.05)$) -- ($(l3.north)+(0pt,-0.05)$);
\draw ($(l3.south)+(0,0.05)$) -- ($(l4.north)+(0pt,-0.05)$);
\node (BTP) at (1.6,3.5) {$\BT{[M]_{n\in\nat}}$};
\node[below of = BTP, node distance=10pt] (eq) {$\shortparallel$};
\node (root) at ($(BTP)+(0,-20pt)$) {};
\node (BTM) at (root) {$\lam y.y$};
\node (M0) at ($(root)+(-10pt,-22pt)$) {$\BT{M_0}$};
\node (x1) at ($(root)+(50pt,-22pt)$) {$\lam y.y$};
\node (M1) at ($(x1)+(-13pt,-22pt)$) {$\BT{M_1}$};
\node (x2) at ($(x1)+(50pt,-22pt)$) {$\lam y.y$};
\node (M2) at ($(x2)+(-13pt,-22pt)$) {$\BT{M_2}$};
\draw ($(BTM.south east)+(-4pt,3pt)$) -- ($(x1.north)+(0pt,-2pt)$);
\draw ($(BTM.south east)+(-10pt,3pt)$) -- (M0);
\draw ($(x1.south east)+(-2pt,2pt)$) -- ($(x2.north)+(0pt,-2pt)$);
\draw ($(x1.south east)+(-10pt,2pt)$) -- (M1);
\draw ($(x2.south east)+(-10pt,2pt)$) -- (M2);
\draw[densely dotted] ($(x2.south east)+(-2pt,2pt)$) -- ($(x2.south east)+(30pt,-7pt)$);
\node (BTM) at (6,3.5) {$\BT{\Omega} = \bot$};

\end{tikzpicture}
\caption{Some examples of B\"ohm trees.}\label{fig:Boehm}
\end{figure}

%% file: include/theories.tex
Several interesting \lam-theories are obtained via suitable observational preorders defined with respect to a set $\cO$ of \emph{observables}. 
This has been first done by Morris's in his PhD thesis~\cite{Morristh}.

\begin{definition}
Let $\cO$ be a non-empty subset of $\Lam$.
\begin{itemize}
\item The \emph{$\cO$-observational preorder} is given by:
$$
	M\obsle[\cO] N\iff \forall C\hole-\ .\ C\hole{M}\in_\beta\cO\textrm{ entails } C\hole{N}\in_\beta\cO.
$$
where $M\in_\beta \cO$ means that there exists a \lam-term $M'$ such that $M\msto[\beta] M'\in\cO$.
\item The \emph{$\cO$-observational equivalence} $M\obseq[\cO] N$ is defined as $ M\obsle[\cO] N$ and $N\obsle[\cO] M$.
\end{itemize}
\end{definition}

In the rest of the section we will discuss the \lam-theories $\Hst$ and $\Hpl$ generated  as observational equivalences by considering as observables the head normal forms and the $\beta$-normal forms, respectively, and the corresponding preorders\footnote{When considering these particular sets $\cO$ of observables it is not difficult to check that the relations $\obsle[\cO]$ and $\obseq[\cO]$ are actually inequational and \lam-theories (cf.~\cite[Prop.~16.4.6]{Bare}). The general case is treated in \cite{Paolini08}.}.
In both cases we also recall the characterizations given in terms of extensional equivalences on B\"ohm-trees.

\subsection{$\Hst$: B\"ohm Trees and Infinite $\eta$-Expansions}
The \lam-theory $\Hst$ has been defined by Wadsworth and Hyland as an observational equivalence in~\cite{Wadsworth76,Hyland76}, where they proved that it corresponds to the equational theory induced by Scott's  model $\mathscr{D}_\infty$.
In the years, $\Hst$ has become the most well studied \lam-theory~\cite{Bare,GouyTh,GianantonioFH99,RonchiP04,Manzonetto09,Breuvart14}.

\begin{definition}\label{def:Hst}
We let $\sqle_\Hst$ be the $\cO$-observational preorder obtained by taking as $\cO$ the set of head normal forms and $\Hst$ be the corresponding equivalence. 
\end{definition}

Notice that $M=_\Hst N$ is equivalent to say that $M,N$ are not semi-separable.
It is easy to check that $\Hst$ is an extensional \lam-theory.
A first characterization of $\Hst$ can be given in terms of maximal consistent extension (also known as \emph{Post-completion}) of~$\cH$, and such a maximality property extends to the corresponding inequational theory. 

The following lemma is a generalization of~\cite[Thm.~16.2.6]{Bare} that encompasses the inequational case.

\begin{lemma}\label{lemma:maximality}
The preorder $\sqle_\Hst$ and the equivalence $\Hst$ are maximal among consistent sensible inequational theories and \lam-theories, respectively.
\end{lemma}
\begin{proof} Let $\sqle_\cT$ be an inequational theory such that $\sqle_\Hst\subsetneq\ \sqle_\cT$.
This means that there exist $M,N$ such that $M\sqle_\cT N$ while $M\not\sqle_\Hst N$.
By Definition~\ref{def:Hst}, there is a context $C\hole-$ such that $C\hole M$ has an hnf while $C\hole N$ does not. 
By Theorem~\ref{thm:Wads}, there is a head context $H\hole- = (\lam \seq x.\hole-)\seq P$ such that $H\hole{C\hole M} =_\beta \bI$ and $H\hole{C\hole N}$ is unsolvable. 
As $M \sqle_\cT N$ we~get:
$$
	\bI =_\beta H\hole{C\hole M} \sqle_\cT H\hole{C\hole N} =_\cT \bO_3\qquad \textrm{ for $\bO_3 = \Delta_3\Delta_3$ and $\Delta_3 = \lam x.xxx\;.$ }
$$
The rightmost equality $=_\cT$ holds because $\bO_3$ is unsolvable and the inequational theory $\sqle_\cT$ is sensible.
Since $\bO_3\sqle_\Hst \bI$ and $\sqle_\Hst\subseteq\ \sqle_\cT$, we obtain $\bO_3 =_\cT \bI$ which leads to: 
$$
	\bI =_\cT\bO_3 =_\beta \bO_3 \Delta_3 =_\cT\bI \Delta_3 =_\beta 	\Delta_3  \,.
$$
Since $\bI$ and $\Delta_3$ are $\beta\eta$-distinct normal form, this contradicts the B\"ohm Theorem~\cite{Bohm68}.
\end{proof}

The characterization of $\Hst$ in terms of trees requires the notion of ``infinite $\eta$-expansion'' of a B\"ohm-like tree.
Intuitively, the B\"ohm-like tree $V$ is an infinite $\eta$-expansion of $U$, if it is obtained from $U$ by performing countably many possibly infinite $\eta$-expansions.

The classic definition is given in terms of tree extensions~\cite[Def.~10.2.10]{Bare}; here we rather follow the coinductive approach introduced in~\cite{Lassen99}. 

\begin{definition}\label{def:leeta}
Given two B\"ohm-like trees $U$ and $V$, we define coinductively the relation $U\leeta[\infty] V$ expressing the fact that $V$ is a (possibly) \emph{infinite $\eta$-expansion} of $U$.
We let $\leeta[\infty]$ be the greatest relation between B\"ohm-like trees such that $U\leeta[\infty] V$ entails that
\begin{itemize} 
\item either $U = V=\bot$, 
\item or (for some $i,k,m,n\ge 0$):
$$
U = \lam x_1\dots x_n.x_iU_1\cdots U_k\ \textrm{ and }\
V =\lam x_1\dots x_nz_1\dots z_m.x_iV_1\cdots V_k V'_1\cdots V'_{m}
$$
where $\seq z\cap\FV(x_iU_1\cdots U_k)=\emptyset$, $U_j \leeta[\infty] V_j$ for all $j\le k$ and $z_\ell\leeta[\infty] V'_\ell$ for all $\ell\le m$.
\end{itemize}
\end{definition}
%
Notice that in Barendregt's book~\cite[Def.~10.2.10(iii)]{Bare}, the relation above is denoted by $\le_\eta$.
We prefer to use a different notation because we want to emphasize the possibly infinitary nature of such $\eta$-expansions.\input{include/fig2}
\begin{thmC}[{\cite[Thm.~19.2.9]{Bare}}]\label{thm:Hst_as_etainf}
Let $M,N\in\Lambda$.
\begin{enumerate}[(i)]
\item $M\sqle_ \Hst N$ if and only if there are B\"ohm-like trees $U,V$ such that 
	$\BT{M}\leeta[\infty] U \BTle V\geeta[\infty]\BT{N}$.
\item $M =_{\Hst} N$ if and only if there is a B\"ohm-like tree $U$ such that $\BT{M}\leeta[\infty] U\geeta[\infty]\BT{N}$.
\end{enumerate}
\end{thmC}

In other words, $\Hst$ equates all \lam-terms whose B\"ohm trees are equal up to countably many (possibly) infinite $\eta$-expansions.
From Theorem~\ref{thm:BTMiffeffectiveT}, it follows that the trees $U,V$ appearing in the statements above can always be chosen \lam-definable (see~\cite[Ex.~10.6.7]{Bare}).

\begin{example}\label{ex:Hst}\ 
\begin{enumerate}
\item
	The typical example is $\bI =_\Hst \bold{J}$, since clearly $\BT{\bJ}$ is an infinite $\eta$-expansion of $\bI$, so $\BT{\bI}\leeta[\infty] \BT{\bJ}$ holds.
\item\label{ex:Hst2}
	As a consequence, we get that $\BT{[\bI]_{n\in\nat}}\leeta[\infty] \BT{[\bJ]_{n\in\nat}}$ for $[\bI]_{n\in\nat} = [\bI,[\bI,[\bI,\dots ]]]$ and $[\bJ]_{n\in\nat} = [\bJ,[\bJ,[\bJ,\dots ]]]$.
\item 
	For $M = \lam xy.xx\Omega(\bold{J}x)$ and $N = \lam xy.x(\lam z_0w_0.x(\bJ z_0)(\bJ w_0))y x$, we have $M\sqle_\Hst N$ (as shown in Figure~\ref{fig:BTle_eta}) but $M \neq_\Hst N$. 
	As we will see in Section~\ref{ssec:J_T}, it is possible to represent the subterm $\lam z_0w_0.x(\bJ z_0)(\bJ w_0)$ as $\bJ_Tx$, where $\bJ_T$ is an infinite $\eta$-expansion of $\bI$ following a suitable tree $T$.
\end{enumerate}
\end{example}
The point $(2)$ shows that for proving $\BT{M}\leeta[\infty] \BT{N}$ one may need to perform denumerably many infinite $\eta$-expansions.
Point (3) shows that for proving $M \sqle_\Hst N$, it may not be enough to infinitely $\eta$-expand $\BT{M}$ to match the structure of $\BT N$: 
one may need to perform infinite $\eta$-expansions on both sides and cut some subtrees of $\BT N$.

\begin{remark}\label{rem:simh} From Theorem~\ref{thm:Hst_as_etainf} and Definition~\ref{def:leeta} it follows that if $M\sqle_\Hst N$ and $M$ is solvable then also $N$ is solvable and $M,N$ have similar hnf's.
\end{remark}

\subsection{The Infinite $\eta$-Expansion $\bold{J}_T$.}\label{ssec:J_T}
Wadsworth's combinator $\bJ$ is the typical example of an infinite $\eta$-expansion of the identity.
However, there are many \lam-terms $M$ that satisfy the property of being infinite $\eta$-expansions of the identity, namely $\BT{\bI}\leeta[\infty] \BT{M}$.

Recall from Section~\ref{sec:BT} that $\nak{\BT{M}}$ denotes the underlying tree of $\BT{M}$.
\input{include/fig3}

\begin{definition} Let $M\in\Lamo$ and $T\in\Trees[]$.
We say that $M$ is an \emph{infinite $\eta$-expansion of the identity following $T$} whenever $x\leeta[\infty] \BT{Mx}$ and $\nak{\BT{Mx}} = T$ for any $x\in\Var$.
\end{definition}

For instance, $\bJ$ follows the infinite unary tree $T$ corresponding to the map $T(\sigma) = 1$ for all $\sigma = \seqof{0,\dots,0}$.
We now provide a characterization of all such infinite $\eta$-expansions.

\begin{proposition} \label{prop:J_T}
For all $T\in\Trees[]$, there exists a \lam-term $\bJ_T$ which is an infinite $\eta$-expansion of the identity following $T$ if and only if $T$ is recursive.
\end{proposition}

\begin{proof}
$(\Rightarrow)$ By~Theorem~\ref{thm:BTMiffeffectiveT}, $\BT{\bJ_Tx}$ is partial recursive and so is its underlying tree $T$. Since $x\leeta[\infty] \BT{\bJ_Tx}$, $\BT{\bJ_Tx}$ cannot have any occurrences of $\bot$ so $\dom(T)$ is decidable.

$(\Leftarrow)$
We fix a bijective encoding of all finite sequences of natural numbers $\# : \Seq\to\nat$ which is \emph{effective} in the sense that the code $\#(\sigma.n)$ is computable from $\#\sigma$ and~$n$.
We write $\code{\sigma}$ for the corresponding Church numeral $\cnc_{\#\sigma}$.
Using a fixed point combinator $\bY$, we define a \lam-term $X\in\Lamo$ satisfying the following recursive equation (for all $\sigma\in\dom(T)$):
\begin{equation}\label{eq:receq}
	X \code{\sigma}x =_\beta \lam z_1\dots z_m.x(X\code{\sigma.0}z_1)\cdots(X\code{\sigma.m-1}z_m)\textrm{ where }m = T(\sigma).
\end{equation}
(The existence of such a \lam-term follows from the fact that $T$ is recursive, the effectiveness of the encoding $\#$ and Church's Thesis.) 
We prove by coinduction that for all $\sigma\in\dom(T)$, $X\code{\sigma}$ is an infinite $\eta$-expansion of the identity following $\restr{T}{\sigma}$. 
Indeed, $X \code{\sigma}x$ is $\beta$-convertible to the \lam-term of Equation~\ref{eq:receq}. 
By coinductive hypothesis we get for all $i < T(\sigma)$ that $z_i \leeta[\infty] \BT{X\code{\sigma.i}z_i}$ and $\nak{\BT{X\code{\sigma.i}z_i}} = \restr{T}{\sigma.i}$.
From this, we conclude that $x\leeta[\infty] \BT{X\code{\sigma}x}$ and $\nak{\BT{X\code{\sigma}x}}= T$.
Therefore, the \lam-term $\bJ_T$ we are looking for is $X\code{\emptyseq}$. 
\end{proof}


\subsection{$\Hpl$: B\"ohm Trees and Their Finitary $\eta$-Expansions}
Perhaps surprisingly, it turns out that $\Hst$ is not the first \lam-theory that has been defined in terms of contextual equivalence. 
Indeed, Morris's original extensional observational equivalence is the following.

\begin{definition}\label{def:Hpl}
\emph{Morris's inequational theory} $\sqle_\Hpl$ is the $\cO$-observational preorder obtained by taking as $\cO$ the set $\NF[\beta]$ of $\beta$-normal forms. 
We denote by $\Hpl$ the corresponding equivalence, which we call \emph{Morris's \lam-theory}\footnote{
The notation~$\Hpl$ for Morris's \lam-theory has been used in~\cite{ManzonettoR14,BreuvartMPR16,IntrigilaMP17}. 
The same \lam-theory is denoted $\mathscr{T}_\textrm{NF}$ in Barendregt's book~\cite{Bare} and $\mathbf{N}$ in Paolini and Ronchi della Rocca's one~\cite{RonchiP04}.
}.
\end{definition}

Notice that it is equivalent to take as observables the $\beta\eta$-normal forms, since $M\msto[\beta]\nf[\beta](M)$ exactly when $M\msto[\beta\eta]\nf[\beta\eta](M)$. 
From this, it follows that $\Hpl$ is an extensional \lam-theory.
It is easy to show that $M\sqle_\Hpl N$ entails $M\sqle_\Hst N$, therefore we have $\Hpl\subsetneq\Hst$.
In~\cite{CoppoDZ87}, Coppo, Dezani-Ciancaglini and Zacchi defined a filter model having $\Hpl$ as equational theory.
Also $\Hpl$ can be characterized via a suitable extensional equivalence between B\"ohm trees.
Intuitively, the B\"ohm-like tree $U$ is a finitary $\eta$-expansion of $V$, if it is obtained from $V$ by performing countably many finite $\eta$-expansions.

\begin{definition}\label{def:leetafin}
Given two B\"ohm-like trees $U$ and $V$, we define coinductively the relation  $U\leeta V$ expressing the fact that $V$ is a \emph{finitary $\eta$-expansion} of $U$.
We let $\leeta$ be the greatest relation between B\"ohm-like trees such that $U\leeta V$ entails that
\begin{itemize} 
\item either $U = V=\bot$, 
\item or (for some $i,k,m,n\ge 0$):
$$
U = \lam x_1\dots x_n.x_iU_1\cdots U_k\ \textrm{ and }\
V =\lam x_1\dots x_nz_1\dots z_m.x_iV_1\cdots V_k Q_1\cdots Q_{m}
$$
where $\seq z\cap\FV(x_iU_1\cdots U_k)=\emptyset$, $U_j \leeta V_j$ for all $j\le k$ and $\seq Q\in\NF[\beta]$ are such that $Q_\ell \msto[\eta]z_\ell$ for all $\ell\le m$.
\end{itemize}
\end{definition}

Two $\lam$-terms $M,N$ are equivalent in $\Hpl$ exactly when their B\"ohm trees are equal up to countably many $\eta$-expansions of finite depth.

\begin{thmC}[{\cite[Thm.~2.6]{Hyland75errato}}]\label{thm:informal1}
For $M,N\in\Lambda$, we have that $M =_{\Hpl} N$ if and only if there exists a B\"ohm-like tree $U$ such that $\BT{M}\leeta U\geeta \BT{N}$.
\end{thmC}

\begin{example} Recall from Example~\ref{ex:Hst}(\ref{ex:Hst2}) that $[\bI]_{n\in\nat} = [\bI,[\bI,[\bI,\dots ]]]$ and define the sequence $[\onec_n]_{n\in\nat} = [\bI,[\onec_1,[\onec_2,\dots ]]]$ where $\onec_n$ is defined on Page~\pageref{combinators}. 
From Definition~\ref{def:leetafin} it follows $\BT{[\bI]_{n\in\nat}}\leeta[] \BT{[\onec_n]_{n\in\nat}}$ while, for instance, $\BT{[\bI]_{n\in\nat}}\not\leeta[] \BT{[\bJ]_{n\in\nat}}$.
\end{example}

As a brief digression, notice that the \lam-terms $[\bI]_{n\in\nat}$ and $[\onec_n]_{n\in\nat}$ can be used to show that $\BTth\eta\subsetneq\Hpl$. Indeed, for $M,N\in\Lam$, $M\redto[\eta] N$ entails that $\BT{M}$ can be obtained from $\BT{N}$ by performing at most \emph{one} $\eta$-expansion at every position (see~\cite[Lemma~16.4.3]{Bare}). 
However, to equate the B\"ohm trees of $[\bI]_{n\in\nat}$ and $[\onec_n]_{n\in\nat}$,  at every level one needs to perform $\eta$-expansions of increasing depth and this is impossible in $\BTth\eta$ (as shown in~\cite{IntrigilaMP17}).

As proved in~\cite{BreuvartMPR16} by exploiting a revised B\"ohm-out technique, the following weak separation result holds. 
(For the interested reader a fully detailed proof will appear in~\cite{IntrigilaMP18}.)

\begin{theorem}[Morris Separation]\label{thm:newsep} 
Let $M,N\in\Lambda$ such that $M \sqle_\Hst N$ while $M \not\sqle_\Hpl N$.
There exists a context $C\hole-$ such that $C\hole{M} =_{\beta\eta} \bI$ and $C\hole{N}=_{\BTth} \bold{J}_T$ for some $T\in\Trees$.
\end{theorem}
This allows to Morris-separate also \lam-terms like $\bI$ and $\bJ$ that are not semi-separable.

\subsection{Extensional Approximants}
As far as we know, in the literature there is no characterization of $\sqle_\Hpl$ in terms of extensional equality between B\"ohm trees.
However, L\'evy in~\cite{JJ} provides a characterization in terms of ``extensional approximants'' of B\"ohm trees. Recall from Definition~\ref{def:omegastuff}(\ref{def:omegastuff2}) that $\BT[*]{M}$ is the set of all finite approximants of $M$.

\begin{definition}
For $M\in\Lam$, the set $\BT[e]{M}$ of all \emph{extensional (finite) approximants of $M$} is defined as follows:
$$
	\BT[e]{M} = \{\nf[\eta](t)\st t\in\BT[*]{M'}, M'\msto[\eta]M\}.
$$
\end{definition}

\begin{example} The sets of extensional approximants of some notable \lam-terms:
\begin{itemize}
\item $\BT[e]{\bI} = \{\bot,\bI,\lam xz_0.x\bot,\lam xz_0.x(\lam z_1.z_0\bot),\lam xz_0.x(\lam z_1.z_0(\lam z_2.z_1\bot)),\dots\}$,
\item $\BT[e]{\bJ} = \BT[e]{\bI} -\{\bI\}$.
\end{itemize}
(Here we decided to display those approximants having a regular shape, but also \lamb-terms like $\lam xz_0z_1z_2.x(\lam w_0w_1.z_0w_0\bot)z_1\bot$ belong to these sets).
\end{example}

The following result is taken from~\cite{JJ} (see also Theorem~11.2.20 in~\cite{RonchiP04}) and will be used in the proof of Corollary~\ref{cor:extensionality}.

\begin{theorem}\label{thm:JJ}
 For $M,N\in\Lam$ we have $M\sqle_\Hpl N$ if and only if $\BT[e]{M}\subseteq\BT[e]{N}$.
\end{theorem}


%% file: include/fig2.tex
\begin{figure}[t!]
\begin{tikzpicture}
\node (root) at (0,0) {};
\node (Mh1) at (-40pt,-4pt) {$\lam xy.x$};
\node at ($(Mh1)+(10pt,18pt)$) {$\BT{M}$};
\node at ($(Mh1)+(10pt,8.5pt)$) {$\shortparallel$};
\node (Mh21) at ($(Mh1)+(-10pt,-20pt)$) {$x$};
\node (Mh22) at ($(Mh1)+(10pt,-20pt)$) {$\bot$};
\node (Mh23) at ($(Mh1)+(35pt,-20pt)$) {$\lam z_0.x$};
\node (Mh33) at ($(Mh23)+(0pt,-20pt)$) {$\lam z_1.z_0$};
\node (Mh43) at ($(Mh33)+(0pt,-20pt)$) {$\lam z_2.z_1$};
\node (Mh53) at ($(Mh43)+(7pt,-20pt)$) {$\vdots$};
\draw ($(Mh1)+(6pt,-4pt)$) -- (Mh21);
\draw ($(Mh1)+(10pt,-4pt)$) -- (Mh22);
\draw ($(Mh1)+(14pt,-4pt)$) -- ($(Mh23.north)+(0,-2pt)$);
\draw ($(Mh23)+(7pt,-5pt)$) -- ($(Mh33)+(7pt,4pt)$);
\draw ($(Mh33)+(7pt,-5pt)$) -- ($(Mh43)+(7pt,4pt)$);
\draw ($(Mh43)+(7pt,-5pt)$) -- ($(Mh53)+(0pt,4pt)$);
\node (mmsto) at ($(Mh1)+(60pt,0)$) {$\leeta[\infty]$};
\node (M'h1) at (74pt,-4pt) {$\lam xy.x$};
\node (M'h21) at ($(M'h1)+(-24pt,-20pt)$) {$\lam z_0w_0.x$};
\node (M'h22) at ($(M'h1)+(10pt,-20pt)$) {$\bot$};
\node (M'h23) at ($(M'h1)+(42pt,-20pt)$) {$\lam z_0.x$};
\node (M'h31) at ($(M'h21)+(-15pt,-20pt)$) {$\lam z_1.z_0$};
\node (M'h33) at ($(M'h23)+(0pt,-20pt)$) {$\lam z_1.z_0$};
\node (M'h41) at ($(M'h31)+(0pt,-20pt)$) {$\lam z_2.z_1$};
\node (M'h43) at ($(M'h33)+(0pt,-20pt)$) {$\lam z_2.z_1$};
\node (M'h51) at ($(M'h41)+(7pt,-20pt)$) {$\vdots$};
\draw ($(M'h1)+(6pt,-4pt)$) -- ($(M'h21.north)+(0.6,-2pt)$);
\draw ($(M'h1)+(10pt,-4pt)$) -- (M'h22);
\draw ($(M'h1)+(14pt,-4pt)$) -- (M'h23);
\draw ($(M'h21)+(12pt,-5pt)$) -- ($(M'h31)+(7pt,4pt)$);
\draw ($(M'h31)+(7pt,-5pt)$) -- ($(M'h41)+(7pt,4pt)$);
\draw ($(M'h41)+(7pt,-5pt)$) -- ($(M'h51)+(0pt,4pt)$);
\node (Ph31) at ($(M'h21)+(30pt,-20pt)$) {$\lam w_1.w_0$};
\node (Ph41) at ($(Ph31)+(0pt,-20pt)$) {$\lam w_2.w_1$};
\node (Ph51) at ($(Ph41)+(7pt,-20pt)$) {$\vdots$};
\draw ($(M'h21)+(16pt,-5pt)$) -- ($(Ph31)+(7pt,4pt)$);
\draw ($(Ph31)+(7pt,-5pt)$) -- ($(Ph41)+(7pt,4pt)$);
\draw ($(Ph41)+(7pt,-5pt)$) -- ($(Ph51)+(0pt,4pt)$);
\draw ($(M'h33)+(7pt,-5pt)$) -- ($(M'h43)+(7pt,4pt)$);
\draw ($(M'h23)+(7pt,-5pt)$) -- ($(M'h33)+(7pt,4pt)$);
\node (M'h53) at ($(M'h43)+(7pt,-20pt)$) {$\vdots$};
\draw ($(M'h43)+(7pt,-5pt)$) -- ($(M'h53)+(0pt,4pt)$);
\node at ($(M'h1)+(70pt,0)$) {$\BTle$};
\node (N'h1) at (200pt,-4pt) {$\lam xy.x$};
\node (N'h21) at ($(N'h1)+(-24pt,-20pt)$) {$\lam z_0w_0.x$};
\node (N'h22) at ($(N'h1)+(10pt,-20pt)$) {$y$};
\node (N'h23) at ($(N'h1)+(42pt,-20pt)$) {$\lam z_0.x$};
\node (N'h31) at ($(N'h21)+(-15pt,-20pt)$) {$\lam z_1.z_0$};
\node (N'h33) at ($(N'h23)+(0pt,-20pt)$) {$\lam z_1.z_0$};
\node (N'h41) at ($(N'h31)+(0pt,-20pt)$) {$\lam z_2.z_1$};
\node (N'h43) at ($(N'h33)+(0pt,-20pt)$) {$\lam z_2.z_1$};
\node (N'h51) at ($(N'h41)+(7pt,-20pt)$) {$\vdots$};
\draw ($(N'h1)+(6pt,-4pt)$) -- ($(N'h21.north)+(0.6,-2pt)$);
\draw ($(N'h1)+(10pt,-4pt)$) -- (N'h22);
\draw ($(N'h1)+(14pt,-4pt)$) -- (N'h23);
\draw ($(N'h21)+(12pt,-5pt)$) -- ($(N'h31)+(7pt,4pt)$);
\draw ($(N'h31)+(7pt,-5pt)$) -- ($(N'h41)+(7pt,4pt)$);
\draw ($(N'h41)+(7pt,-5pt)$) -- ($(N'h51)+(0pt,4pt)$);
\node (Ph31) at ($(N'h21)+(30pt,-20pt)$) {$\lam w_1.w_0$};
\node (Ph41) at ($(Ph31)+(0pt,-20pt)$) {$\lam w_2.w_1$};
\node (Ph51) at ($(Ph41)+(7pt,-20pt)$) {$\vdots$};
\draw ($(N'h21)+(16pt,-5pt)$) -- ($(Ph31)+(7pt,4pt)$);
\draw ($(Ph31)+(7pt,-5pt)$) -- ($(Ph41)+(7pt,4pt)$);
\draw ($(Ph41)+(7pt,-5pt)$) -- ($(Ph51)+(0pt,4pt)$);
\draw ($(N'h33)+(7pt,-5pt)$) -- ($(N'h43)+(7pt,4pt)$);
\draw ($(N'h23)+(7pt,-5pt)$) -- ($(N'h33)+(7pt,4pt)$);
\node (N'h53) at ($(N'h43)+(7pt,-20pt)$) {$\vdots$};
\draw ($(N'h43)+(7pt,-5pt)$) -- ($(N'h53)+(0pt,4pt)$);
\node at ($(N'h1)+(70pt,0)$) {$\geeta[\infty]$};
\node (Nh1) at (318pt,-4pt) {$\lam xy.x$};
\node at ($(Nh1)+(8pt,18pt)$) {$\BT{N}$};
\node at ($(Nh1)+(8pt,8.5pt)$) {$\shortparallel$};
\node (Nh21) at ($(Nh1)+(-20pt,-20pt)$) {$\lam z_0w_0.x$};
\node (Nh22) at ($(Nh1)+(10pt,-20pt)$) {$y$};
\node (Nh23) at ($(Nh1)+(37pt,-20pt)$) {$x$};
\node (Nh31) at ($(Nh21)+(-15pt,-20pt)$) {$\lam z_1.z_0$};
\node (Nh41) at ($(Nh31)+(0pt,-20pt)$) {$\lam z_2.z_1$};
\node (Nh51) at ($(Nh41)+(7.5pt,-20pt)$) {$\vdots$};
\draw ($(Nh1)+(6pt,-4pt)$) -- ($(Nh21.north)+(0.6,-2pt)$);
\draw ($(Nh1)+(10pt,-4pt)$) -- (Nh22);
\draw ($(Nh1)+(14pt,-4pt)$) -- (Nh23);
\draw ($(Nh21)+(12pt,-5pt)$) -- ($(Nh31)+(7pt,4pt)$);
\draw ($(Nh31)+(7pt,-5pt)$) -- ($(Nh41)+(7pt,4pt)$);
\draw ($(Nh41)+(7pt,-5pt)$) -- ($(Nh51)+(0pt,4pt)$);
\node (Ph31) at ($(Nh21)+(30pt,-20pt)$) {$\lam w_1.w_0$};
\node (Ph41) at ($(Ph31)+(0pt,-20pt)$) {$\lam w_2.w_1$};
\node (Ph51) at ($(Ph41)+(7.5pt,-20pt)$) {$\vdots$};
\draw ($(Nh21)+(16pt,-5pt)$) -- ($(Ph31)+(7pt,4pt)$);
\draw ($(Ph31)+(7pt,-5pt)$) -- ($(Ph41)+(7pt,4pt)$);
\draw ($(Ph41)+(7pt,-5pt)$) -- ($(Ph51)+(0pt,4pt)$);
\end{tikzpicture}
\caption{A situation witnessing the fact that $M\sqle_\Hst N$ holds.}\label{fig:BTle_eta}
\end{figure}

%% file: include/fig3.tex
\begin{figure}[t!]
\begin{tikzpicture}
\node (myspot) at (-5.9,0) {~};
\node (lev1) at (0.5, .1) {\!\!\!\!\!\!\!\!\!\!\!\!\!\!\!\!\!\!\!\!\!\!\!\!\!\!\!\!\!\!\!\!\!$\lam xy_1\dots y_{T\emptyseq}.x$};
\node (lev2) at (0,-1) {$
	\lam z_1\dots z_{T\langle0\rangle}.y_1\qquad\qquad\ \, \cdots \quad \qquad
	\lam z_1\dots z_{T\langle T\emptyseq\textrm{-}1\rangle}.y_{T\emptyseq}$};
\node (lev3) at (.8,-2) {
	$\lam\vec w_{T\langle0,0\rangle}.z_1\quad \cdots \quad 
	  \lam\vec w_{T\langle0,T\langle0\rangle\textrm{-}1\rangle}.z_{T\langle0\rangle}\qquad
	  \lam \vec w_{T\langle T\emptyseq\textrm{-}1,0\rangle}.z_1\quad\cdots\quad 
	  \lam \vec w_{T\langle T\emptyseq\textrm{-}1,T\langle T\emptyseq\textrm{-}1\rangle\textrm{-}1\rangle}.z_{T\langle T\emptyseq\textrm{-}1\rangle}$};
	  \draw[-] ($(lev1.south east)+(-8pt,2pt)$) -- (-2.3,-.8);
	  \draw[-] ($(lev1.south east)+(-2pt,2pt)$) -- (3.8,-.8);	  	  
	  \draw[-] (-2.2,-1.2) -- (-1.2,-1.7);
	  \draw[-] (-2.4,-1.2) -- (-4.75,-1.7);	  
	  \draw[-] (4,-1.2) -- (2.7,-1.7);	  	  	  	  
	  \draw[-] (4.2,-1.2) -- (7.2,-1.7);	
	  
	  \draw[-] (-4.9,-2.15) -- (-3.75,-3);	 
	  \draw[-] (-5.1,-2.15) -- (-6.45,-3);	  	  	     	  
	  \node at (-5.1,-2.75) {$\cdots$}; 	  	  
	 
	  \draw[-] (-1.15,-2.15) -- (-2.45,-3);	  	  	  
  	  \draw[-] (-.95,-2.15) -- (.25,-3);	  	  	  	                      	  
	  \node at (-1.1,-2.75) {$\cdots$};	  
	
	  \draw[-] (2.35,-2.15) -- (3.6,-3);	  	  	                                           	  
	  \draw[-] (2.15,-2.15) -- (.9,-3);	  
	  \node at (2.3,-2.75) {$\cdots$};
	  
	    \draw[-] (7.2,-2.15) -- (8.5,-3);	  	  
	  \draw[-] (7,-2.15) -- (5.8,-3);	  	  	  	  	  
	  \node at (7.2,-2.75) {$\cdots$};	  	  	  
\end{tikzpicture}
\caption{The B\"ohm-like tree of an infinite $\eta$-expansion of $\bI$ following $T\in\Trees$. 
To lighten the notations we write $T\sigma$ rather than $T(\sigma)$
and we let $\seq w_n:=w_1,\dots,w_n$.}
\label{fig:BTJT}
\end{figure}

%% file: include/rgms.tex
In this section we recall the definition of a \emph{relational graph model} (\rgm, for short).
Individual examples of \rgm's have been previously studied in the literature as models of the \lam-calculus~\cite{BucciarelliEM07,Manzonetto09,deCarvalhoTh,HylandNPR06}, of nondeterministic \lam-calculi~\cite{BucciarelliEM12} and of resource calculi~\cite{Manzonetto12,PaganiR10}.
However, the class of relational graph models was formally introduced in~\cite{ManzonettoR14}.

\subsection{The Relational Semantics}

Relational graph models are called \emph{relational} since they are reflexive objects in the Cartesian closed category \MRel~\cite{BucciarelliEM07}, which is the Kleisli category of the finite multisets comonad $\Mfin{-}$ on \Rel. 
Since we do not use the underlying symmetric monoidal category, we present directly its Cartesian closed structure.
Recall that the definitions and notations concerning multisets have been introduced in Subsection~\ref{subs:msets}.

\begin{definition}
The category $\MRel$ is defined as follows:
\begin{itemize}
\item 
    The objects of $\MRel$ are all the sets.
\item 
    A  morphism from $A$ to $B$ is a relation from $\Mfin A$ to $B$; in other words, $\MRel(A,B)=\Pow{\Mfin A\times B}$.
\item 
    The identity of $A$ is the relation $\Id{A}=\{(\Mset \alpha,\alpha)\st \alpha\in A\}\in\MRel(A,A)$.
\item 
    The composition of $f\in\MRel(A,B)$ and $g\in\MRel(B,C)$ is defined by:
    $$
    \begin{array}{ll}
    g\comp f=\big\{(a,\gamma)\quad \st&\exists k\in\nat,\ \exists (a_1,\beta_1),\dots,(a_k,\beta_k)\in f\textrm{ such that } \\
    &a = a_1+\dots+ a_k\ \text{and}\ (\Mset{\beta_1,\dots,\beta_k},\gamma)\in g\ \big\}.\\
    \end{array}
    $$
\end{itemize}
\end{definition}

Given two sets $A_1,A_2$, we denote by $\With{A_1}{A_2}$ their disjoint union $(\{1\}\times A_1) \cup (\{2\}\times A_2)$. 
Hereafter we adopt the following convention.

\begin{convention} 
We consider the canonical bijection (also known as \emph{Seely isomorphism}~\cite{Bierman95}) between $\Mfin{A_1}\times\Mfin{A_2}$ and $\Mfin{\With{A_1}{A_2}}$ 
as an equality.
As a consequence, we still denote by $(a_1,a_2)$ the corresponding element of $\Mfin{\With{A_1}{A_2}}$.
\end{convention}

\begin{theorem}\label{thm:MRel-ccc} The category $\MRel$ is a Cartesian closed category.
\end{theorem}
\begin{proof} The terminal object $\Termobj$ is the empty set $\emptyset$, and the unique element of
$\MRel(A,\emptyset)$ is the empty relation.

Given two sets $A_1$ and $A_2$, their categorical product in $\MRel$ is
their disjoint union $\With{A_1}{A_2}$ and the projections $\Proj{1},\Proj{2}$ are given by:
$$
  \Proj{i}=\big\{(\Mset{(i,a)},a)\st a\in A_i\big\}\in\MRel(\With{A_1}{A_2},A_i)\textrm{, for } i =1,2.
$$
It is easy to check that this is actually the categorical product of $A_1$ and $A_2$ in $\MRel$; 
given $f\in\MRel(B,A_1)$ and $g\in\MRel(B,A_2)$, the corresponding morphism $\Pair fg\in\MRel(B,\With{A_1}{A_2})$ is given by:
$$
  \Pair fg=\big\{(b,(1,\alpha))\st(b,\alpha)\in f\}\cup\{(b,(2,\alpha))\st(b,\alpha)\in g\big\}\,.
$$
Given two objects $A$ and $B$, the exponential object $\Funint AB$ is $\Mfin A\times B$ and the evaluation morphism is given by:
$$
\eval_{AB} =\big\{\big((\Mset{(a,\beta)},a),\beta\big)\st a\in\Mfin A\ \text{and}\ \beta\in B\big\}\in\MRel(\With{(A\To B)}A,B)\,.
$$
Again, it is easy to check that in this way we defined an exponentiation. 
Indeed, given any set $C$ and any morphism $f\in\MRel(\With CA,B)$, there is exactly one morphism $\Curry(f)\in\MRel(C, A\To B)$ such that:
$$
  \eval_{AB}\comp (\Curry(f)\times \Id{S}) = f.
$$
which is $\Curry(f)=\big\{(c,(a,\beta))\st((c,a),\beta)\in f\big\}$. 
\end{proof}

The category $\MRel$ provides a simple example of a non-well-pointed category.

\begin{theorem} No object $A\neq\Termobj$ is well-pointed, so neither is $\MRel$.
\end{theorem}

\begin{proof} For every $A\neq\emptyset$, we can always find $f,g : A\to A$ such that $f\neq g$ and, for all $h : \Termobj\to A, f\comp h = g\comp h$. Indeed, by definition of composition, $f\comp h =\{(\emptymset,\alpha) \st \exists \beta_1,\dots,\beta_k \in A, (\emptymset,\beta_i)\in h, ([\beta_1,\dots,\beta_k],\alpha)\in f\}$, and similarly for $g\comp h$. Hence it is sufficient to choose $f = \{(a_1, \alpha)\}$ and $g = \{(a_2, \alpha)\}$ for $a_1,a_2$ different multisets with the same support.
\end{proof}

\subsection{The Class of Relational Graph Models} 
The class of graph models constitutes a subclass of the continuous semantics~\cite{Scott72} and is the simplest generalization of the Engeler and Plotkin's construction~\cite{Engeler81,Plotkin93}.
This class has been widely studied in the literature and has been used to prove several interesting results~\cite{Berline00}.
We recall that a graph models is given by a set $A$ and a total injection $i : \Powf A\times A\to A$, and induces $\Pow A$ as a reflexive object in the category of Scott's domains and continuous functions.
Therefore, a bijective $i$ does not induce automatically an isomorphism between $\Pow A$ and $\Funint{\Pow A}{\Pow A}$. 
As shown in~\cite[\S5.5]{Berline00}, no graph model $\mathscr{G}$ can be extensional because $\mathscr{G}\models\One\sqle\bI$ is never satisfied.

The definition of a relational graph model mimics the one of a graph model while replacing finite sets with finite multisets.
As we will see, relational graph models capture a particular subclass of reflexive objects living in \MRel. 

\begin{definition}\label{def:graphrel}
A \emph{relational graph model} $\cD=(D,i)$ is given by an infinite set $D$ and a total injection $i : \Mfin D \times D\to D$. 
We say that $\cD$ is \emph{extensional} when $i$ is bijective.
\end{definition}
The equality $i(a,\alpha) = \beta$ indicates that the ``arrow type'' $a\to \alpha$ is equivalent to the type~$\beta$. In particular, in an extensional relational graph model, every element of the model can be seen as an arrow. Keeping this intuition in mind, we adopt the notation below.

\begin{notation}\label{notation:arrow} Given an rgm $\cD = (D,i)$, $a\in\Mfin D$ and $\alpha\in D$, we write $a\to_i \alpha$ (or simply $a\to\alpha$, when $i$ is clear) as an alternative notation for $i(a,\alpha)$.
\end{notation}

As shown in the next proposition, the reflexive object induced by a relational graph model $(D,i)$ is not some powerset of $D$ as in the case of regular graph models, but rather $D$ itself.
This opens the way to define extensional reflexive objects.

\begin{proposition} 
Given an rgm $\cD =(D,i)$ we have that:
\begin{enumerate}[(i)]
\item\label{yeppa} $\cD$ induces a reflexive object $(D,\appl,\Abs)$ where
$$
    \Abs = \big\{(\Mset{(a,\alpha)},a\to_i\alpha)\st a\in\Mfin{D},\alpha\in D\big\}\in\MRel( \Funint DD, D),
$$
$$
	\appl = \big\{(\Mset{a\to_i\alpha},(a, \alpha)) \st a\in\Mfin{D},\alpha\in D\big\}\in\MRel(D, \Funint DD),
$$
\item\label{yeppa2}
	If moreover $\cD$ is extensional, then also the induced reflexive object is.
\end{enumerate}
\end{proposition}

\begin{proof} (\ref{yeppa}) $\appl\circ\Abs = \big\{\big([(a,\alpha)],(a,\alpha) \big) \st ([(a,\alpha)], a\to\alpha)\in\Abs,  ([a\to\alpha],(a,\alpha))\in\appl \big\} =\big\{\big([(a,\alpha)],(a,\alpha) \big) \st a\in\Mfin{D},\alpha\in D \big\}  = \Id{\Mfin{D}\times D} = \Id{D\To D}$. 

 (\ref{yeppa2}) If  $i$ is bijective  for every $\beta\in D$ we have $\beta = a\to_i \alpha$ for some $a\in\Mfin{D}$ and $\alpha\in D$. So $\Abs\circ\appl = \big\{\big([a\to\alpha],a\to\alpha \big) \st  ([a\to\alpha],(a,\alpha))\in\appl, ([(a,\alpha)], a\to\alpha)\in\Abs \big\} =\big\{\big([a\to\alpha],a\to\alpha \big) \st a\in\Mfin{D},\alpha\in D \big\} =\big\{\big([\beta],\beta \big) \st \beta\in D \big\} =\Id{D}$.
\end{proof}

Note that, when $i$ is just injective, there are in principle different morphisms that could be chosen as inverses of $\Abs$.
Therefore, there exist reflexive objects in \MRel{} that are not relational graph models.
However, all inverses of $\Abs$ induce relational models that are strongly linear in the sense of~\cite{PaoliniPR15}. Finally, since every isomorphism $f\in\MRel(A,A)$ is of the form $f = \{([\alpha],i(\alpha))\st \alpha\in A\}$ for some bijective map $i$, the class of extensional relational graph models coincides with the one of extensional reflexive objects living in this category.

\subsection{Building Relational Graph Models by Completion}

Relational graph models -- just like the regular ones -- can be built by performing the free completion of a partial pair.

\begin{definition}\ 
\begin{enumerate}
\item 
	A \emph{partial pair} $\cA$ is a pair $(A,j)$ where $A$ is a non-empty set of elements (called \emph{atoms}) and $j:\Mfin A\times A\to A$ is a partial injection. 
\item 
	A partial pair $\cA$ is called \emph{extensional} when $j$ is a bijection between $\dom(j)$ and $A$. 
\item
	A partial pair $\cA$ is called \emph{total} when $j$ is a total function, and in this case $\cA$ is a relational graph model.
\end{enumerate}
\end{definition}
Hereafter, we consider without loss of generality partial pairs $\cA$ whose underlying set $A$ does not contain any pair of elements. In other words, we assume $(\Mfin{A}\times A)\cap A = \emptyset$. 
This is not restrictive because partial pairs can be considered up to isomorphism.
\begin{definition}\label{def:completion} 
The \emph{free completion} $\Compl{\cA}$ of a partial pair $\cA$ is the pair $(\Compl{A},\Compl{j})$ defined as:
	$\Compl{A} = \bigcup_{n\in\nat} A_n$, where $A_0 = A$ and $A_{n+1} = ((\Mfin {A_n}\times A_n) - \dom(j)) \cup A$\,; moreover
$$
			\Compl{j}(a,\alpha) =\begin{cases}
										j(a,\alpha)&\textrm{if }(a,\alpha)\in \dom(j),\\
										(a,\alpha)&\textrm{otherwise.}\\
										\end{cases}
$$
\end{definition}

It is well known for graph models, and easy to check for relational graph models, that every $\cD$ is isomorphic to its own free completion $\Compl{\cD} \cong \cD$. 
In particular, given a partial pair $\cA$, we have that $\Compl{\Compl{\cA}}\cong\Compl{\cA}$.





%
\begin{proposition} If $\cA$ is an (extensional) partial pair, then $\Compl{\cA}$ is an (extensional) \rgm.
\end{proposition}
\begin{proof} The proof of the fact that $\Compl{\cA}$ is an \rgm{} is analogous to the one for regular graph models \cite{Berline00}. It is easy to check that when $j$ is bijective, also its completion $\Compl{j}$ is. 
\end{proof}

The following relational graph models are built by free completion and will be running examples in the next sections.

\begin{example}\label{ex:graphrel} 
Some examples of relational graph model:
\begin{itemize}
\item $\cE = \Compl{(\nat,\emptyset)}$ was introduced in~\cite{HylandNPR06},
\item $\cD_{\omega} = \Compl{(\{\star\}, \{(\emptymset,\star)\mapsto \star\})}$ was first defined 
(up to isomorphism) in~\cite{BucciarelliEM07},
\item $\cD_{\star} = \Compl{(\{\star\}, \{([\star],\star)\mapsto \star\})}$ was introduced in \cite{ManzonettoR14}.
\end{itemize}
Notice that $\cD_{\omega}$ and $\cD_{\star}$ are extensional, while $\cE$ is not.
\end{example}

\subsection{Categorical Interpretation}
We now show how \lam-terms and B\"ohm trees can be interpreted in a relational graph model, 
and we review their main properties.

Recall that notions and notations concerning multisets have been introduced in~Section~\ref{subs:msets}.
Given two $n$-uples $\seq a,\seq b\in\Mfin{A}^n$ we write $\seq a + \seq b$ for $(a_1+b_1,\dots,a_n+b_n)\in\Mfin{A}^n$.

\begin{definition}\label{def:IntSem} 
Let $\cD$ be an rgm, $M\in \Lam$ and $\FV(M) \subseteq \{x_1, \dots, x_n\}$.
The \emph{categorical interpretation of $M$ in $\cD$ w.r.t.~$\seq{x}$} is the relation $\Int[\cD]{M}_{\seq{x}}\subseteq\Mfin{D}^n\times D$ defined by:
\begin{enumerate}[i.]
\item\label{def:IntSem1} 
	$\Int[\cD]{x_i}_{\seq{x}} \,=\, \big\{ ((\emptymset,\dots,\emptymset,[\alpha],\emptymset,\dots,\emptymset),\alpha\big) \st \alpha\in D\big\}$,
		where $[\alpha]$ stands in $i$-th position.
		\item\label{def:IntSem2} 
		$\Int[\cD]{\lam y.N}_{\seq{x}}  \,=\, \big\{(\seq{a}, a\to_i\alpha) \st ((\seq{a},a),\alpha) \in \Int[\cD]{N}_{\seq{x},y}\big\}$ where we take~$y\notin\seq x$ by $\alpha$-conversion.
\item\label{def:IntSem3} 
	$\Int[\cD]{PQ}_{\seq{x}}  \,=\, \big\{((\seq{a_0}+\dots+\seq{a_k}),\alpha) \st \exists\, \alpha_1,\dots,\alpha_k\in\cD \textrm{ such that } \\ 
	{}\hspace{57pt}(\seq{a_0}, [\alpha_1,\dots,\alpha_k]\to_i \alpha)\in\Int[\cD]{P}_{\seq{x}} \textrm{ and }(\seq{a_j},\alpha_j)\in\Int[\cD]{Q}_{\seq{x}}  \textrm{ for all } 1\le j\le k$ \big\}.
\end{enumerate}
This definition extends to \lamb-terms $M$ by setting $\Int[\cD]{\bot}_{\seq x} = \emptyset$ and to B\"ohm trees of \lam-terms by interpreting all their finite approximants, namely by setting $\Int[\cD]{\BT{M}}_{\seq x} = \bigcup_{t\in\BT[*]{M}} \Int[\cD]{t}_{\seq x}$.
  \end{definition}
  
It is easy to check that the definition above is an inductive characterization of the usual categorical interpretation of \lam-terms as morphisms of a Cartesian closed category. 

\begin{convention}
From now on, whenever we write $\Int[\cD]{M}_{\seq x}$ we always assume that $ \FV(M)\subseteq\seq x$. 
When $M$ is a closed \lam-term we consider $\Int[\cD]{M}$ simply as a subset of $D$.
In all our notations we omit the model $\cD$ when it is clear from the context.
\end{convention}

\begin{example}\label{ex:interp} Let $\cD$ be any \rgm. Then we have:
\begin{enumerate}
\item 
	$\Int[\cD]{\bI}=\{ [\alpha]\to\alpha \st \alpha\in D\}$ and 
\item 
	$\Int[\cD]{\One} = \{[a\to\alpha]\to a\to\alpha \st \alpha\in D, a\in\Mfin{D}\}$, thus:
\item
	$\Int[\cD]{\bold{J}} = \{[\alpha]\to\alpha \st \alpha\in D'\}\subseteq \Int[\cD]{\One}\subseteq \Int[\cD]{\bI}$, where $D'$ is the smallest subset of $D$ satisfying: if $\alpha\in D$ then $\emptymset\to\alpha\in D'$; if $\alpha\in D$ and $a\in\Mfin{D'}$ then $a\to\alpha\in D'$,
	\item
	$\Int[\cD]{\Delta} = \{(a + [a\to\alpha])\to\alpha \st \alpha\in D, a\in\Mfin{D}\}$ therefore:
\item
	$\Int[\cD]{\Omega} = \Int[\cD]{\bot} = \emptyset$,
\item
	$\Int[\cD]{\lam x.x\bO} = \{ [\emptymset\to\alpha]\to\alpha \st \alpha\in D \}$.
\end{enumerate}
Consider the relational graph models $\cD_\omega$ and $\cD_\star$ from Example~\ref{ex:graphrel}.
From the calculations above it follows that $\Int{\bI}=\Int{\One}$ in both $\cD_\omega$ and $\cD_{\star}$, but
$\Int[\cD_\omega]{\bI}=\Int[\cD_\omega]{\bold{J}}$, while $\star\in\Int[\cD_{\star}]{\bI} - \Int[\cD_{\star}]{\bold{J}}$.
\end{example}

\subsection{Soundness}
Relational graph models satisfy the following substitution property and are sound
models of \lam-calculus in the sense that they equate all $\beta$-convertible \lam-terms.

\begin{lemma}[Substitution]
Let $M,N\in\Lambda$ and $\cD$ be an rgm.
For $\seq a\in\Mfin{D}^n$ and $\alpha\in D$ we have that
$(\seq a,\alpha)\in\Int[\cD]{M\{N/y\}}_{\seq x}$ if and only if there exist
        $b=[\beta_1,\dots,\beta_k]\in \Mfin{D}$ and $\seq a_0,\dots,\seq a_k\in\Mfin{D}^n$ such
        that $(\seq a_\ell,\beta_\ell)\in\Int[\cD]{N}_{\seq x}$, for $1\le\ell\le k$,
        $((\seq a_0,b),\alpha)\in\Int[\cD]{M}_{\seq x,y}$ and ${\seq a=\sum_{\ell=0}^k \seq a_\ell}$.
\end{lemma}

\begin{proof} We proceed by induction on $M$, the only interesting case being $M=PQ$.

$(\Rightarrow)$
We know that $(\seq a,\alpha)\in\Int{M\subst{y}{N}}_{\seq x}$ if and only if there are $\gamma_1,\dots,\gamma_k$ and a decomposition $\seq a = \sum_{\ell=0}^k \seq a^\ell$ such that $(\seq a^{\,0},[\gamma_1,\dots,\gamma_k]\to\alpha)\in\Int{P\subst{y}{N}}_{\seq x}$ and $(\seq a^\ell,\gamma_\ell)\in\Int{Q\subst{y}{N}}_{\seq x}$ for $1 \le \ell \le k$.
By applying the induction hypothesis to the former assumption, we get $b^0 =[\beta^0_1,\dots,\beta^0_m]$ and a decomposition $\seq a^{\,0} = \sum_{j=0}^m \seq a^{\,0}_j$ such that $((\seq a^{\,0}_0,b^0),{[\gamma_1,\dots,\gamma_k]\to \alpha})\in\Int{P}_{\seq x,y}$ and $(\seq a^{\,0}_i,\beta^0_i)\in\Int{N}_{\seq x}$ for $1\le i\le m$.
From the latter, for each $\ell = 1,\dots,k$ we get $b^\ell = [\beta^\ell_1,\dots,\beta^\ell_{k_\ell}]$ and a decomposition $\seq a^\ell = \sum_{j=0}^{k_\ell} \seq a_j^\ell$ such that $((\seq a^\ell_0,b^\ell),\gamma_\ell)\in\Int{Q}_{\seq x,y}$ and $(\seq a^\ell_j,\beta^\ell_j)\in\Int{N}_{\seq x}$ for $1\le j \le k_\ell$.
We conclude that $((\sum_{\ell=0}^k \seq a_0^\ell,\sum_{\ell=0}^k  b^\ell),\alpha)\in\Int{PQ}_{\seq x,y}$.

$(\Leftarrow)$ By analogue calculations.
\end{proof}

\begin{lemma}[Monotonicity]\label{lemma:monotonicity} Let $\cD$ be an rgm, and $M,N\in\Lam$. 
If $\Int[\cD]{M}_{\seq x} \subseteq \Int[\cD]{N}_{\seq x}$ then for all contexts $C\hole-$ we have $\Int[\cD]{C\hole M}_{\seq x} \subseteq \Int[\cD]{C\hole N}_{\seq x}$.
\end{lemma}

\begin{proof} Notice that, by the Convention above, we assume $\FV(M)\cup\FV(N)\cup\FV(C[-])\subseteq\seq x$.
The result follows by a straightforward induction on $C\hole-$. 
\end{proof}
\begin{theorem}[Soundness]\label{thm:soundness}
Let $M,N\in\Lambda$ and $\cD$ be a relational graph model.
If $M =_{\beta} N$ then $\Int[\cD]{M}_{\seq x} = \Int[\cD]{N}_{\seq x}$. 
\end{theorem}

\begin{proof}
From the substitution lemma and Lemma~\ref{lemma:monotonicity} we have that $M\redto[\beta] M'$ entails $\Int{M}_{\seq x} = \Int{M'}_{\seq x}$. 
By Church Rosser $M=_\beta N$ if and only if they have a common reduct $P$ such that $M\msto[\beta]P$ and $N \msto[\beta] P$. 
Summing up, we have $\Int{M}_{\seq x} = \Int{P}_{\seq x} = \Int{N}_{\seq x}$.
\end{proof}

\begin{definition}\ 
\begin{itemize}
\item
	The \emph{\lam-theory} induced by a relational graph model $\cD$ is defined by  
	$$
		\Th(\cD) = \{ (M,N)\in\Lam\times\Lam\st \Int{M}_{\seq{x}} =\Int{N}_{\seq{x}}\}\,.
	$$ 
	We write $\cD\models M = N$ for $(M,N)\in\Th(\cD)$.
\item 
	Similarly, the \emph{inequational theory} induced by $\cD$ is given by 
	$$
		\Thle(\cD) = \{ (M,N)\in\Lam\times\Lam\st \Int{M}_{\seq{x}}\subseteq\Int{N}_{\seq{x}}\}\,.
	$$ 
	 We write $\cD\models M \sqsubseteq N$ for $(M,N)\in\Thle(\cD)$.  
\item
	 An rgm $\cD$ is called \emph{$\cO$-fully abstract} when $\cD\models M= N$ if and only if $M\obseq[\cO] N$. 
\item	 
	$\cD$ is \emph{inequationally $\cO$-fully abstract} when $\cD\models M\sqsubseteq N$ if and only if $M\obsle[\cO] N$.
\item 
	We say that a \lam-theory (resp.\ inequational theory) $\cT$ is \emph{representable by a relational graph model} if there exists an rgm $\cD$ such that $\Th(\cD) = \cT$ (resp.~$\Thle(\cD) = \cT$).
\item 
	We say that a \lam-theory (resp.\ inequational theory) is a \emph{(resp.\ inequational) relational graph theory} if it is represented by some relational graph model.
\end{itemize}
\end{definition}

\begin{lemma}\label{lemma:D_ext_actually_ext} Let $\cD$ be an rgm.
\begin{enumerate}[(i),ref={\roman*}]
\item\label{lemma:D_ext_actually_ext1}
	If $M\to_\eta N$ then $\Int[\cD]{N}_{\seq{x}}\subseteq \Int[\cD]{M}_{\seq x}$,
\item \label{lemma:D_ext_actually_ext2}
	$\cD$ is extensional if and only if $\blam\eta\subseteq\Th(\cD)$. 
\end{enumerate}
\end{lemma}
\begin{proof}
\eqref{lemma:D_ext_actually_ext1} By inspecting their interpretations (Example~\ref{ex:interp}(1-2)) it is clear that $\Int{\One}\subseteq \Int{\bI}$. 
For all $N\in\Lam$, $\cD\models \One N \sqle \bI N$ (monotonicity) which entails $\cD\models \lam x.Nx \sqle N$ (soundness).

\eqref{lemma:D_ext_actually_ext2} 
By monotonicity $\blam\eta\subseteq\Th(\cD)$ if and only if $\cD\models \One=\bI$.
As above, we know that $\cD\models \One\sqle\bI$ holds.
The other inclusion holds if and only if for all $\alpha\in D$ there exist $a,\beta$ such that $i(a,\beta) = \alpha$ if and only if $i$ is bijective.
\end{proof}

As a consequence, the \lam-theories induced by relational graph models and by 
ordinary graph models are different, since no graph model is extensional.


%% file: include/logical.tex
Every relational graph model can be presented as a type system. 
The interpretation of a $\lambda$-term $M$ is given by the set of all pairs $(\Gamma,\alpha)$ such that $\Gamma\vdash M: \alpha$ is derivable in the system. 
Such a \emph{logical} interpretation turns out to be equivalent to the categorical one (Theorem~\ref{thm:type_semantics}). 
This presentation exposes the \emph{quantitative} nature of the relational semantics and allows to provide a combinatorial proof of the Approximation Theorem (Theorem~\ref{thm:app}).

\input{include/fig4}
\subsection{Relational Graph Models as Type Systems}\label{subsec:TS}
The \emph{types} of the system associated to a relational graph model $\cD = (D,i)$ are the elements of the underlying set $D$ themselves. 
We recall from Notation~\ref{notation:arrow} that $a\to \alpha$ denotes the element $i(a,\alpha)$ which belongs to $D$.  

\begin{definition}
 An \emph{environment} for $\cD$ is a map $\Gamma:\Var\to\Mfin{D}$ such that $\supp(\Gamma)= \{\, x\in\Var \st \Gamma(x)\neq\emptymset \,\}$ is finite. The set of all environments for $\cD$ is denoted by $ \Env{\cD}$.  
 \end{definition}

We write $x_1 : a_1,\dots,  x_n:a_n$ for the environment $\Gamma$ such that $\Gamma(x_i)=a_i$ if $1\le i\le n$ and $\Gamma(y) = \emptymset$ otherwise.
When $\supp(\Gamma)=\emptyset$ the environment $\Gamma$ is just omitted.

\begin{definition}
Given $\Gamma,\Delta\in\EnvD$ we define the environment $\Gamma+\Delta$ by setting $(\Gamma+\Delta)(x) = \Gamma(x)+\Delta(x)$ for all $x\in\Var$.
\end{definition}

A \emph{(type) judgment} is a triple $(\Gamma, M, \alpha)\in \Env{\cD}\times\Lambda_\bot\times D$ that will be denoted as $\,\Gamma \vdash_{\scriptscriptstyle\cD} M: \alpha$, or simply $\,\Gamma \vdash M: \alpha$ whenever $\cD$ is clear from the context. 

\begin{definition}\label{def:systems} 
Let $\cD$ be a relational graph model. 
The inference rules of the \emph{type  system $\;\vdash_{\scriptscriptstyle \cD}$ for $\Lam_\bot$ associated with $\cD$} are given in Figure~\ref{fig:Typing}. 
\end{definition}

We remark that these type systems are \emph{relevant} since the weakening is not available. 
Indeed, the rule~\texttt{var} does not allow a generic environment $\Gamma, x:[\alpha]$ and the sum of contexts in~\texttt{app} takes multiplicities into account.
The types are \emph{strict} in the sense that multisets may only appear at the left hand-side of an arrow. 
In particular, no \lamb-term can have type~$\omega$.

The number $n$ appearing in the rule \texttt{app} can be $0$. So  we have 
the inference rule
\begin{equation}\label{noarg}
\infer{ \;\Gamma \;\vdash_{\scriptscriptstyle{\cD}}  MN : \alpha \;}{\;\Gamma\;\vdash_{\scriptscriptstyle{\cD}} M:\emptymset\to\alpha\;}
\end{equation}
for every $N\in\Lam_\bot$. 
For example, even if $\bO$ is not typable in the system associated with any relational graph model, the following derivation is always possible for every $\alpha\in D$:
\[
\infer{
	\vdash_{\scriptscriptstyle{\cD}} \lam x.x\,\bO : [\emptymset\to\alpha]\to\alpha
	}{\infer{
		x:[\emptymset\to\alpha] \vdash_{\scriptscriptstyle{\cD}} x\,\bO : \alpha
		}{
			x:[\emptymset\to\alpha] \vdash_{\scriptscriptstyle{\cD}} x:\emptymset\to\alpha
		}
}
\]

Hereafter, when writing $\Gamma \vdash_{\scriptscriptstyle{\cD}} M: \alpha$, we intend that such a judgment is derivable.
We write $\pi \derof \Gamma \vdash_{\scriptscriptstyle{\cD}} M: \alpha$ to indicate that $\pi$ is a derivation tree of the judgment $ \Gamma \vdash_{\scriptscriptstyle{\cD}} M: \alpha$. 

\begin{definition}
Let $\pi,\pi'$ be two derivation trees.
We set $\pi \justlike \pi'$ if and only if the trees obtained from $\pi$ and $\pi'$ by removing the \lamb-terms from each of their nodes coincide.
\end{definition}

\begin{example}\label{ex:nakedtrees}
Let $\pi$ and $\pi'$ be the following derivation trees:
\[
	\infer{ 
		x: [\alpha]\; \vdash_{\scriptscriptstyle{\cD}} (\lam y.x)\,z:\alpha
	}{
		\infer{
			x:[\alpha] \;\vdash_{\scriptscriptstyle{\cD}} \lam y.x: \emptymset \to\alpha
		}{ 
			x:[\alpha] \;\vdash_{\scriptscriptstyle{\cD}} x: \alpha 
		}
}
\qquad\qquad\quad
	\infer{
		x: [\alpha] \;\vdash_{\scriptscriptstyle{\cD}} (\lam y.x)\,\bK : \alpha 
	}{
		\infer{ 
			x: [\alpha] \;\vdash_{\scriptscriptstyle{\cD}} \lam y.x: \emptymset \to\alpha 
		}{ 
			x: [\alpha] \;\vdash_{\scriptscriptstyle{\cD}} x: \alpha 
		}
	}
\]
We have $\pi \justlike \pi'$ because once all \lamb-terms are erased they both become like this:
\[
	\infer{ 
		x: [\alpha] \;\vdash_{\scriptscriptstyle{\cD}}\quad:\alpha 
		}{
		\infer{ 
			x: [\alpha] \;\vdash_{\scriptscriptstyle{\cD}} \quad: \emptymset \to\alpha
		}{ 
			x: [\alpha] \;\vdash_{\scriptscriptstyle{\cD}} \quad: \alpha 
		}
	}
\]
\end{example}
From an intuitive perspective, the equivalence $\pi \justlike \pi'$ says that $\pi$ and $\pi'$ are roughly the \emph{same} derivation tree (but it can possibly be used to type distinct \lamb-terms). 

\begin{lem}\label{lem:env_fv}
Let $\cD$ be an $\rgm$. If $\, \Gamma \vdash_{\scriptscriptstyle \cD} M:\alpha\,$  then $\,\supp(\Gamma)\subseteq \FV(M)$.
\end{lem}
\begin{proof}
By a straightforward induction on the derivation of the  judgment.
\end{proof}

As shown by the first derivation of Example~\ref{ex:nakedtrees}, the inclusion in Lemma~\ref{lem:env_fv} can be strict, indeed we have $x : [\alpha]\vdash (\lam y.x)z : \alpha$ and $\supp(x : [\alpha]) = \{x\}\subsetneq\FV((\lam y.x)z)$.
In general, one should realize that whenever $\pi \derof \Gamma \vdash M: \alpha\,$ and $\,\supp(\Gamma)\subsetneq\FV(M)$ then along $\pi$ some subterm $N$ of $M$ comes in argument position and is not actually typed, as in \eqref{noarg}.  

We could formalize the type systems $\,\vdash_{\scriptscriptstyle \cD}$ in a style more similar to traditional intersection type systems (as in~\cite{PaoliniPR15}), in order to expose clearly the intuition of relational graph models as resource-sensitive versions of filter models~\cite[Part III]{BareTypes}, not only of graph models.  
We followed that approach in~\cite{ManzonettoR14} and~\cite{phdruoppolo}, where the multisets occurring in the types were actually denoted as non-idempotent intersections and an explicit conversion rule
\[
\infer[\texttt{eq}]{\Gamma\vdash_{\scriptscriptstyle \cD} M : \alpha}{{\Gamma\vdash_{\scriptscriptstyle \cD} M : \beta}\quad\;{\beta \simeq \alpha}}
\]
was available in the system. 
Since such a presentation does not change the expressive power of the systems, but complicates the technical proofs (as it obliges to consider the possible commutations of the rule \texttt{eq} along a given derivation $\pi$) we decided to avoid it here.

\subsection{Logical Interpretation}\label{subsec:LI}
The type system associated to a relational graph model $\cD$ provides an alternative way to define the interpretation of a \lamb-term. 
This \emph{logical/type-theoretical} approach is rather common and fits in the tradition of filter models~\cite[Part~III]{BareTypes}.

\begin{definition}\label{def:interpretazione} 
 Let $\cD$ be an $\rgm$ and $M\in\Lam_\bot$. The \emph{logical interpretation} of $M$ in $\cD$ is:
$$
	\Lint[\cD]{M} =\big\{\,(\Gamma,\alpha)\in \EnvD \times D \st \,\Gamma\,\vdash_{\scriptscriptstyle \cD} M : \alpha\,\big\}.
$$
When $\cD$ is clear from the context we write $\Lint{M}$, and when $M$ is closed we consider $\Lint{M}\subseteq D$. 
This definition extends to B\"ohm trees of \lam-terms like Definition~\ref{def:IntSem}.
\end{definition}

The interpretation of a \lamb-term cannot be just the set of its types as in the case of filter models.
This is related with the fact that \MRel{} is not well-pointed (see Section~\ref{ssec:Pippone}).

We now show that the logical interpretation $\Lint[\cD]{-}$ is equivalent to the categorical one $\Int[\cD]{-}_{-}$ in the sense that they induce the same (in)equalities between \lam-terms.

\begin{theorem}[Semantic Equivalence]\label{thm:type_semantics}
Let $M\in \Lam$ and $\,\FV(M) \subseteq \{x_1, \dots, x_n\}$. Then 
\begin{enumerate}[(i),ref={\roman*}]
\item\label{thm:type_semantics1} 
	$\Lint[\cD]{M} = \Big\{ \big((x_1:a_1,\dots,x_n:a_n),\alpha\big) \in \Env{\cD}\times D\st \big((a_1, \dots, a_n),\alpha\, \big)\in\Int[\cD]{M}_{\seq{x}} \Big\}$,
\item\label{thm:type_semantics2} 
	$\Int[\cD]{M}_{\seq{x}} \,=\, \Big\{ \, \big((\Gamma(x_1) , \dots, \Gamma(x_n)),\alpha\, \big) \in \Mfin{D}^n \times D \st (\Gamma,\alpha)\in \Lint[\cD]{M} \Big\}$.
\end{enumerate}
\end{theorem}

\begin{proof}
\eqref{thm:type_semantics1} It is enough to prove that $\Gamma\vdash M : \alpha\,$ if and only if $\big((\Gamma(x_1), \dots, \Gamma(x_n),\alpha) \big) \in \Int{M}_{\seq{x}}$ and $\supp(\Gamma)\subseteq\{x_1,\dots,x_n\}$. 
We proceed by induction on $M$.

\noindent{\bf Case $M=x_i$.} By Definition~\ref{def:systems} we have $ \Gamma\vdash x_i: \alpha$ if and only if $\Gamma = x_i:[\alpha]$.
Thus, we have $\big((\Gamma(x_1) , \dots, \Gamma(x_n)),\alpha \big) = \big((\emptymset, \dots, \emptymset, [\alpha ], \emptymset, \dots, \emptymset),\alpha \big)   \;\in\; \Int{x_i}_{\seq{x}}
$
by Definition~\ref{def:IntSem}\eqref{def:IntSem1}. 

\noindent{\bf Case $M= \lam y. P$.}  
By Definition~\ref{def:systems} we have that $\Gamma\vdash \lam y . P : \alpha$ holds exactly when $\Gamma, y : a \vdash P: \alpha$. 
By induction hypothesis this is equivalent to ask that $\supp(\Gamma)\subseteq\{x_1,\dots,x_n\}$ and $\big((\Gamma(x_1) , \dots, \Gamma(x_n), a), \alpha \big) \in \Int{P}_{\seq{x},y}$ which holds whenever $\big((\Gamma(x_1), \dots, \Gamma(x_n)),a\to\alpha \big) \in \Int{\lam y.P}_{\seq{x}}$, by Definition~\ref{def:IntSem}\eqref{def:IntSem2}.

\noindent{\bf Case $M=PQ$.} By Definition~\ref{def:systems} we have  $\Gamma\vdash PQ: \alpha$ if and only if $\Gamma_{ 0}\vdash P : [\beta_1,\dots,\beta_k ]\to\alpha$ and $\Gamma_i\vdash Q: \beta_i$ for all $ 1\le i \le k$ where $\Gamma = \sum_{i=0}^n\Gamma_i$ and $\beta_1,\dots,\beta_k \in D$.
By induction hypothesis this is equivalent to require that 
$\supp(\Gamma_i)\subseteq\{x_1,\dots,x_n\}$ for all $0\le i \le k$, 
$\big((\Gamma_{0}(x_1) , \dots, \Gamma_{ 0}(x_n)), [\beta_1,\dots,\beta_k ]\to \alpha\big) \in \Int{P}_{\seq{x}}$
and 
$\big((\Gamma_j(x_1), \dots, \Gamma_j(x_n)), \beta_j \big) \in \Int{Q}_{\seq{x}}$ for all $1\le j \le k$.
By Definition~\ref{def:IntSem}\eqref{def:IntSem3}, this is equivalent to $\big((\Gamma(x_1) , \dots, \Gamma(x_n)),\alpha \big) \in \Int{PQ}_{\seq{x}}$.

\eqref{thm:type_semantics2} This point follows since every element $((a_1,\dots,a_n),\alpha)\in\Mfin{D}^n \times D $ has the form $\big((\Gamma(x_1), \dots, \Gamma(x_n)), \alpha \big)$ for $\Gamma = x_1:a_1,\dots,x_n:a_n$.
\end{proof}  

\begin{corollary}\label{cor:semantic_equiv}
Let $\cD$ be an $\rgm$. Let $M,N\in\Lam$ such that $\,\FV(MN) \subseteq \{x_1, \dots, x_n\}$. Then 
\[
		\cD\models M\sqle N \;\iff\;\Int[\cD]{M}_{\seq x} \subseteq \Int[\cD]{N}_{\seq x}\;\iff\;\Lint[\cD]{M} \subseteq \Lint[\cD]{N}.
\]
\end{corollary}

\subsection{On the choice of the interpretation.}\label{ssec:Pippone} 
The categorical interpretation of a \lam-term~$M$ in a reflexive object $\cD$ gives a morphism $\Int{M}_{\seq x} : D^{\seq x}\to D$ such that $\Th(\cD)$ is a \lam-theory.
If the category is well-pointed, like Scott's continuous semantics, then it is equivalent to interpret $M$ as a \emph{point} of $D$ through a \emph{valuation} $\rho: \Var\to D$, namely $\Int{M}_\rho: \Termobj \to D$~\cite[\S5.5]{Bare}.\linebreak
For this reason, in the context of graph and filter models, it became standard to consider the interpretation of $M$ as an element of the domain and, when presented like a type system, as the set of its types~\cite{Berline00,CoppoDZ87,Ronchi82,RonchiP04}.
As shown by Koymans in~\cite{Koymans82}, when the category is not well-pointed, points are no more suitable for interpreting \lam-terms since the induced equality is not a \lam-theory because of the failure of the $\xi$-rule~\cite{Selinger02}. 
In the algebraic terminology, the set of points gives a \emph{\lam-algebra} which is not a \emph{\lam-model}~\cite[\S5.2]{Bare}. 
In~\cite{BucciarelliEM07}, Bucciarelli \emph{et al.} show that a \lam-model can however be constructed from a reflexive object $\cD$ of a non-well pointed category, by considering the set $\catC_\mathrm{f}(D^\Var,D)$ of ``finitary'' morphisms from $D^\Var$ to $D$ and valuations $\rho : \Var \to \catC_\mathrm{f}(D^\Var,D)$. 
For instance, this is the approach followed in~\cite{PaoliniPR15}.
However, in~\cite{ManzonettoTh}, the author remarks that the use of valuations in this context becomes redundant since $\forall\rho.\Int{M}_\rho = \Int{N}_\rho$ holds exactly when they are equal under the valuation $x\mapsto \pi^\Var_x$ sending $x$ to the corresponding projection.
By applying this fact to the logical interpretation given in~\cite{PaoliniPR15}, we recover Definition~\ref{def:interpretazione} and this justifies Theorem~\ref{thm:type_semantics} from a broader perspective.
See \cite{BucciarelliEM07,Selinger02,ManzonettoTh} for more detailed discussions on well-pointedness.

\subsection{Quantitative Properties}\label{subsec:QP}
We show some quantitative properties satisfied by the type systems issued from a relational graph model.
This quantitative behavior was first noticed by de~Carvalho while studying the relational model $\cE$ and linear head reduction~\cite{deCarvalho09}.
Our statements are rather a more refined version\footnote{
Indeed, the authors of~\cite{PaoliniPR15} just consider the head reduction strategy and take as measure the size of the whole derivation tree.
We show that it is the number of application rules that actually decreases along head reduction and we provide a measure that decreases whenever any occurrence of a $\beta$-reduction is contracted.
} of the ones appearing in~\cite{PaoliniPR15}.

Until the end of the section, the symbol $\,\vdash\,$ refers to any fixed relational graph model.
We introduce some auxiliary notions that will be useful in the subsequent proofs.

\begin{definition} Let $M,N\in\Lamb$ and $\cD$ be an rgm.
\begin{itemize}
\item 
	Given $\pi_{ 0}\derof\Gamma\vdash M : [\beta_1,\dots,\beta_n]\to \alpha$ and $\pi_i \derof \Delta_i\vdash N : \beta_i$ for all $1 \le i\le n$,  we let $\Apprule{\pi_{ 0}}{\{\pi_i\}_{i=0}^n}$ be the derivation tree of $\Gamma + (\sum_{i=1}^n \Delta_i)\vdash MN : \alpha$ obtained by applying the rule \texttt{app} to those premises.  
\item Similarly, given $\pi\derof\Gamma, x : a\vdash M : \beta$ we define $\Lamrule{x}{\pi}$ as the derivation obtained by applying the rule \texttt{lam} to the derivation $\pi$.
\end{itemize}
\end{definition}
Let $\app(\pi)$ be the number of instances of the rule $\mathtt{app}$ that occur in the derivation~$\pi$.

\setcounter{sec}{\value{section}}
\setcounter{temp}{\value{thm}}
\begin{lemma}[Weighted Substitution Lemma]\label{lemma:ws}
Let $M,N\in\Lambda_\bot$.
Consider some derivations $\pi_{\scriptscriptstyle 0}\derof \Gamma_{ 0}, x : [\beta_1,\dots, \beta_n]\,\vdash M : \alpha\,$ for  $n\in\nat$ and 
$\pi_i\derof \Gamma_i\vdash N : \beta_i\,$ for all $1\le i\le n$. Then there exists 
$\pi\derof \sum_{i=0}^n\Gamma_i\vdash M\subst{x}{N} : \alpha\;$ such that    $\;\app(\pi) = \sum_{i=0}^n \app(\pi_i)$.
\end{lemma}

\begin{proof} We proceed by structural induction on $M$.
\smallskip

{\bf Case $M=\bot$.} This case is vacuous, as $\bot$ cannot be typed.

\smallskip
{\bf Case $M=y\neq x$.} 
Then $\pi_0\derof\Gamma_{ 0}, x : [\beta_1,\dots, \beta_n] \vdash y:\alpha$ entails $n=0$ and $\Gamma_{ 0} = y:\alpha$. Hence the judgment $\,\sum_{i=0}^{n}\Gamma_i \vdash M\subst{x}{N} : \alpha$ is nothing but $y:\alpha \vdash y: \alpha$ and we can take $\pi = \pi_0$.
Clearly we have that $\app(\pi) = 0 = \app(\pi_{ 0}) = \sum_{i=0}^n \app(\pi_i).$

\smallskip

{\bf Case $M=x$.} 
Then $\Gamma_{ 0}, x : [\beta_1,\dots, \beta_n] \vdash x:\alpha$ implies that $n=1$, $\beta_1 = \alpha$ and $\Gamma_{ 0}$ is empty. 
Hence the judgment $\,\sum_{i=0}^{n}\Gamma_i \vdash M\subst{x}{N} : \alpha$ is just $\Gamma_1 \vdash N : \alpha$ and we can take $\pi = \pi_1$. 
Therefore we have  $\app(\pi) = \app(\pi_{ 1}) = \sum_{i=0}^n \app(\pi_i).$

\smallskip
{\bf Case $M=\lam y.P$.} Then there is a derivation $\pi'_{ 0}$ such that $\pi_{ 0}$ has the form 
\[
\infer{\qquad\Gamma_{ 0},\,  x:[\beta_1, \dots, \beta_n]  \vdash \lam y.M :[\alpha_1, \dots, \alpha_n]\to \alpha' \qquad }{ {\pi'_{ 0}}\derof\;\Gamma_{ 0},\, y:[\alpha_1, \dots, \alpha_n]  ,\, x:[\beta_1, \dots, \beta_n]\vdash M :\alpha' }
\]
for $\alpha = [\alpha_1, \dots, \alpha_n]\to \alpha'$. 
Notice that $\app(\pi'_{ 0}) = \app(\pi_{ 0})$. 
By induction hypothesis, there exists a derivation $\pi'$ such that
\begin{equation}\label{eq:dedylam}
  \pi' \,\derof\; \big(\Gamma_{ 0} ,\, y:[\alpha_1, \dots, \alpha_n] \big)+\sum_{i=1}^{n}\Gamma_i  \vdash M\{N/x\} :\alpha'
\end{equation}
with
$\app(\pi') = \app(\pi'_{ 0}) +\sum_{i=1}^{n}\app(\pi) = \sum_{i=0}^{n}\app(\pi)$.
By Lemma~\ref{lem:env_fv} we have $\supp(\Gamma_i) \subseteq \FV(N)$ for all $1\le i\le n$.  
By $\alpha$-convertion we assume $y\notin\FV(N)$, thus $y\notin \supp(\Gamma_i)$ for all $i$. 
So the judgment in \eqref{eq:dedylam} is in fact ${ \sum_{i=0}^{n}\Gamma_i ,\, y:[\alpha_1, \dots, \alpha_n] \vdash M\{N/x\} :\alpha' }$. 
We can then take $\pi = \Lamrule{y}{\pi'} \derof { \sum_{i=0}^{n}\Gamma_i \vdash \lam y. M\{N/x\} : [\alpha_1, \dots, \alpha_n] \to\alpha' }$. The thesis is proved since $\lam y. M\{N/x\} = (\lam y. M)\{N/x\}$ and 
$
\app(\pi) = \app(\pi') = \sum_{i=0}^{n}\app(\pi) .
$

\smallskip
{\bf Case $M=PQ$.} 
Then there are derivations $\pi_{00},\pi_{0i}$ such that $\pi_{ 0}$ has the form
\[
\infer{  \sum_{i=0}^k\Gamma_{{ 0}i}  ,\, x:[\beta_1, \dots, \beta_n] \vdash P\, Q  :\alpha }{{\pi_{ 00} \derof \Gamma_{{ 00}} ,\, x: [\beta_j]_{j\in I_0}  \vdash P : [\gamma_1, \dots, \gamma_k]\to \alpha \;\;\;}{\pi_{0i} \derof\Gamma_{{ 0}i},\, x: [\beta_j]_{j\in I_i} \vdash Q :\gamma_i \,\mbox{ for all }1\le i \le k }}
\]
where $k\in\nat$, $ \Gamma_{0} = \sum_{i=0}^k\Gamma_{{ 0}i} $ and $\{I_i\}_{i=0}^k$  is a partition of the set $\{1,\dots,n\}$.
By induction hypothesis we get a derivation $\pi'_{ 0} \derof \Gamma_{ 00} +\sum_{j\in I_0}{\Gamma_j}\vdash P\,\{N/x\} : [\gamma_1, \dots, \gamma_k]\to \alpha$ such that 
$\app(\pi'_{ 0} ) = \app(\pi_{ 00}) +\sum_{j\in I_0}\app(\pi_j) $. Also, for all $1 \le i\le k$ the induction hypothesis provides a derivation
 $\pi'_{ i} \derof \Gamma_{{ 0}i}+\sum_{j\in I_i}{\Gamma_j}\vdash Q\,\{N/x\} : \gamma_i$ such that  
$\app(\pi'_{ i}) = \app(\pi_{0i}) +\sum_{j\in I_i}\app(\pi_j) .$

We take 
$\pi = \Apprule{
	\pi'_{ 0}
	}{
	\{ \pi'_{i} \}_{1 \le i\le k}
	}\derof  
	\sum_{i=0}^k \big(\Gamma_{{ 0}i} + \sum_{j\in I_i}{\Gamma_j}\big)
	\vdash (P\subst{x}{N})(Q\subst{x}{N}) : \alpha.$
Clearly $\sum_{i=0}^k \big(\Gamma_{{ 0}i} + \sum_{j\in I_i}{\Gamma_j}\big) = 
\sum_{i=0}^k\Gamma_{{ 0}i} +\sum_{i=0}^k\sum_{j\in I_i}{\Gamma_j} = 
\Gamma_{ 0}+\sum_{i=1}^n{\Gamma_i} = \sum_{i=0}^n{\Gamma_i} $ 
and $(P\subst{x}{N})(Q\subst{x}{N}) = (PQ)\subst{x}{N}$. 
Moreover we get that $\app(\pi) = 1 + \sum_{i=0}^k \app(\pi'_i) = 1 + \sum_{i=0}^k \big(\app(\pi_{0i}) + \sum_{j\in I_i}\app(\pi_j)\big) = 
\big(1 + \sum_{i=0}^k \app(\pi_{0i})\big) + \sum_{i=0}^k \sum_{j\in I_i} \app(\pi_j) = 
\app(\pi_0) + \sum_{i=1}^n \app(\pi_i) = \sum_{i=0}^n\app(\pi_i) $,
which concludes the proof.
\end{proof}

From this, it follows that the number of rules $\texttt{app}$ in the derivation of a $\beta$-redex, decrements in a derivation of its contractum.

\begin{corollary}\label{cor:base_per_dopo}
Let $M,N\in\Lambda_\bot$.
If $\pi\derof \Gamma \vdash (\lam x.M)N : \alpha$ there exists a derivation $\pi'$  such that $\pi'\derof \Gamma \vdash M\subst{x}{N} : \alpha$ with $\app(\pi') = \app(\pi)-1$.
\end{corollary}

\begin{proof}
The derivation $\pi$ has the form
\[
\infer[\texttt{app}]{  
	\sum_{i=0}^n\Gamma_i \vdash (\lam x.M)N :\alpha 
	}{
	\infer[\texttt{lam}]{
		\qquad\Gamma_{ 0} \vdash \lam x.M :[\beta_1, \dots, \beta_n]\to \alpha \qquad 
		}{
		\pi_{\scriptscriptstyle 0} \derof \Gamma_{ 0}, x:[\beta_1, \dots, \beta_n] \vdash M :\alpha 
		}
		&
		\pi_i \derof\Gamma_i \vdash N :\beta_i \,\mbox{ for }1\le i \le n 
	}
\]
 where $n\in\nat$ and $\Gamma = \sum_{i=0}^n\Gamma_i$.
By Lemma~\ref{lemma:ws} there exists $\pi'\derof \sum_{i=0}^n\Gamma_i\vdash M\subst{x}{N} : \alpha\;$ such that  $\app(\pi') = \sum_{i=0}^n \app(\pi_i) = \app(\pi)-1$.
\end{proof}

From Corollary~\ref{cor:base_per_dopo} it follows that the number of $\texttt{app}$ decreases exactly by $1$ at each step of head reduction.
So $\app(-)$ provides an upper bound for the number of steps necessary to get the principal hnf of a solvable term, as observed by de Carvalho in~\cite{deCarvalho09}.

\begin{lemma}\label{lemma:hr-decreases} 
If $\pi\derof\Gamma\vdash M : \alpha$ then the head reduction of $M$ has length at most $\app(\pi)$.
\end{lemma}

Unfortunately, the measure $\app(-)$ is not enough for proving the approximation theorem. 
The reason is that, in order to compute the B\"ohm tree of a \lam-term $M$, one needs to reduce also redexes that are not in head-position.
E.g., consider the following derivation~$\pi$:
\[  
	\infer[\texttt{app}]{ 
	x: [\emptymset \to \alpha] \vdash x \,(\bI \,y) :\alpha
	}{
	x: [\emptymset \to \alpha] \vdash x: \emptymset \to \alpha
	} 
\]
When reducing $x (\bI \,y) \redto[\beta] xy$, the only possible derivation $\pi'$ of $x: [\emptymset \to \alpha] \vdash xy :\alpha$ is:
\[  
	\infer[\texttt{app}]{
	x: [\emptymset \to \alpha] \vdash xy :\alpha
	}{
	x: [\emptymset \to \alpha] \vdash x: \emptymset \to \alpha
	} 
\]
The number of \texttt{app} has not decreased in this case, so we cannot use $\app(-)$ as a decreasing measure in the proof of the left-to-right implication of Theorem~\ref{thm:app}.  
We can however find an approximant $t\in \BT[*]{x\,(\bI \,y)}$ with that typing by realizing that in $\pi$ the subterm $\bI \,y$ is basically used as the approximant $\bot$.  
This is the key idea behind the next lemma.

Given $M\in\Lam_\bot$ and a redex occurrence $R=(\lam x.P)Q$ in $M$, say $M = C\hole R$ for some single hole $\lam\bot$-context $C\hole-$, we denote by $M\subst{R}{\bot}$ the \lamb-term $C\hole \bot$.
Moreover, we write $M \stackrel{\scriptscriptstyle R}{\to}_\beta M'$ to indicate that the contracted $\beta$-redex is $R$, namely that $M' = C\hole{P\subst{x}{Q}}$.

\begin{lemma}[Weighted Subject Reduction]\label{lem:wsRed}
Let $M,M'\in\Lam_\bot$ be such that $M \stackrel{\scriptscriptstyle R}{\to}_\beta M'$. 
If $\pi \derof \Gamma \vdash M: \alpha$, then there is $\pi' \derof \Gamma \vdash M': \alpha$ such that one of the following cases holds:
\begin{enumerate}[(i),ref={\roman*}]
\item\label{lem:splitRed1} 
	$\app(\pi') < \app(\pi)$,
\item\label{lem:splitRed2} 
	$\pi' \justlike \pi$ and there exists $\pi'' \derof \Gamma \vdash M\subst{R}{\bot} : \alpha$ such that $\pi'' \justlike \pi$.
\end{enumerate}
\end{lemma}
\begin{proof}
Let $M=C\hole R$ for a single hole $\lam\bot$-context $C\hole -$. We proceed by induction on  $C\hole -$.

{\bf Case $\hole -$.} We have $M = R = (\lam x.P) Q$ and $M' = P\{Q/x\}$. 
The thesis is then given by Corollary~\ref{cor:base_per_dopo}. 
More precisely, we are in case~\eqref{lem:splitRed1}.

{\bf Case $P\, (C \hole -)$.} We have $M = P(C \hole R)$ and $M'= P(C\hole {R'})$ where $R \to_\beta R'$. 
Then, the derivation $\pi$ has the form
\[
\infer{
	\sum_{i=0}^n \Gamma_i \vdash P(C \hole R) : \alpha 
	}{
	\pi_{ 0} \derof \Gamma_{ 0} \vdash P :[\beta_1, \dots, \beta_n]\to \alpha
	&
	\pi_i \derof\Gamma_i \vdash C\hole R :\beta_i \,\mbox{ for }1\le i \le n 
	}
\]
where $n\in\nat$ and $\Gamma = \sum_{i=0}^n\Gamma_i$. 

For all $1\le i \le n\,$ by induction hypothesis there is $\pi'_i \derof \Gamma_i \vdash C\hole {R' } : \beta_i$ such that either 
\begin{equation}\label{eq:primaequazza}
 \app(\pi'_{ i}) < \app(\pi_{ i}),
\end{equation} 
or 
\begin{equation}\label{eq:secondaequazza}
  \pi_{ i}' \justlike \pi_{ i} \justlike \pi_{ i}'' \mbox{\; for a derivation \;} \pi_{ i}'' \derof \Gamma_i \vdash C \hole \bot: \beta_i. 
\end{equation}
In this case we take $\pi' = \Apprule{\pi_{ 0}}{\{\pi_i' \}_{i =1}^n}$. 
 
If every $i$ satisfies \eqref{eq:secondaequazza} then $\pi' \justlike \pi$. 
By taking $\pi'' = \Apprule{\pi_{ 0}} \{\pi''_i\}_{i =1}^n $ we obtain the case~\eqref{lem:splitRed2} of the thesis. Notice that the eventuality $n = 0$ falls in this case.

If there is an $i$ that satisfies \eqref{eq:primaequazza} then  $\app(\pi') < \app(\pi)$, so the case~\eqref{lem:splitRed1} is proved.


\vspace{0.13cm}
{\bf Case $(C\hole -) P$.} We have $M = (C\hole R) P$ and $M'= (C\hole {R'}) P$ where $R \to_\beta R'$. 
Then, the derivation $\pi$ has the form
\[
\infer{  \sum_{i=0}^n\Gamma_i \vdash (C\hole R) P :\alpha }{{\pi_{ 0} \derof \Gamma_{ 0} \vdash C\hole R :[\beta_1, \dots, \beta_n]\to \alpha \qquad }{\pi_i \derof\Gamma_i \vdash P :\beta_i \,\mbox{ for all }1\le i \le n }}
\]
 where $n\in\nat$ and $\Gamma = \sum_{i=0}^n\Gamma_i$. 
 By induction hypothesis there exists a derivation $\pi'_0 \derof \Gamma_{ 0} \vdash C\hole {R' } :[\beta_1, \dots, \beta_n]\to\alpha$ such that either $\app(\pi_{ 0}') < \app(\pi_{ 0})$, or  $\pi_{ 0}' \justlike \pi_{ 0} \justlike{\pi_{ 0}''}$ for some $\pi_{ 0}'' \derof \Gamma_{ 0} \vdash C \hole \bot: [\beta_1, \dots, \beta_n] \to~\alpha$. 
In the former case, the thesis is proved taking $\pi' = \Apprule{\pi_{ 0}'}{\{\pi_i\}_{i =1}^n}$ and in the latter also $\pi'' \!=  \Apprule{\pi_{ 0}''}{\{\pi_i\}_{i =1}^n}$.

\vspace{0.17cm}
{\bf Case $\lam x. C\hole -$.} We have  $M = \lam x. C\hole R$ and $M' = \lam x. C\hole {R'}$, where $R \to_\beta R'$. Then $\pi$ has the form 
\[
{\infer{\qquad\Gamma \vdash \lam x.C\hole R :[\beta_1, \dots, \beta_n]\to \beta \qquad }{\pi_{ 0} \derof \Gamma, x:[\beta_1, \dots, \beta_n] \vdash C\hole R :\beta }}
\]
for  $[\beta_1, \dots, \beta_n]\to \beta = \alpha$. 
By induction hypothesis there is $\pi'_0 \derof \Gamma, x: [\beta_1, \dots, \beta_n] \vdash C\hole {R' } :\beta$ such that either $\app(\pi_{ 0}') < \app(\pi_{ 0})$, or $\pi_{ 0}' \justlike \pi_{ 0}\justlike \pi_{ 0}''$ for some $\pi_{ 0}'' \derof \Gamma, x :[\beta_1, \dots, \beta_n] \vdash C \hole \bot: \beta$. 
In the former case, the thesis is proved by taking $\pi' = \Lamrule{x}{\pi_{ 0}'}$, and in the latter also $\pi'' = \Lamrule{x}{\pi_{ 0}''}$.
\end{proof}

As a note aside, Lemma~\ref{lem:wsRed} gives in particular the subject reduction property. 
The subject expansion can be proved equally easily, as done in~\cite[\S 2.4]{phdruoppolo}. 
These two properties provide yet another soundness proof for relational graph models.






\subsection{The Approximation Theorem}\label{subsec:AT}
We show that all relational graph models satisfy the Approximation Theorem stating that the interpretation of a \lam-term is given by the union of the interpretations of its finite approximants (Theorem~\ref{thm:app}).
As mentioned in the introduction, we provide a new combinatorial proof that does not exploits reducibility candidates nor Ehrhard's notion of Taylor expansion.
Actually, it is an easy consequence of our Weighted Subject Reduction (Lemma~\ref{lem:wsRed}). 

From the Approximation Theorem, we get that the \lam-theory induced by any relational graph model $\cD$ includes $\BTth$ (Corollary~\ref{cor:BTminimalth}) and, if $\cD$ is extensional, also $\Hpl$ (Corollary~\ref{cor:extensionality}).

\begin{lemma}\label{lem:typomega}
Let $M\in\Lam_\bot$. If $M$ is in $\beta$-normal form and $\, \Gamma \vdash M: \alpha$ then $\, \Gamma \vdash \da{M}: \alpha$.
\end{lemma}
\begin{proof}
By a straightforward induction on $M$.
\end{proof}

Notice that the hypothesis that $M$ is in $\beta$-normal form is necessary to prove Lemma~\ref{lem:typomega}.
Indeed, for $M = \bI\, x$ we have $x:[\alpha] \vdash M:\alpha$, whereas $\da{M} = \bot$ cannot be typed.

Given $M\in\Lam_\bot$ we denote by $\rednum{M}$ the number of occurrences of $\beta$-redexes in $M$.


\begin{theorem}[Approximation Theorem]\label{thm:app} 
Let $M\in\Lam_{\bot}$. Then $( \Gamma,\alpha)\in\Lint{M} $ if and only if there exists $t\in\BT[*]{M}$ such that $( \Gamma,\alpha) \in \Lint{t}$.
Therefore $\Lint{M} = \bigcup_{t\in\BT[*]{M}} \Lint{t}$.
\end{theorem}
\begin{proof}
$(\Rightarrow)$
Let $\,\pi \derof \Gamma \vdash M: \alpha$. We proceed by induction on the pair $\big(\app(\pi),\rednum{M}\big)$ lexicographically ordered. 

{\bf Case $\rednum{M} = 0$.}  By Lemma~\ref{lem:typomega}  $\, \Gamma \vdash \da{M}: \alpha$ and clearly $\da{M}\in \BT[*]{M}$.

{\bf Case $\rednum{M} > 0$.} Let $R$ be any occurrence of a $\beta$-redex in $M$ and $M \stackrel{\scriptscriptstyle R\,}{\to}_\beta M'$. 
By Lemma~\ref{lem:wsRed}  there is a derivation $\pi' \derof \Gamma \vdash M': \alpha$ such that either (i) $\app(\pi') < \app(\pi)$, or (ii) there exists $\pi'' \derof \Gamma \vdash M\{\bot/R\}: \alpha$ such that $\pi'' \justlike \pi$.

In Case (i) we  apply the induction hypothesis to $\pi'$ and get $t\in\BT[*]{M'} = \BT[*]{M}$ such that $\Gamma \vdash t: \alpha$.

In Case (ii) we can apply the induction hypothesis to $\pi''$, as $\pi''\justlike\pi$ implies $\app(\pi'') = \app(\pi)$, and moreover $\rednum{M\{\bot/R\}} < \rednum{M}$. We get $t\in\BT[*]{M\{\bot/R\}} $ such that $\Gamma \vdash t: \alpha$. 
Since moreover $\BT[*]{M\{\bot/R\}} \subseteq \BT[*]{M}$ by Lemma~\ref{lemma:xDomenico}, we are done.

$(\Leftarrow)$ We proceed by induction on $t$. 
The case $t=\bot$ is vacuous, as $\bot$ is not typable. 
So let $t = \lam x_1 \dots x_n.xt_1 \cdots t_m$ where $n,m\in\nat$. 
We suppose that the variable $x$ is bound, the other case being analogous.
The given derivation tree of $\Gamma \vdash t:\alpha$ must have the form\footnote{The ``double line'' in the derivation tree is a shortcut to indicate the simultaneous application of zero, one, or many rules of the same kind.}
\[
\infer={
	\Gamma\vdash  \lam x_1 \dots x_n.x\,t_1 \cdots t_m : a_1\to\dots \to a_n\to\beta
	}
	{
	\infer={
		\Delta+\sum_{i=1}^m\sum_{j=1}^{k_i}\Gamma_{ij}, x_1:\sum_{i=1}^m\sum_{j=1}^{k_i} a_1^{ij},\dots, x_n:\sum_{i=1}^m\sum_{j=1}^{k_i} a_n^{ij} \vdash xt_1 \cdots t_m : \beta
		}{
		\Delta \vdash x: b_1 \to \cdots  \to b_m  \to \beta 
		&
		\Gamma_{ij}, x_1:a_1^{ij},\dots, x_n:a_n^{ij} \vdash t_i : \beta_{ij}
		& 
		\textrm{for all }i,j
		}
	}
\]
where for all $i \le m$ we have $b_i = [\beta_{i1}, \dots, \beta_{i{k_i}}]$ for some $k_i\in\nat$; 
$\Delta = x:b_1  \to \cdots  \to b_m  \to \beta$;
$\Gamma = \Delta + \sum_{i=1}^m\sum_{j=1}^{k_i}\Gamma_{ij}$; 
$a_\ell = \sum_{i=1}^m\sum_{j=1}^{k_i} a_\ell^{ij}$ for all $\ell \le n$; 
finally, $\alpha =a_1\to\dots \to a_n\to\beta$.
As $t\in\BT[*]{M}$ we have that $t = \da{N}$ for some $N=_\beta M$. 
By Definition~\ref{def:omegastuff}, we have $N = \lam x_1\dots x_n.x\,N_1 \cdots N_m$ with $t_i = \da{N_i}$ for all $i \le m$. 
By induction hypothesis we get $\Gamma_{ij}, x_1:a_1^{ij},\dots, x_n:a_n^{ij} \vdash N_i : \beta_{ij}$ for all $i,j$. 
By replacing each $t_i$ by $N_i$ in the proof tree above, we get a derivation of $\Gamma \vdash N: \alpha$. 
Since $N =_\beta M$, by soundness we get $\,\Gamma \vdash M: \alpha$.
\end{proof}



The following result first appeared in print in~\cite{ManzonettoR14}, but was already known in the folklore (see the discussion in~\cite{CarraroES10}). Notice that it is stronger than Theorem~6 in~\cite{PaoliniPR15} which only shows that the \lam-theories induced by strongly linear relational models are sensible.
\begin{corollary}\label{cor:BTminimalth} For all $\rgm$'s $\cD$ we have that $\BTth\subseteq\Th(\cD)$.
In particular $\Th(\cD)$ is sensible and $\Lint[\cD]{M} = \emptyset$ for all unsolvable \lam-terms $M$.
\end{corollary}
\begin{proof} From Theorem~\ref{thm:app} we get $\Lint{M} = \Lint{\BT{M}} = \bigcup_{t\in\BT[*]{M}} \Lint{t}$.
Therefore, whenever $\BT{M} =\BT{N}$ we have $\Lint{M} = \Lint{\BT{M}} = \Lint{\BT{N}} = \Lint{N}$. 
Thus $\BTth\subseteq\Th(\cD)$.
\end{proof}

In the next corollary we are going to use the L\'evy's characterization of Morris's inequational theory $\sqle_\Hpl$ in terms of extensional approximants (Theorem~\ref{thm:JJ}). The extensional approximants of a \lam-term $M$ are interpreted as usual by setting $\Lint{\BT[e]{M}} = \bigcup_{t\in\BT[e]{M}} \Lint{t}$.

\begin{corollary}\label{cor:extensionality} 
For an rgm $\cD$, the following are equivalent:
\begin{enumerate}[(i),ref={\roman*}]
\item\label{cor:extensionality1}  
	$\cD$ is extensional,
\item\label{cor:extensionality2} 
	$\sqle_\Hpl\subseteq\Thle(\cD)$,
\item\label{cor:extensionality3} 
	$\BTeth\subseteq\Th(\cD)$.
\end{enumerate}
\end{corollary}
\begin{proof} 
$(\ref{cor:extensionality1}\To \ref{cor:extensionality2})$
From Theorem~\ref{thm:app} we obtain $\Lint{M} = \bigcup_{t\in\BT[*]{M}}\Lint{t}$. 
From the extensionality~of~$\cD$, we get $\Lint{M} = \bigcup_{M'\msto[\eta] M,\, t\in\BT[*]{M'}}\Lint{t} = 
\bigcup_{M'\msto[\eta] M,\, t\in\BT[*]{M'}}\Lint{\nf[\eta](t)} =\Lint{\BT[e]{M}}$.
So, we have that $\BT[e]{M} \subseteq\BT[e]{N}$ entails $\Lint{M} = \Lint{\BT[e]{M}} \subseteq \Lint{\BT[e]{N}} = \Lint{N}$. 

\noindent $(\ref{cor:extensionality2}\To \ref{cor:extensionality3})$ Trivial.

\noindent $(\ref{cor:extensionality3}\To \ref{cor:extensionality1})$ By Lemma~\ref{lemma:D_ext_actually_ext}.
\end{proof}


%% file: include/fig4.tex
\begin{figure}[t!]
\begin{center}
    \AxiomC{$\vphantom{M\rightarrow_\mu}$}
     \RightLabel{\texttt{var}}    
    \UnaryInfC{$x : [\alpha]\vdash_{\scriptscriptstyle \cD} x : \alpha$}
    \DisplayProof \hspace{2em}
    \AxiomC{$\Gamma, x : a\vdash_{\scriptscriptstyle \cD} M : \alpha$}
     \RightLabel{\texttt{lam}}        
    \UnaryInfC{$\Gamma\vdash_{\scriptscriptstyle \cD} \lambda x.M : a\to\alpha$}
    \DisplayProof \hspace{2em}
    %
    %
    
   \bigskip
  
    \AxiomC{$\Gamma\vdash_{\scriptscriptstyle \cD} M : [\beta_1,\dots,\beta_n]\to \alpha$}
    \AxiomC{$\Delta_i\vdash_{\scriptscriptstyle \cD} N : \beta_i$}
    \AxiomC{$1 \le i\le n$}
     \RightLabel{\texttt{app}}
    \TrinaryInfC{$\Gamma + (\sum_{i=1}^n \Delta_i)\vdash_{\scriptscriptstyle \cD} MN : \alpha$}
    \DisplayProof    
\end{center}
\caption{The typing rules of the type system associated with a relational graph model~$\cD$.}\label{fig:Typing}
\end{figure}

%% file: include/minimalth.tex

In this section we show that a minimal inequational graph theory exists, and that it is exactly the inequational theory induced by the model $\cE$ defined in Example~\ref{ex:graphrel}.

\subsection{The Minimal Inequational Graph Theory}

We start by defining an inequational theory $\ler$ and prove that it is included in $\Thle(\cD)$ for every relational graph model $\cD$.

\begin{definition}\label{def:ler}
Given $M,N\in\Lam$, we let $M\ler N$ whenever there exists a B\"ohm-like tree $U$ such that $\BT{M}\BTle U\geeta[\infty] \BT{N}$.
\end{definition}

The fact that this definition involves $\eta$-expansions and that $\ler\, \subseteq\Thle(\cD)$ holds also for non-extensional relational graph models should not be surprising. 
Indeed, when considering the original graph model $\mathscr{P}_\omega$ defined by Plotkin~\cite{Plotkin93} and Scott~\cite{Scott76} the situation is actually analogous (but symmetrical\footnote{
I.e.\ $\mathscr{P}_\omega\models M\sqle N$ exactly when there is a B\"ohm-like tree $U$ such that $\BT{M} \leeta[\infty] U\BTle \BT{N}$ \cite{Hyland76}.
}), and no graph model is extensional.

\begin{example} \ 
\begin{enumerate}[(1)]
\item $\bJ\ler \bI$ while $\bI\not\ler \bJ$,
\item Let $(M)_{n\in\nat}$ be the effective sequence defined by $M_n = \bJ$ if $n$ is even and $M_n = \bO$ otherwise.
Then we have $[M_n]_{n\in\nat}\ler [\bI]_{n\in\nat}$.
\end{enumerate}
\end{example}

It follows from \cite{Ronchi82} that $M\ler N\ler M$ if and only if $\BTth\vdash M = N$.
The inequational theory $\ler$ admits the following characterization in terms of B\"ohm tree approximants.

\begin{lemma}\label{lemma:char_ler_app} 
Let $M,N\in\Lam$. We have $M\ler N$ if and only if for all $t\in\BT[*]{M}$  there exists $s\in\BT[*]{N}$ such that $t\BTle u\msto[\eta] s$ for some $u\in\App$.
\end{lemma}

\begin{proof}$(\Rightarrow)$ By structural induction on $t$.
If $t = \bot$, then we can take $s = u = \da{M}$.

Otherwise $t\neq \bot$ entails that $M$ has a hnf $M=_\beta \lam\seq x z_1\dots z_m.x_iM_1\cdots M_kP_1\cdots P_m$. 
By Definition~\ref{def:ler}, there is a B\"ohm-like tree $U = \lam\seq x z_1\dots z_m.x_iU_1\cdots U_kV_1\cdots V_m$ such that $\BT{M_j}\BTle U_j$ and $\BT{P_\ell}\BTle V_\ell\geeta[\infty]z_\ell$, and  $N =_\beta \lam\seq x.x_iN_1\cdots N_k$ with $U_j\geeta[\infty] \BT{N_j}$.
Thus, we have $t = \lam\seq xz_1\dots z_m.x_it_{1}\cdots t_{k}t'_1\cdots t'_m$ for $t_j\in\BT[*]{M_j}$ and $t'_\ell\in\BT[*]{P_\ell}$.
Since for all $j\le k$ we have $M_j\ler N_j$ and for all $\ell\le m$ we have $P_\ell\ler z_\ell$ we can apply the induction hypothesis and get $t_j\BTle u_j\msto[\eta]s_j$ for some $s_j\in\BT[*]{N_j}$ and $t'_\ell\BTle u'_\ell\msto[\eta] z_\ell$. 
As a consequence $t\BTle \lam\seq x z_1 \dots z_m.x_iu_1\cdots u_ku'_1\cdots u'_m\msto[\eta] \lam\seq x.x_is_1\cdots s_k\in\BT[*]{N}$.

$(\Leftarrow)$ We prove $M\ler N$ coinductively. 
If $M$ is unsolvable, then we are done.

Otherwise there is $t\in\BT[*]{M}$ of the form $t = \lam\seq x z_1\dots z_m.x_it_1\cdots t_kt'_1\cdots t'_m$ and $s = \lam\seq x.x_is_1\cdots s_k\in\BT[*]{N}$ such that $t\BTle \lam\seq xz_1\dots z_m.x_iu_1\cdots u_ku'_1\cdots u'_m\msto[\eta] s$. 
This entails that $M =_\beta \lam\seq xz_1\dots z_m.x_iM_1\cdots M_kP_1\cdots P_m$ and $N =_\beta \lam\seq x.x_iN_1\cdots N_k$ with $t_j\in\BT[*]{M_j}$, $s_j\in\BT[*]{N_j}$ for all $j\le k$, $z_\ell\notin\FV(\BT{x_i\seq M\seq N})$, and $t'_\ell\in\BT[*]{P_\ell}$ for all $\ell\le m$.
By coinductive hypothesis $M_j\ler N_j$, i.e., there are B\"ohm-like trees $U_j$ such that $\BT{M_j}\BTle U_j\geeta[\infty]\BT{N_j}$.
Since for all $t'_\ell\in\BT[*]{P_\ell}$ there is an $\eta$-expansion $u'_\ell$ of $z_\ell$ such that $t'_\ell\BTle u'_\ell$, then there exists a B\"ohm-like tree $V_\ell$ such that $\BT{P_\ell}\BTle V_\ell\geeta[\infty] z_\ell$. 
As a consequence $\BT{M}\BTle  \lam\seq x\seq z.x_iU_1\cdots U_kV_1\cdots V_m\geeta[\infty] \BT{N}$, which shows $M\ler N$.
\end{proof}

\begin{proposition}\label{prop:lerminimal}
Let $M,N\in\Lam$. If $M\ler N$ then $\cD \models M \sqle N$ for every rgm $\cD$.
\end{proposition}

\begin{proof} By Lemma~\ref{lemma:char_ler_app} for all $t\in\BT[*]{M}$ there are $s\in\BT[*]{N}$ and $u\in\App$ such that $t\BTle u\msto[\eta] s$.
Since $t\BTle u$ we get $t\in\BT[*]{u}$ and, by Theorem~\ref{thm:app}, we obtain $\Lint{t}\subseteq\Lint{u}$. 
From Lemma~\ref{lemma:D_ext_actually_ext}(\ref{lemma:D_ext_actually_ext1}) and Corollary~\ref{cor:semantic_equiv} we have that $\Lint{u}\subseteq\Lint{s}$ holds. 
It follows that $\Lint{\BT{M}}\subseteq\Lint{\BT{N}}$, so we conclude by applying Theorem~\ref{thm:app}. 
\end{proof}

We now show that $\ler$ is representable by some relational graph model.

\subsection{The Model $\cE$ Induces Minimal Theories}
Let $\cE = (E,\iota) = \Compl{(\nat,\emptyset)}$ be the relational graph model defined in Example~\ref{ex:graphrel}.
This model has infinitely many atoms, that we denote by $(\xi_n)_{n\in\nat}$, and as injection $\iota : \Mfin{E}\times E\to E$ simply the inclusion, therefore no atom $\xi_n$ can be equal to an arrow $a\to_\iota\alpha$.  
In other words, the elements $\alpha\in E$ are generated by: 
$$
	\qquad\qquad\alpha \ ::=\  \xi_n\mid a \to \alpha \qquad\qquad a \ ::=\  [\alpha_1,\dots,\alpha_k]\qquad\qquad \textrm{ (for $n,k\ge 0$).}
$$
Recall that we provided the type inference rules in Figure~\ref{fig:Typing}.
We are going to show that the interpretations of two \lam-terms $M$ and $N$ are different whenever $M\not\ler N$.

Notice that every element $\alpha\in E$ can be written uniquely as $\alpha = a_1\to\cdots\to a_n\to\xi_i$.
In this case, the atom $\xi_i$ is called the \emph{range of $\alpha$} and denoted by $\rg{\alpha}$.
We use the compact notation $\omega^k\to\alpha$ to denote the element $\omega\to\cdots\to\omega\to\alpha$ (with $k$ occurrences of $\omega$).

Recall that the size $\size t$ of $t\in\App$ has been introduced in Definition~\ref{def:sizet}.

\begin{lemma}\label{lemma:bourdelique}
Let $M,N\in\Lambda$.
If $M\sqle_\Hst N$ but $M\not\ler N$, then there are $\Gamma,\alpha$ such that:
\begin{enumerate}[(i)]
\item $\Gamma\vdash_\cE t : \alpha$ for some $t\in\BT[*]{M}$, 
\item for all $u\in\BT[*]{N}$ we have $\Gamma\not\vdash_\cE u : \alpha$,
\item $\rg{\alpha}= \xi_{\size{t}}$,
\item for all $\beta\in\Gamma$, $\rg{\beta} = \xi_j$ for some $j\le \size{t}$.
\end{enumerate}
\end{lemma}

\begin{proof}
Since $M\not\ler N$, we have that $M$ must be solvable. As moreover $M\sqle_\Hst N$, by Remark~\ref{rem:simh} we get that $M,N$ have similar hnf's. Therefore, only two cases are possible.

1) $M =_\beta \lam x_1\dots x_n.x_iM_1\cdots M_k$ and $N =_\beta \lam\seq xz_1\dots z_m.x_iN_1\cdots N_kP_1\cdots P_m$ for $m > 0$.
We suppose that $x_i$ is free, the other case being analogue. 
This case follows easily by taking $t = \lam\seq x.x_i\bot\cdots\bot\in\BT[*]{M}$ whose size is $n+1$,
$\Gamma = x_i : [\emptymset^k\to\xi_{n+1}]$ and $\alpha = \emptymset^n\to \xi_{n+1}$.
The fact that $\Gamma\not\vdash u : \alpha$ for all $u = \lam\seq x\seq z.x_iu_1\cdots u_{k+m}\in\BT[*]{N}$ follows from $m > 0$ and the fact that $\xi_{n+1}$ is an atom, hence different from any arrow type by definition of $\cE$.

2) $M =_\beta \lam x_1\dots x_nz_1\dots z_m.x_iM_1\cdots M_kP_1\cdots P_m$ and $N =_\beta \lam\seq x.x_iN_1\cdots N_k$ where, for every $\ell\le m$ there is a B\"ohm-like tree $V$ such that $\BT{P_\ell}\BTle V\geeta[\infty] z_\ell$, for every $j\le k$ we have $M_j\sqle_\Hst N_j$ but $M_q\not\ler N_q$ for some $q$.
We suppose that $x_i$ is free, the other case being analogue.
By induction hypothesis, there is $t_q\in\BT[*]{M_q}$ such that $\Gamma\vdash t_q : \alpha$ with $\rg{\alpha} = \xi_\size{t_q}$, for all $\beta\in\Gamma$ we have $\rg{\beta} = \xi_j$ for some $j\le \size{t_q}$ and for all $u_q\in\BT[*]{N_q}$ we have $\Gamma\not\vdash u_q : \alpha$.
On the one side, we construct the derivation:
$$
	\infer={
	\Gamma_0+\Gamma_1\vdash t =\lam\seq x\seq z.x_it_1\cdots t_ks_1\cdots s_m : a_1\to\cdots\to a_n\to\emptymset^m\to \xi_{\size{t}}
	}{
	\infer={
	\Gamma_0+\Gamma\vdash x_it_1\cdots t_ks_1\cdots s_m : \xi_\size{t}
	}{
		\Gamma_0\vdash x_i : \emptymset^{q-1}\to [\alpha]\to \emptymset^{k+m-q}\to \xi_{\size{t}}
		& 
		\Gamma \vdash t_q : \alpha 
		}
	}
$$
where $\Gamma_0 = x_i : [\emptymset^{q-1}\to [\alpha]\to \emptymset^{m+k-q}\to \xi_{\size{t}}]$, 
$\Gamma = \Gamma_1 + (x_1:a_1,\dots,x_n : a_n)$, for all $j\le k$ we have 
$t_j\in\BT[*]{M_j}$, and for all $\ell\le m$ we have $s_\ell\in\BT[*]{P_\ell}$.
On the other side, each $u\in\BT[*]{N} - \{\bot\}$ must have the shape $u = \lam\seq x.x_iu_1\cdots u_k$ and each derivation of $\Gamma_0+\Gamma_1\vdash u : \xi_{\size t}$ requires to derive $\Gamma\vdash u_q : \alpha$.

Indeed, since all $\beta\in\Gamma$ satisfy $\rg{\beta}\le\size t_q <\size t$, they cannot be used to produce a $\xi_\size t$ so the decomposition $\Gamma_0 + \Gamma$ is in fact unique: 
$$
	\infer={
	\Gamma_0+\Gamma_1\vdash u = \lam\seq x.x_iu_1\cdots u_k : a_1\to\cdots\to a_n\to\emptymset^m\to \xi_{\size{t}}
	}{
	\infer={
	\Gamma_0+\Gamma\vdash x_iu_1\cdots u_k : \emptymset^{m}\to\xi_\size{t}
	}{
		\Gamma_0\vdash x_i : \emptymset^{q-1}\to [\alpha]\to \emptymset^{m+k-q}\to \xi_{\size{t}}
		& 
		\Gamma \vdash u_q : \alpha 
		}
	}
$$
By induction hypothesis $\Gamma\vdash u_q : \alpha$ is impossible, which entails $\Gamma_0+\Gamma_1\not\vdash u : a_1\to\cdots\to a_n\to\emptymset^m\to \xi_{\size{t}}$.
\end{proof}

We are now able to show that $\cE$ induces the same \lam-theory as Engeler's model~\cite{Engeler81} and the same inequational theory as the filter model defined by Ronchi Della Rocca in~\cite{Ronchi82}. 

\begin{theorem}\label{thm:ThofE} For $M,N\in\Lam$ we have 
$$
	M\ler N \iff \cE\models M\sqle N
$$
Therefore $\Thle(\cE) =\ \ler$ and $\Th(\cE) = \BTth$.
\end{theorem}

\begin{proof} $(\Rightarrow)$ It follows immediately by~Proposition~\ref{prop:lerminimal}.

$(\Leftarrow)$ Suppose, by the way of contradiction, that $M\not\ler N$ but $\cE\models M\sqle N$. 
Since $\Thle(\cE)$ is sensible, Lemma~\ref{lemma:maximality} entails $M\sqle_\Hst N$.
By Lemma~\ref{lemma:bourdelique}, there exists $t\in\BT[*]{M}$ and $\Gamma,\alpha$ such that $\Gamma\vdash t : \alpha$ while $\Gamma\not\vdash u : \alpha$ for all $u\in\BT[*]{N}$. 
By the Approximation Theorem, we have $(\Gamma,\alpha)\in\Lint{M} - \Lint{N}$, which is impossible.
\end{proof}

The above proof-technique could be suitably generalized to prove that all non-extensional relational graph models induce the same inequational theory, namely $\ler$.

\begin{theorem}\label{thm:minimalth} The relation $\ler$ is the minimal inequational graph theory.
Similarly, $\BTth$ is the minimal relational graph theory.
\end{theorem}

\subsection{A Semantic Characterization of Normalizability}

We now show that in the model $\cE$ $\beta$-normalizable \lam-terms have a simple semantic characterization.
Indeed, since there are no equations between atoms and arrow types, it makes sense to define whether $\emptymset$ occurs in a type $\alpha$ with a certain \emph{polarity} $p\in\{+,-\}$.

\begin{definition}
For all elements $\alpha$ of $\cE$ the relations $\emptymset\in^+\alpha$ and $\omega\in^{-}\alpha$ are defined by mutual induction as follows (where $p$ is a polarity and $\lnot p$ denotes the opposite polarity): 
\begin{enumerate}[(i),ref={\roman*}]
\item $\emptymset\in^- a \to \beta$ if $a = \emptymset$;
\item if $\emptymset\in^p\beta$ then $\emptymset\in^p a \to \beta$; 
\item if $\emptymset\in^{\neg p} \beta$ then $\emptymset\in^p ([\beta]+a)\to \gamma$.
\end{enumerate}
When $\emptymset\in^+\alpha$ (resp.\ $\emptymset\in^-\alpha$) holds we say that $\emptymset$ \emph{occurs positively} (resp.\ \emph{negatively}) in $\alpha$.
We write $\emptymset\notin^p\alpha$ whenever $\omega$ does not occur in $\alpha$ with polarity $p$.
\end{definition}
These notions extend to multisets in the obvious way, that is, $\emptymset\in^p [\alpha_1,\dots,\alpha_n]$ whenever $\emptymset\in^p \alpha_i$ for some index~$i$.
Similarly, $\omega\in^p\Gamma$ whenever there is $x\in\Var$ such that $\omega\in^p\Gamma(x)$.

\begin{theorem}\label{thm:characterization_nf} 
Let $M\in\Lam$. The following are equivalent:
\begin{enumerate}
\item\label{thm:characterization_nf1}
	$M$ has a $\beta$-normal form,
\item\label{thm:characterization_nf2}
	$\Gamma\vdash_\cE M : \alpha$ for some environment $\Gamma$ and 
	type $\alpha$ such that $\emptymset\notin^+\alpha$ and $\emptymset\notin^-\Gamma$.
\end{enumerate}
\end{theorem}

\begin{proof}
(\ref{thm:characterization_nf1} $\Rightarrow$ \ref{thm:characterization_nf2})
Straightforward induction on the derivation of $\Gamma\vdash \nf(M) : \alpha$.

(\ref{thm:characterization_nf2} $\Rightarrow$ \ref{thm:characterization_nf1})
We proceed by structural induction on the pair $(\Gamma,\alpha)$.
By~Lemma~\ref{lemma:hr-decreases}, $M$ has a head normal form $\lam x_1\dots x_n.x_iM_1\cdots M_k$ having type $\alpha$ in the context $\Gamma$, which entails $\alpha = a_1\to\cdots\to a_n\to \alpha'$. 
So, there exists a derivation of the form:
$$
	\infer=[\texttt{lam}]{\Gamma\vdash \lam x_1\dots x_n.x_iM_1\cdots M_k : \alpha = a_1\to\cdots\to a_n\to\alpha'}{
		\infer=[\texttt{app}]{\Gamma,x_1:a_1,\dots,x_n:a_n\vdash x_iM_1\cdots M_k : \alpha'}{
		\infer[\texttt{var}]{\Gamma_0\vdash x_i : b_1\to\cdots\to b_k\to \alpha'}{}
		&
		\Gamma_{j\ell}\vdash M_j : \beta_{j\ell}
		&
		1\le j\le k
		&
		1\le \ell\le k_j
		}
	}
$$
where $\Gamma_0 = [x_i : b_1\to\cdots\to b_k\to \alpha']$, for all $j\le k$ we have $b_j = [\beta_{j1},\dots,\beta_{jk_j}]$ and $\Gamma,x_1:a_1,\dots,x_n:a_n = \Gamma_0 + (\sum_{j =1}^k\sum_{\ell =1}^{k_j} \Gamma_{j\ell})$.
Since $\emptymset\notin^+\alpha$ and $\emptymset\notin^-\Gamma$ we get for each $j\le k$ that $b_j$ is non-empty, and moreover that $\emptymset\notin^+\beta_{j\ell}$ and $\emptymset\notin^-\Gamma_{j\ell}$ for all $\ell\le k_j$.
From the induction hypothesis, we get that all the $M_i$'s have a $\beta$-normal form.
\end{proof}

As proved in~\cite{ManzonettoR14}, Theorem~\ref{thm:characterization_nf} holds for every relational graph model~$\cD$ \emph{preserving $\emptymset$-polarities} (in a technical sense). 
An analogue of this theorem also holds for the usual intersection type systems~\cite{interbcd}. 
However, for an extensional relational graph model $\cD$ this lemma is enough to conclude $\Thle(\cD) =\, \sqle_\Hpl$, while this is not the case for filter models. 
In other words, being extensional and preserving $\emptymset$-polarities are sufficient conditions for a relational graph model to be fully abstract for $\Hpl$. 
We will now provide conditions that are both necessary and sufficient.


%% file: include/fullabsHpl.tex
In this section we provide a characterization of those relational graph models that are (inequationally) fully abstract for~$\BTeth$. 
We first introduce the notion of \emph{\lam-K\"onig} relational graph model (Definition~\ref{def:lambdaK}), and show
that a relational graph model $\cD$ is extensional and \lam-K\"onig exactly when the induced inequational theory is the preorder $\sqle_\Hpl$ (Theorem~\ref{thm:semantic_characterization}). 

Since our proof technique does not rely on the quantitative properties of relational graph models, hereafter we rather prefer to use the categorical interpretation.

\subsection{Lambda K\"onig Relational Graph Models}

Before entering into the technicalities we try to give the intuition behind our condition. 
The main issue is to find suitable conditions for assuring that if $M,N$ have the same interpretation in a relational graph model $\cD$ then $M =_{\Hpl}\! N$, equivalently that $M \neq_{\Hpl}\! N$ implies $\Int[\cD]{M} \not= \Int[\cD]{N}$.
Now, the idea behind Theorem~\ref{thm:newsep} is that two \lam-terms $M,N$ are equal in $\Hst$, but different in $\BTeth$, when there is a (possibly \emph{virtual}\footnote{%
Intuitively, a position $\sigma$ is virtual if it does not belong to $\BT{M}$, but rather to one of its $\eta$-expansions.
For more details we refer to~\cite[\S10.3]{Bare}.
}) position $\sigma\in\Seq$ such that, say, $\BT{M}_\sigma = x$ while $\BT N_\sigma$ is an infinite $\eta$-expansion of $x$ following some $T\in\Trees$. 
As a consequence of this fact, our models need to separate $x$ from any $\bJ_Tx$ for $T\in\Trees$ in order to be fully abstract for $\Hpl$.

Notice now that in any extensional relational graph model $\cD$, every element $\alpha_0$ is equal to an arrow, so one can always try to unfold $\alpha$ following a function $f$, starting with:
$$
	\alpha_0 = a_0\to\dots\to a_{f(0)}\to\beta
$$
If there is an $\alpha_1\in a_{f(0)}$, then one can keep unfolding:
$$
	\alpha_1= a'_0\to\dots\to a'_{f(1)}\to\beta'
$$
and so on. More generally, at level $\ell$ we have $\alpha_\ell = b_0\to\dots\to b_{f(\ell)}\to\alpha'$ for some $b_i\in \Mfin{D}$ and $\alpha'\in D$ and as long as there exist an $\alpha_{\ell+1}\in a_{f(\ell)}$, we can keep unfolding it at level $\ell+1$.
There are now two possibilities.
	\begin{enumerate}
	\item 
	If this process continues indefinitely, then we consider that $\alpha$ can actually be unfolded following $f$. 
	\item Otherwise, if  at some level $\ell$ we have $a_{f(\ell)} = \emptymset$, then the process is forced to stop and we consider that $\alpha$ cannot be unfolded following $f$.
	\end{enumerate}

Now, as $T\in\Trees$ is a finitely branching infinite tree, by K\"onig's lemma there exists an infinite path $f$ in $\BT{\bJ_T}$. Since the interpretation of $\bJ_T$ is inductively defined (rather than coinductively), we will then have that $[\alpha]\to\alpha\notin\Int{\bJ_T}$ for any $\alpha$ whose unfolding can actually follow~$f$.

In some sense such an $\alpha$ is \emph{witnessing within the model} the existence of an infinite path $f$ in $T$, and therefore in $\bJ_T$. 
The following is a formal definition
 of such a witness.

Recall from Definition~\ref{def:inf_paths}, that $\pathsof{T}$ denotes the set of infinite paths of a tree $T\in\Trees[]$.

\begin{definition}\label{def:witness} 
Let $\cD$ be an \rgm, $T\in\Trees$ and $f\in\pathsof{T}$.
\begin{itemize}
\item We coinductively define the set $\Wit[\cD,f]{T}$ of all \emph{witnesses for $T$ (in $\cD$) following $f$}.
	An element $\alpha\in D$ belongs to $\Wit[\cD,f]{T}$ whenever there exist $a_{0},\dots, a_{f(0)}\in\Mfin{\cD}$ and $\alpha'\in\cD$ such that
	$$
		\alpha = a_{0}\to\cdots\to a_{f(0)}\to\alpha'
	$$
	and there is a $\beta\in a_{f(0)}$ belonging to $\Wit[\cD,\funshift{f}{1}]{\subt{T}{\langle f(0)\rangle}}$ where $\funshift{f}{1}$ maps $k\mapsto f(k+1)$.

\item 
	We say that $\alpha$ is a \emph{witness for $T$ in $\cD$} when there exists an $f\in\pathsof{T}$ such that $\alpha$ is a witness for $T$ in $\cD$ following $f$.
\item
	We let $\Wit[\cD]{T}$ be the set of all witnesses for $T$ in $\cD$.
\end{itemize}
\end{definition}
We formalize the intuition given above by showing that $\Wit[\cD]{T}$ is constituted by those $\alpha\in D$ such that $[\alpha]\to\alpha\notin\Int{\bJ_T}$. 
We first prove the following technical lemmas.

\begin{lemma}\label{lemma_ruoppoloide}
Let $\cD$ be an \rgm.
For all $T\in\Trees$ and $x\in\Var$ we have $\Int[\cD]{\bJ_T x}_{x}\subseteq \Int[\cD]{x}_{x}$.
\end{lemma}
\begin{proof}
Let $(a,\alpha)\in \Int{\bJ_T x}_{x}$. By Theorem~\ref{thm:app} there is $t\in\BT[*]{\bJ_T x}$ such that ${(a,\alpha)\in\Int{t}_x}$.
Proceeding by induction on $t$, we show that $a = [\alpha]$. 
The case $t=\bot$ is vacuous.  

Consider $t = \lam z_0\dots z_{T(\emptyseq)-1}.xt_0\cdots t_{T(\emptyseq)-1}$ where  $t_i \in\BT[*]{\bJ_{\subt{T\ }{\langle i\rangle}} z_i}$.
By Definition~\ref{def:IntSem}(\ref{def:IntSem2}) we have $(a,\alpha)\in\Int{t}_{x}$ if and only if $\alpha = a_{0}\to\cdots\to a_{T(\emptyseq)-1}\to \alpha'$ for some 
$a_i\in\Mfin{D}$ and $\alpha'\in D$ such that 
$((a,a_{0},\dots, a_{T(\emptyseq)-1}), \alpha')\in\Int{xt_0\cdots t_{T(\emptyseq)-1}}_{x,z_0,\dots,z_{T(\emptyseq)-1}}$. 
By Definition~\ref{def:IntSem}\eqref{def:IntSem1} and \eqref{def:IntSem3}, we get $a = [b_{0}\to\cdots\to b_{T(\emptyseq)-1}\to \alpha']$ where $b_i = [\beta_{i,1},\dots,\beta_{i,k_i}]$ and there is a decomposition $a_i = \sum_{j=1}^{k_i}a_{i,j}$ such that $(a_{i,j},\beta_{i,j})\in\Int{t_i}_{z_i}$. 
By the inductive hypothesis we have $a_{i,j} = [\beta_{i,j}]$. 
Therefore $a_i = [\beta_{i,1},\dots,\beta_{i,k_i}] = b_i$ which in its turn entails $a =[ b_{0}\to\cdots\to b_{T(\emptyseq)-1}\to \alpha'] =  [a_{0}\to\cdots\to a_{T(\emptyseq)-1}\to \alpha'] = [\alpha].$  
\end{proof}

\begin{lemma}\label{lemma:technicalJ_Tk} 
Let $\cD$ be an \rgm.
For all $T\in\Trees$, $\alpha\in \Wit[\cD]{T}$ and $t\in \BT[*]{\bJ_{T}x}$ we have $([\alpha],\alpha)\notin\Int[\cD]{t}_{x}$.
\end{lemma} 
\begin{proof}
We proceed by induction on the size $\size{t}$ of $t$. 

Case $\size{t} = 0$. This case is trivial since $t =\bot$ and $\Int{\bot}_x = \emptyset$.

Case $\size{t} > 0$. 
Then $t = \lam z_0\dots z_{T(\emptyseq)-1}.xt_0\cdots t_{T(\emptyseq)-1}$ where each $t_i \in\BT[*]{\bJ_{\subt{T\,\, }{\langle i\rangle}} z_i}$ is such that $\size{t_i} < \size{t}$.
By Definition~\ref{def:IntSem}(\ref{def:IntSem2}) we have $([\alpha],\alpha)\in\Int{t}_{x}$ if and only if $\alpha = a_{0}\to\cdots\to a_{T(\emptyseq)-1}\to \alpha'$ for some $a_i = [\alpha_{i,1},\dots,\alpha_{i,k_i}]$ and
$(([\alpha],a_{0},\dots, a_{T(\emptyseq)-1}), \alpha')\in\Int{xt_0\cdots t_{T(\emptyseq)-1}}_{x,z_0,\dots,z_{T(\emptyseq)-1}}$.
As $\Int{t_i}_{z_i}\subseteq \Int{z_i}_{z_i}$ by Lemma~\ref{lemma_ruoppoloide}, we obtain $([\alpha_{i,j}],\alpha_{i,j})\in\Int{t_i}_{z_i}$ for all $i\le T(\emptyseq)-1$ and $j\le k_i$.
Since $\alpha\in\Wit[\cD,f]{T}$ for some $f$, there exists a witness $\alpha_{f(0),j}\in a_{f(0)}$ for $\subt{T}{\langle f(0)\rangle}$ following~$\funshift{f}{1}$.
By $\alpha_{f(0),j}\in\Wit{\subt{T}{\langle f(0)\rangle}}$ and the induction hypothesis we get $([\alpha_{f(0),j}],\alpha_{f(0),j})\notin\Int{t_{f(0)}}_{z_{f(0)}}$, which is a contradiction.
\end{proof}

By applying the Approximation Theorem we get the following characterization of~$\Wit[\cD]{T}$.

\begin{proposition}\label{prop:char_WT}
  For any extensional \rgm~$\cD$ and any tree $T\in\Trees$:
   $$ 
   		\Wit[\cD] T  =  \{\alpha\in D \st ([\alpha],\alpha)\not\in \Int{\bJ_Tx}_x\}. 
	$$
\end{proposition}
\begin{proof}
$(\subseteq)$ Follows immediately from the Approximation Theorem~\ref{thm:app} and from Lemma~\ref{lemma:technicalJ_Tk}.

$(\supseteq)$
    Let $\alpha\in D$ such that $([\alpha],\alpha)\not\in \Int{\bJ_T x}_x$.
We coinductively construct a path $f$ such that $\alpha\in \Wit[f] T$.
As $T$ is infinite we have $\bJ_Tx=_{\beta} \lambda z_0\dots z_n.x(\bJ_{\subt{T\,\, }{\langle 0\rangle}} z_0)\cdots(\bJ_{\subt{T\,\,}{\langle n\rangle}} z_n)$ and since~$\cD$ is extensional $\alpha = a_{0}\to\cdots\to a_{n}\to \alpha'$. 
From $([\alpha],\alpha)\not\in\Int{\bJ_Tx}_x$ and the soundness we get $([\alpha],\alpha)\notin\Int{\lambda z_0\dots z_n.x(\bJ_{\subt{T\,\, }{\langle0\rangle}} z_0)\cdots(\bJ_{\subt{T\,\,}{\langle n\rangle}} z_n)}_x$.
Therefore there exist an index $k\le n$ such that $a_k\neq\emptymset$ and an 
element $\beta\in a_k$ such that  $([\beta],\beta)\not\in \Int{\bJ_{\subt{T\,\,}{\langle k\rangle}} z_k}_{z_k}$. 
In particular, this entails that the subtree $\subt{T}{\langle k\rangle}$ is infinite because $([\beta],\beta)$ belongs to the interpretation of any finite $\eta$-expansion of $z_k$.
Therefore we set $f(0) = k$ and, for all $n\in\nat$, $f(n+1)=g(n)$ where $g$ is the function given by the coinductive hypothesis and satisfying $\beta\in\Wit[\cD,g]{\subt{T}{\langle k\rangle}}$.
By construction of $f$, we conclude that $\alpha\in\Wit[\cD,f]T$.
\end{proof}

It should be now clear that a relational graph model $\cD$, to be fully abstract for $\Hpl$, needs for every \lam-definable infinite $\eta$-expansion of the identity an element in $D$ witnessing its infinite path, which exists by K\"onig's lemma. This justifies the definition below.

\begin{definition}[\lam-K\"onig models]\label{def:lambdaK}
An \rgm{} $\cD$ is \emph{\lam-K\"onig} if for every $\treeof{T}\in\Trees$, $\Wit[\cD]{T}\neq\emptyset$.
\end{definition}


We will mainly focus on the \lam-K\"onig condition, since the extensionality is clearly necessary, as $\blam\eta\subseteq\Hpl$.

We start by showing that, if $\cD$ is an extensional \lam-K\"onig relational graph model, then $\Thle(\cD) =\ \sqle_{\Hpl}$.
Indeed, since every $T\in\Trees$ has a non-empty set of witnesses $\Wit[\cD] T$, 
by Proposition~\ref{prop:char_WT}, there is an element $\alpha\in\Wit[\cD] T$ such that $[\alpha]\to\alpha\notin\Int{\bI} - \Int{\bJ_T}$. 
Thus, $\cD$ separates $\bI$ from all the~$\bJ_T$'s for $T\in\Trees$.

\begin{theorem}[Inequational Full Abstraction]\label{thm:FA}
Let $\cD$ be an extensional \lam-K\"onig \rgm, then:
$$
	M\sqle_\Hpl N \iff \cD\models M\sqsubseteq N
$$
\end{theorem}

\begin{proof} $(\Rightarrow)$ This follows directly from Corollary~\ref{cor:extensionality}.

$(\Leftarrow)$ We assume, by the way of contradiction, that $\cD\models M\sqsubseteq N$ but $M\not\sqle_\Hpl N$. 
By the maximality of $\sqle_\Hst$ shown in~Lemma~\ref{lemma:maximality} and the fact that $\Int{M}_{\seq x}\subseteq \Int{N}_{\seq x}$ we must have $M\sqle_\Hst N$.
By Theorem~\ref{thm:newsep} there exists a context $C\hole-$ such that $C\hole M =_{\beta\eta} \bI$ and $C\hole N =_{\BTth} \bJ_T$ for some $T\in\Trees$.
By monotonicity of the interpretation $\Int-$ and since $\BTth\eta\subseteq\BTeth\subseteq\Th(\cD)$ by Corollary~\ref{cor:extensionality}, we have $\Int{\bI} = \Int{C\hole M}\subseteq \Int{C\hole N} = \Int{\bJ_T}$. 
We derive a contradiction by applying Proposition~\ref{prop:char_WT}.
\end{proof}

We now show the converse, namely that if a relational graph model is (inequationally) fully abstract for $\Hpl$, then it is extensional and \lam-K\"onig.

\begin{theorem}\label{thm:theotherside} 
Let $\cD$ be an \rgm. If $\Thle(\cD)=\ \sqle_\Hpl$ or $\Th(\cD) = \Hpl$ then $\cD$ is extensional and \lam-K\"onig.
\end{theorem}
\begin{proof} Obviously $\cD$ must be extensional since $\Hpl$ is an extensional \lam-theory.
By contradiction, we suppose that it is not \lam-K\"onig. Then there is $T\in\Trees$ such that $\Wit[\cD]{T} = \emptyset$ and, by Proposition~\ref{prop:char_WT}, we get $\Int{\bI}=\Int{\bJ_T}$.
This is impossible since $\bI \not\sqle_\Hpl \bJ_T$.
\end{proof}

From Theorem~\ref{thm:FA} and Theorem~\ref{thm:theotherside} we get the main result of this section.

\begin{theorem}\label{thm:semantic_characterization}
For an \rgm~$\cD$, the following are equivalent:
\begin{enumerate}[(i),ref={\roman*}]
\item $\cD$ is extensional and \lam-K\"onig,
\item $\cD$ is inequationally fully abstract for $\sqle_\Hpl$, 
\item $\cD$ is fully abstract for $\Hpl$.
\end{enumerate}
\end{theorem}
In particular, the model $\cD_\star$ induces the same inequational and \lam-theories as Coppo, Dezani-Ciancaglini and Zacchi's filter model~\cite{CoppoDZ87}, a result first appeared  in \cite{ManzonettoR14}.

\begin{corollary}\label{cor:Dstar}
The model $\cD_\star$ of Example~\ref{ex:graphrel} is inequationally fully abstract for Morris's preorder $\sqle_\Hpl$. In particular $\Th(\cD_\star) = \Hpl$.
\end{corollary}


%% file: include/fullabsHst.tex
In this section we provide a characterization of those relational graph models that are (inequationally) fully abstract for~$\Hst$\!. 
At this purpose, we introduce the notion of  \emph{hyperimmune}  relational graph model (Definition~\ref{def:hyperim}), which is in some sense dual to the notion of \lam-K\"onig.
We prove that a relational graph model $\cD$ is extensional and hyperimmune exactly when the induced inequational theory is the preorder $\sqle_\Hst$ (Theorem~\ref{thm:semantic_characterization2}).

The technique in this section is B\"ohm-like tree oriented, therefore it is convenient to fix some notations concerning B\"ohm-like trees and their interpretation.

\begin{notation} We simply denote by $\jc^x_T$ the tree $\BT{\jc_Tx}$, where $x$ is a fresh variable. 
Given a B\"ohm-like tree $V$, we write $V^*$ for the set $\{t\in\App\st t\BTle V\}$ of its finite approximants and we set $\Lint[\cD]{V} = \bigcup_{t\in V^*} \Lint[\cD]{t}$.
\end{notation}

\subsection{Decomposing Infinite $\eta$-Expansions}
When considering B\"ohm trees of \lam-terms, the infinite $\eta$-expansion  $\leeta[\infty]$ can be decomposed into the finitary one $\,\leeta\,$ followed by a more restricted infinite $\eta$-expansion $\,\leeta[!]$ that only allows to $\eta$-expand variables (Lemma~\ref{lemma:decomposeGeeta}). 
 
Remember that the difference between $\leeta[\infty]$ and  $\leeta$ lies in the fact that the former allows countably many possibly infinite $\eta$-expansions, whereas the latter only countably many finite ones. 
As first noticed by Severi and de~Vries in their recent work on the infinitary $\lambda$-calculus~\cite{SeveriV16}, in a way this difference only concerns $\eta$-expansions of variables. 
Consider for instance $yy\not\leeta  \lam x. y\jc^y \jc^x$ and $yy \leeta[\infty] \lam x. y\jc^y \jc^x$. The tree  $\lam x. y\jc^y \jc^x$ is an infinite $\eta$-expansion of $yy$, which is not a variable. Nevertheless, one can narrow down to infinite $\eta$-expansions  of the variables $x$ and $y$, by noticing that $yy \leeta \lam x. yyx  \leeta[\infty] \lam x. y\jc^y\jc^x$. 

\begin{definition}\label{def:eta!inf}
Let $\leeta[!]$ be the greatest relation between B\"ohm-like trees such that $U\leeta[!]V$ entails that:
  \begin{itemize}
  \item $U=V=\bot$,
  \item or $U = x$ and $V= \jc^x_T\,$ for $T\in\Trees$,
  \item or $U = \lambda x_1\dots x_n.x_i U_1\cdots U_k$ and $V =\lambda x_1\dots x_n. x_i V_1\cdots V_k$ (for some $i,k,n\in\nat$) where $U_j\leeta[!] V_j$ for all $j\le k$.
  \end{itemize}
\end{definition}
\input{include/fig5}

\begin{example}
We have $\lam x. yyx  \leeta[!] \lam x. y\jc^y\jc^x$, whereas $yy \not\leeta[!] \lam x. y\jc^y\jc^x.$ 
\end{example}

Clearly $\leeta[!]$ is a subrelation of $\leeta[\infty]$, as it is the case for $\leeta$. 
Also notice that $\leeta[!]$ and $\leeta$ are completely orthogonal relations, in the sense that  $U\!\leeta\! V$ and $U\!\leeta[!]\!V$ imply $U=V$.

For technical reasons we also need an \emph{inductive} version of the relation $\leeta[!]$, that we denote by $\!\mstoinv[\eta!]$. 
Intuitively $U \mstoinv[\eta!] V$ means that $V$ is obtained from $U$ by performing \emph{finitely} many infinite $\eta$-expansions of variables. 
 (The notation $\!\mstoinv[\eta!]$ is borrowed from~\cite{SeveriV16}, although Severi and de Vrijes use the symbol for a more general notion of $\eta!$-rule.)

As we will see, actually hyperimmune relational graph models cannot distinguish between the relation $U \mstoinv[\eta!] V$ and its coinductive version $U\leeta[!]V,$ in the sense expressed by the equivalence (\ref{lemma:hyp=>eta!3}$\iff$\ref{lemma:hyp=>eta!4}) in Proposition~\ref{lemma:hyp=>eta!}.

\begin{definition}\label{def:eta!fin}
Let $\!\mstoinv[\eta!]\,$ be the smallest relation between B\"ohm-like trees closed under the following rules:
  \begin{itemize}
  \item $U \,\mstoinv[\eta!] U$ for $U\in\BTset$,
  \item $x \,\mstoinv[\eta!] \jc^x_T$ for $T\in\Trees$,
  \item $\lambda x_1\dots x_n. x_i U_1\cdots U_k \,\mstoinv[\eta!] \lambda x_1\dots x_n.x_i V_1\cdots V_k$ (for some $i,k,n\in\nat$) and $U_j \mstoinv[\eta!] \!V_j$ for all $j\le k$.
  \end{itemize}
\end{definition}
Notice that this can be seen as the inductive version of the coinductive Definition~\ref{def:eta!inf}.
Alternatively, one can define $\msto[\eta!]$ as the transitive-reflexive and contextual closure of 
$$
(\eta!)\qquad \jc^x_T \redto[\eta!] x\textrm{ \;\;for all } T\in\Trees.
$$

\begin{example}\label{ex:BTlikethinghys}
Let $\bU = \BT{\lambda xy.\bY (\lambda u.xyu)} $ and $\bV = \BT{\lambda xy.\bY (\lambda u.x(\jc y)u)}$. 
These two trees are depicted in Figure~\ref{fig:exEta!}. 
We have that $\bU\leeta[!] \bV$ while $\bU \mstoinv[\eta!]\!\!\!\!\!\not\ \ \bV$, because $\bV$ is obtained from $\bU$ by performing an infinite amount of $\eta!$-expansions of variables.
\end{example}

\begin{lemma}(Decomposition of $\leeta[\infty]$)\label{lemma:decomposeGeeta}
Let $M,N\in\Lambda$. 
We  have $\BT{M}\leeta[\infty]\BT{N} $ if and only if there exists a B\"ohm-like tree $W$ such that $\BT{M} \leeta W\leeta[!]\BT{N}. $
\end{lemma}
\begin{proof}
 $(\Leftarrow)$  By transitivity of $\leeta[\infty]$, using the fact that both $\leeta$ and $\leeta[!]$ are contained in $\leeta[\infty]$.

$(\Rightarrow)$
We construct $W$ coinductively.

In case $M$ and $N$ are unsolvable we just take $W = \bot$.

 Otherwise, $M =_\beta \lam\seq x.x_iM_1\cdots M_k$ and $N =_\beta \lam\seq x z_1\dots z_m.x_iN_1\cdots N_{k}P_1\cdots P_m$ where $\BT{M_j}\leeta[\infty] \BT{N_j}$ for all $j\le k$ and $z_\ell\leeta[\infty] \BT{P_\ell}$ for all $\ell\le m$.
    By the coinductive hypothesis, for all $j$ there exists $W_j$ such that $\BT{M_j}\leeta W_j\leeta[!] \BT{N_j}$. 
    We then set $W = \lam\seq x z_1\dots z_m.x_i W_1\cdots W_n z_{1}\cdots z_{m}.$
    Clearly $\BT{M} \leeta W$. 
    In order to show that $W\leeta[!] \BT{N}$ holds we have to prove, for all $\ell\le m$, that  $\BT{P_\ell} = \jc^{z_{\ell}}_{T_\ell}$ for some $T_\ell\in\Trees$. 
Since $z_\ell \leeta[\infty] \BT{P_\ell}$, we conclude by Proposition~\ref{prop:J_T}. 
\end{proof}

As shown in~\cite{SeveriV16}, the decomposition of $\leeta[\infty]$ could be extended to all B\"ohm-like trees using possibly non-recursive infinite $\eta$-expansions of $x$ in the definition of $\leeta[!]$.
Since here we use $\leeta[!]$ in connection with hyperimmune relational graph models (Proposition~\ref{lemma:hyp=>eta!}), a semantic notion only concerning recursive trees,
it is crucial to use our restricted version.

\begin{lemma}\label{lemma:FiniteRed}
  Let $U,V$ be two B\"ohm-like trees such that $U\leeta[!] V$ and let $t\in U^*$. 
  Then there exists a B\"ohm-like tree $W$ such that $U \leeta[!]W\mstoinv[\eta!] V$ and $t\in W^*$.
\end{lemma}
\begin{proof}
We proceed by induction on $t$.
\begin{itemize}

\item If $t=\bot$ it suffices to take $ W=V$.
 
\item If $t = U = x$ and $V=\jc^x_T$ we get the thesis by taking $W = U$. 

\item If $t = \lam\seq x.x_i t_1\cdots t_k$, then 
      $U = \lam\seq x.x_iU_1\cdots U_{k}$ and
      $V = \lam\seq x.x_i V_1\cdots V_k$ where $t_j \in U_j^*$ and $U_j \leeta[!] V_j$ for all $j\le k$.
   By induction hypothesis for each $j$ we get some $W_j$ such that $U_j \leeta[!]W_j\mstoinv[\eta!] V_j$ and $t_j\BTle W_j$.
   We conclude by setting $W = \lam\seq x. x_iW_1\cdots W_k$.
  \qedhere
\end{itemize}
\end{proof}

\begin{example} Let $\bU,\bV$ be the trees from Example~\ref{ex:BTlikethinghys} (and Figure~\ref{fig:exEta!}) and let $t = \lam xy. x y \bot\in U^*$. 
A possible B\"ohm-like tree $W$ given by Lemma~\ref{lemma:FiniteRed} is the following:
\begin{center}
\begin{tikzpicture}
\node (B) at (.5,3.5) {$W$};
\node[below of = B, node distance=11pt] (eq) {$\shortparallel$};
\node[below of = eq, node distance=10pt] (l1) {\!\!\!\!\!\!\!\!\!$\lam xy.x$};
\node[below of = l1, node distance=18pt,xshift=20pt] (l2) {$x$};
\node[below of = l2, node distance=18pt,xshift=20pt] (l3) {$x$};
\node[below of = l3, node distance=18pt,xshift=20pt] (l4) {$x$};
\node[below of = l4, xshift=0pt, node distance=10pt] (l5) {$\vdots$};
\draw ($(l1.south)+(10pt,0.05)$) -- ($(l2.north)-(0.1,0.05)$);
\draw ($(l2.south)+(6pt,0.05)$) -- ($(l3.north)-(0.1,0.05)$);
\draw ($(l3.south)+(6pt,0.05)$) -- ($(l4.north)-(0.1,0.05)$);

\node[below of = l1, node distance=18pt,xshift=-42pt] (l12) {$y$};
\draw ($(l1.south)+(-8pt,0.05)$) -- ($(l12.north)+(8pt,-0.15)$);

\node[below of = l2, node distance=18pt,xshift=-34pt] (l22) {$\lam z_1.y$};
\node[below of = l22, node distance=18pt] (l23) {$\lam z_2.z_1$};
\node[below of = l23, xshift=8pt, node distance=10pt] (l24) {$\vdots$};
\draw ($(l2.south)+(-2pt,0.05)$) -- ($(l22.north)+(8pt,-0.15)$);
\draw ($(l22.south)+(8pt,0.05)$) -- ($(l23.north)+(8pt,-0.15)$);

\node[below of = l3, node distance=18pt,xshift=-17pt] (l32) {$\lam z_1.y$};
\node[below of = l32, xshift=8pt, node distance=10pt] (l33) {$\vdots$};
\draw ($(l3.south)+(0pt,0.05)$) -- ($(l32.north)+(8pt,-0.15)$);

\end{tikzpicture}
\end{center}
\end{example}


\begin{lemma}\label{lemma:FiniteRed2}
Let $U,V$ be two B\"ohm-like trees such that $U\leeta[\infty] V$ and let $t\in V^*$. 
Then there exists a B\"ohm-like tree $W$ such that $U\leeta W\leeta[\infty] V$ and $t\in W^*$.
\end{lemma}

\begin{proof}
We proceed by induction on $t$.

If $t=\bot$ we take $ W=U$.

Let $t = \lam\seq x z_1 \dots z_m.x_i t_1\cdots t_k t'_1\dots t'_m \BTle 
      V = \lam\seq xz_1 \dots z_m .x_i V_1\cdots V_{k}V'_1\cdots V'_{m}$ and let
      $U = \lam\seq x.x_i U_1\cdots U_k$ be such that $U_j \leeta[\infty] V_j$ for all $j\le k$ 
and  $z_\ell \leeta[\infty] V'_\ell$ for all $\ell\le m$.
  For each $j\le k$, since  $t_j \in V_j^*$, the induction hypothesis gives some $W_j$ satisfying $U_j \leeta W_j\leeta[\infty] V_j$ and $t_j\in W^*_j$.
   We conclude by setting $W = \lam\seq x z_1 \dots z_m. x_iW_1\cdots W_k z_1 \dots z_m$.
\end{proof}

  

\subsection{Hyperimmune Relational Graph Models}

In Section~\ref{sec:Hpl} we exploited the Morris Separation (Theorem~\ref{thm:newsep}) to reduce the problem of being fully abstract for $\Hpl$ to the property  $\Lint{x} \neq \Lint{\jc_Tx}$, namely $\Wit[\cD]{T} = \Lint{x}-\Lint{\jc_Tx} \not= \emptyset$, for every tree $T\in\Trees$. The notion of $\lambda$-K\"onig relational graph model provided that. 
Here there is a similar phenomenon: we exploit the decomposition seen above (Lemma~\ref{lemma:decomposeGeeta}) 
to reduce the full abstraction for $\Hst$ to the property $\Lint{x}  = \Lint{\jc_Tx}$, namely $\Wit[\cD]{T} = \Lint{x}-\Lint{\jc_Tx} = \emptyset$, for every $T\in\Trees$. 
This is the intuition behind the following definition, which is some kind of dual of $\lambda$-K\"onig.

\begin{definition}[Hyperimmune models]\label{def:hyperim} 
  An \rgm{} $\cD$ is \emph{hyperimmune} if for every $\treeof{T}\in\Trees$ $\Wit[\cD]{T}=\emptyset$.
\end{definition}


  



The name  refers to a standard concept in computability theory: a function $f:\nat \to \nat$ is called \emph{hyperimmune} if it is not  bounded (upwardly) by any recursive function. One can prove that a relational graph model $\cD$ is hyperimmune exactly when it only admits witnesses that follow hyperimmune functions.
This observation justifies the choice of the terminology. 
The notion of hyperimmunity  first appeared in these terms in~\cite{Breuvart14} for Krivine's models.

\begin{example}
The model $\cD_\omega$ of Example~\ref{ex:graphrel} is hyperimmune. 
As a matter of fact, since the model is freely generated by the equation $\star = \omega \to \star $, it is easy to verify that no element of $\cD_\omega$ can be a witness for any $T\in\Trees$ following any infinite path $f$.
\end{example}

In the hypothesis of extensionality, this notion admits some additional characterizations. 

\begin{proposition}\label{lemma:hyp=>eta!}
  Let $\cD$ be an extensional  \rgm{}. The following statements are equivalent:
  \begin{enumerate}[(i),ref={\roman*}]
  \item\label{lemma:hyp=>eta!1} $\cD$ is hyperimmune,
  \item\label{lemma:hyp=>eta!2}  $\Lint{x}\subseteq \Lint{\jc_Tx}$ for all $T\in\Trees,$
\item\label{lemma:hyp=>eta!3} $ U \mstoinv[\eta!] V $ implies $  \Lint U \subseteq \Lint{V}$ for all B\"ohm-like trees $U,V,$
\item\label{lemma:hyp=>eta!4} $ U \leeta[!] V $ implies $   \Lint U \subseteq \Lint{V}$ for all B\"ohm-like trees $U,V.$
  \end{enumerate}
\end{proposition}
\begin{proof}
(\ref{lemma:hyp=>eta!1} \!\!$\iff$\!\!  \ref{lemma:hyp=>eta!2}) 
By Proposition~\ref{prop:char_WT} and Corollary~\ref{cor:semantic_equiv} an extensional \rgm{} $\cD$ is hyperimmune if and only if, for all $T\in\Trees$, we have that $\Lint{x} - \Lint{\jc_Tx} = \{\alpha\in D \st x:[\alpha]\not\vdash \jc_Tx : \alpha \} = \Wit[\cD]{T} = \emptyset$. 
This is equivalent to requiring that $\Lint{x}\subseteq \Lint{\jc_Tx}$ for all $T\in\Trees$.

(\ref{lemma:hyp=>eta!2} $\Rightarrow$ \ref{lemma:hyp=>eta!3}) 
By a straightforward induction on the derivation of $U\mstoinv[\eta!]V$. 

(\ref{lemma:hyp=>eta!3} $\Rightarrow$ \ref{lemma:hyp=>eta!2}) 
Trivial, as  $ x \mstoinv[\eta!] {\jc_Tx}$.

(\ref{lemma:hyp=>eta!3} $\Rightarrow$ \ref{lemma:hyp=>eta!4})  
Suppose that there exist two B\"ohm-like trees $U,V$ such that $U \leeta[!] V$ and $\Lint{U} \not\subseteq \Lint{V}$. Since $\Lint{U}=\bigcup_{t\in U^*}\Lint t$, there is a $t \in U^*$ such that  $\Lint{t} \not\subseteq \Lint{V}$. 
By Lemma~\ref{lemma:FiniteRed} there exists $W\mstoinv[\eta!]V$ such that $t\in W^*$. 
Since $\Lint{W}=\bigcup_{t\in {W}^*}\Lint t$, we get $\Lint{W}\not\subseteq \Lint{V}$, whereas $\Lint{W}\subseteq \Lint{V}$ by  \eqref{lemma:hyp=>eta!3}. 

(\ref{lemma:hyp=>eta!4} $\Rightarrow$ \ref{lemma:hyp=>eta!3}) Trivial, since the relation $ \mstoinv[\eta!]\,$ is included in $\leeta[!]$.
\end{proof}






\begin{lemma}\label{lem:yeppa!}
Let $\cD$ be an extensional \rgm{} and let $U,V$ be two B\"ohm-like trees. 
We have that $U  \leeta[\infty] V$ implies $\Lint{V}\subseteq\Lint{U}$.
\end{lemma}
\begin{proof}
By Lemma~\ref{lemma:FiniteRed2}, if $t\in V^*$ then $t\in W^*$ for some $U\leeta W$. So $\Lint{t} \subseteq \Lint{W} = \Lint{U}$, where the equality holds by extensionality. 
Since $\Lint{V}=\bigcup_{t\in V^*}\Lint t$, we get $\Lint{V} \subseteq \Lint{U}$. 
\end{proof}

In general, the hypothesis $U \leeta[\infty] V$ does not imply $\Lint{U}\subseteq\Lint{V}$, even in  presence of extensionality. 
This implication only holds when the model is in addition hyperimmune.

\begin{lemma}\label{lem:yeppa!hyper!}
Let $\cD$ be an extensional and hyperimmune \rgm{} and let $M,N\in\Lambda$. 
Then  $\,\BT{M}  \leeta[\infty] \BT{N}$ implies $\Lint{M}\subseteq\Lint{N}$.
\end{lemma}

\begin{proof}
By Lemma~\ref{lemma:decomposeGeeta} there exists a B\"ohm-like tree $W$ such that $\BT{M} \leeta W\leeta[!] \BT{N}$. Then we have $\Lint{\BT{M}} \subseteq \Lint{W}$ by extensionality and $\Lint{W} \subseteq \Lint{\BT{N}}$ by the characterization (\ref{lemma:hyp=>eta!4}) of hyperimmunity provided by Proposition~\ref{lemma:hyp=>eta!}. 
By transitivity we obtain $\Lint{\BT{M}} \subseteq \Lint{\BT{N}}$, so we conclude $\Lint{M} \subseteq \Lint{N}$ by   Theorem~\ref{thm:app}.
\end{proof}

The following theorem constitutes the main result of the section. It is actually an adaptation to relational graph models of the characterization of fully abstract Krivine's models provided in~\cite{Breuvart14, Breuvart16}.

\begin{theorem}\label{thm:semantic_characterization2}
  Let $\cD$ be an  \rgm{}. The following statements are equivalent:
\begin{enumerate}[(i),ref={\roman*}]
\item $\cD$ is extensional and hyperimmune, \label{it:H*1}
\item $\cD$ is inequationally fully abstract for $\sqle_\Hst$, \label{it:H*2}
\item $\cD$ is fully abstract for $\Hst$. \label{it:H*3}
\end{enumerate}
\end{theorem}
\begin{proof}
(\ref{it:H*1} $\Rightarrow$ \ref{it:H*2})
We must prove that $M\sqsubseteq_\Hst \!N$ if and only if $\Lint{M}\subseteq\Lint{N}$. The right-to-left implication is true for all \rgm{}'s, since they are sensible (Corollary~\ref{cor:BTminimalth}) and $\sqsubseteq_\Hst$ is the maximal sensible inequational theory. Let us prove the left-to-right implication.

By Theorem~\ref{thm:Hst_as_etainf} the hypothesis $M\sqsubseteq_\Hst \!N$ means that $\BT{M} \leeta[\infty] U \BTle V \geeta[\infty] \BT{N}$ for some B\"ohm-like trees $U$ and $V$. 
In particular $U$ can be taken of the form $U=\BT{P}$ for some $P\in\Lam$ (see~\cite[Ex.~10.6.7]{Bare}). Then we have
\[
\begin{array}{lcl@{\hspace{2cm}}r}
  \Lint M &\subseteq& \Lint{P}                   &\text{by Lemma~\ref{lem:yeppa!hyper!} }  \\
  &= & \Lint{U}               & \text{by  Theorem~\ref{thm:app} } \\
  &\subseteq &\Lint{V}        & \text{by def. \!of $\Lint{-}$ for B\"ohm-like trees } \\
  &\subseteq &\Lint{\BT{N}}  & \text{by Lemma~\ref{lem:yeppa!} } \\
  &=& \Lint{N}                   & \text{by  Theorem~\ref{thm:app}.} 
   \end{array}
\]

(\ref{it:H*2} $\Rightarrow$ \ref{it:H*3}) Trivial.

(\ref{it:H*3} $\Rightarrow$ \ref{it:H*1}) The theory $\Hst$ is extensional, so that is the case for any fully abstract \rgm{}. Moreover, the theory $\Hst$  satisfies $x\sqsubseteq_\Hst \! {\jc_Tx}\,$  for all $\,T\in\Trees\,.$ So the model is hyperimmune by the characterization (\ref{lemma:hyp=>eta!2}) of hyperimmunity provided by Proposition~\ref{lemma:hyp=>eta!}. 
\end{proof}


  

As a consequence, we get that the model $\cD_\omega$ has the same inequational and \lam-theories as Scott's $\cD_\infty$~\cite{Scott72}, namely it is fully abstract for $\Hst$.
Such a result first appeared in~\cite{Manzonetto09}.

\begin{corollary}\label{cor:Domega}
The model $\cD_\omega$ of Example~\ref{ex:graphrel} is inequationally
fully abstract for $\sqle_\Hst$. In particular $\Th(\cD_\omega) =
\Hst$.
\end{corollary}

%% file: include/fig5.tex
\begin{figure}[t!]
\begin{tikzpicture}
\node (myspot) at (-5.9,0) {~};
\node (A) at (-6.5,3.5) {$\bU$};
\node[below of = A, node distance=11pt] (eq) {$\shortparallel$};
\node[below of = eq, node distance=10pt] (l1) {\!\!\!\!\!\!\!\!\!$\lam xy.x$};
\node[below of = l1, node distance=18pt,xshift=20pt] (l2) {$x$};
\node[below of = l2, node distance=18pt,xshift=20pt] (l3) {$x$};
\node[below of = l3, node distance=18pt,xshift=20pt] (l4) {$x$};
\node[below of = l4, xshift=0pt, node distance=10pt] (l5) {$\vdots$};
\draw ($(l1.south)+(10pt,0.05)$) -- ($(l2.north)-(0.1,0.05)$);
\draw ($(l2.south)+(6pt,0.05)$) -- ($(l3.north)-(0.1,0.05)$);
\draw ($(l3.south)+(6pt,0.05)$) -- ($(l4.north)-(0.1,0.05)$);

\node[below of = l1, node distance=18pt,xshift=-20pt] (l12) {$y$};
\draw ($(l1.south)+(-6pt,0.05)$) -- ($(l12.north)+(1pt,-0.1)$);

\node[below of = l2, node distance=18pt,xshift=-20pt] (l22) {$y$};
\draw ($(l2.south)+(-6pt,0.05)$) -- ($(l22.north)+(1pt,-0.1)$);

\node[below of = l3, node distance=18pt,xshift=-20pt] (l32) {$y$};
\draw ($(l3.south)+(-6pt,0.05)$) -- ($(l32.north)+(1pt,-0.1)$);
\node (ge) at (-3,3.5) {$\leeta[!]$};
\node (B) at (.5,3.5) {$\bV$};
\node[below of = B, node distance=11pt] (eq) {$\shortparallel$};
\node[below of = eq, node distance=10pt] (l1) {\!\!\!\!\!\!\!\!\!$\lam xy.x$};
\node[below of = l1, node distance=18pt,xshift=20pt] (l2) {$x$};
\node[below of = l2, node distance=18pt,xshift=20pt] (l3) {$x$};
\node[below of = l3, node distance=18pt,xshift=20pt] (l4) {$x$};
\node[below of = l4, xshift=0pt, node distance=10pt] (l5) {$\vdots$};
\draw ($(l1.south)+(10pt,0.05)$) -- ($(l2.north)-(0.1,0.05)$);
\draw ($(l2.south)+(6pt,0.05)$) -- ($(l3.north)-(0.1,0.05)$);
\draw ($(l3.south)+(6pt,0.05)$) -- ($(l4.north)-(0.1,0.05)$);

\node[below of = l1, node distance=18pt,xshift=-55pt] (l12) {$\lam z_0.y$};
\node[below of = l12, node distance=18pt] (l13) {$\lam z_1.z_0$};
\node[below of = l13, node distance=18pt] (l14) {$\lam z_2.z_1$};
\node[below of = l14, xshift=8pt, node distance=10pt] (l15) {$\vdots$};
\draw ($(l1.south)+(-8pt,0.05)$) -- ($(l12.north)+(8pt,-0.15)$);
\draw ($(l12.south)+(8pt,0.05)$) -- ($(l13.north)+(8pt,-0.15)$);
\draw ($(l13.south)+(8pt,0.05)$) -- ($(l14.north)+(8pt,-0.15)$);

\node[below of = l2, node distance=18pt,xshift=-34pt] (l22) {$\lam z_0.y$};
\node[below of = l22, node distance=18pt] (l23) {$\lam z_1.z_0$};
\node[below of = l23, xshift=8pt, node distance=10pt] (l24) {$\vdots$};
\draw ($(l2.south)+(-2pt,0.05)$) -- ($(l22.north)+(8pt,-0.15)$);
\draw ($(l22.south)+(8pt,0.05)$) -- ($(l23.north)+(8pt,-0.15)$);

\node[below of = l3, node distance=18pt,xshift=-17pt] (l32) {$\lam z_0.y$};
\node[below of = l32, xshift=8pt, node distance=10pt] (l33) {$\vdots$};
\draw ($(l3.south)+(0pt,0.05)$) -- ($(l32.north)+(8pt,-0.15)$);

\end{tikzpicture}
\caption{An example of infinitely many $\eta!$-reductions. \label{fig:exEta!}}
\end{figure}

%% file: include/conclusions.tex

We have studied the class of the relational graph models living inside the relational semantics of \lam-calculus, and proved that they all enjoy the Approximation Theorem.
We exhibited a model inducing the minimum relational (in)equational graph theory, and provided sufficient and necessary conditions for a relational graph model to be fully abstract for $\Hpl$ (resp.\ $\Hst$). 
Actually such characterizations of full abstraction hold more generally for all relational models, since these theories are extensional and the class of extensional relational graph models coincide with the class of extensional reflexive objects in \MRel.

We conclude presenting some open problems that we consider interesting.



\begin{problem}
It is well known that the \lam-theory $\Hst$ satisfies the $\omega$-rule~\cite[Def.~4.1.10]{Bare}, and the analogous result was recently proved for $\Hpl$ in~\cite{BreuvartMPR16}.
Does every extensional graph model satisfy the $\omega$-rule?
\end{problem}

\begin{problem} Are all \lam-theories in the interval $[\Hpl,\Hst]$ relational graph theories? 
If it is not the case, is it possible to provide a characterization of the representable ones?
\end{problem}



%% file: main.bbl
\newcommand{\online}[1]{Available at \url{#1}}
\begin{thebibliography}{10}

\bibitem{abelPTS13}
A.~Abel, B.~Pientka, D.~Thibodeau, and A.~Setzer.
\newblock Copatterns: programming infinite structures by observations.
\newblock In R.~Giacobazzi and R.~Cousot, editors, {\em The 40th Annual {ACM}
  {SIGPLAN-SIGACT} Symposium on Principles of Programming Languages, {POPL}
  '13}, pages 27--38. {ACM}, 2013.

\bibitem{Abramsky91}
S.~Abramsky.
\newblock Domain theory in logical form.
\newblock {\em Ann. Pure Appl. Logic}, 51(1-2):1--77, 1991.

\bibitem{AmadioC98}
R.~Amadio and P.-L. Curien.
\newblock {\em Domains and Lambda Calculi}.
\newblock Cambridge tracts in theoretical computer science. Cambridge
  University Press, 1998.

\bibitem{AschieriZ13}
F.~Aschieri and M.~Zorzi.
\newblock Non-determinism, non-termination and the strong normalization of
  system {T}.
\newblock In M.~Hasegawa, editor, {\em Typed Lambda Calculi and Applications,
  11th International Conference, {TLCA} 2013. Proceedings}, volume 7941 of {\em
  Lecture Notes in Computer Science}, pages 31--47. Springer, 2013.

\bibitem{AspertiL91}
A.~Asperti and G.~Longo.
\newblock {\em Categories, types and structures: an introduction to category
  theory for the working computer scientist}.
\newblock MIT Press, Cambridge, MA, 1991.

\bibitem{Barendregt77}
H.P. Barendregt.
\newblock The type free lambda calculus.
\newblock In J.~Barwise, editor, {\em Handbook of Mathematical Logic},
  volume~90 of {\em Studies in Logic and the Foundations of Mathematics}, pages
  1091--1132. North-Holland, Amsterdam, 1977.

\bibitem{Bare}
H.P. Barendregt.
\newblock {\em The lambda-calculus, its syntax and semantics}.
\newblock Number 103 in Studies in Logic and the Foundations of Mathematics.
  North-Holland, second edition, 1984.

\bibitem{interbcd}
H.P. Barendregt, M.~Coppo, and M.~Dezani-Ciancaglini.
\newblock A filter lambda model and the completeness of type assignment.
\newblock {\em Journal of Symbolic Logic}, 48:931--940, 1983.

\bibitem{BareTypes}
H.P. Barendregt, W.~Dekkers, and R.~Statman.
\newblock {\em Lambda Calculus with Types}.
\newblock Perspectives in logic. Cambridge University Press, 2013.

\bibitem{Berline00}
C.~Berline.
\newblock From computation to foundations via functions and application: The
  $\lambda$-calculus and its webbed models.
\newblock {\em Theoretical Computer Science}, 249(1):81--161, 2000.

\bibitem{BerlineMS09}
C.~Berline, G.~Manzonetto, and A.~Salibra.
\newblock Effective lambda-models versus recursively enumerable
  lambda-theories.
\newblock {\em Mathematical Structures in Computer Science}, 19(5):897--942,
  2009.

\bibitem{Berry78}
G.~Berry.
\newblock Stable models of typed lambda-calculi.
\newblock In {\em Proceedings of the Fifth Colloquium on Automata, Languages
  and Programming}, volume~62 of {\em LNCS}, Berlin, 1978. Springer-Verlag.

\bibitem{Bierman95}
G.~Bierman.
\newblock What is a categorical model of {L}inear {L}ogic?
\newblock In M.~Dezani-Ciancaglini and G.~Plotkin, editors, {\em Typed Lambda
  Calculi and Applications '95}, volume 902 of {\em Lecture Notes in Computer
  Science}. Springer, April 1995.

\bibitem{Bohm68}
C.~B\"ohm.
\newblock Alcune propriet\`a delle forme $\beta$-$\eta$-normali nel
  $\lambda$-{$K$}-calcolo.
\newblock {\em INAC}, 696:1--19, 1968.

\bibitem{Breuvart13}
F.~Breuvart.
\newblock The resource lambda calculus is short-sighted in its relational
  model.
\newblock In M.~Hasegawa, editor, {\em Typed Lambda Calculi and Applications,
  11th International Conference, {TLCA} 2013. Proceedings}, volume 7941 of {\em
  Lecture Notes in Computer Science}, pages 93--108. Springer, 2013.

\bibitem{Breuvart14}
F.~Breuvart.
\newblock On the characterization of models of $\mathcal{H}^*$.
\newblock In T.~A. Henzinger and D.~Miller, editors, {\em Joint Meeting of the
  23 {EACSL} Annual Conference on Computer Science Logic {(CSL)} and the 29
  Annual {ACM/IEEE} Symposium on Logic in Computer Science (LICS)}, pages
  24:1--24:10. {ACM}, 2014.

\bibitem{Breuvart16}
F.~Breuvart.
\newblock On the characterization of models of $\mathcal{H}^*$: the semantical
  aspect.
\newblock {\em Logical Methods in Computer Science}, 12(2), 2016.

\bibitem{BreuvartMPR16}
F.~Breuvart, G.~Manzonetto, A.~Polonsky, and D.~Ruoppolo.
\newblock New results on {M}orris's observational theory.
\newblock In D.~Kesner and B.~Pientka, editors, {\em Formal Structures for
  Computation and Deduction}, volume~52 of {\em LIPIcs}, pages 15:1--15:18.
  Schloss Dagstuhl, 2016.

\bibitem{BucciarelliE91}
A.~Bucciarelli and T.~Ehrhard.
\newblock Sequentiality and strong stability.
\newblock In {\em Sixth Annual IEEE Symposium on Logic in Computer Science},
  pages 138--145. IEEE Computer Society Press, 1991.

\bibitem{BucciarelliEM07}
A.~Bucciarelli, T.~Ehrhard, and G.~Manzonetto.
\newblock Not enough points is enough.
\newblock In J.~Duparc and T.~A. Henzinger, editors, {\em Computer Science
  Logic, 21st International Workshop, {CSL} 2007. Proceedings}, volume 4646 of
  {\em Lecture Notes in Computer Science}, pages 298--312. Springer, 2007.

\bibitem{BucciarelliEM12}
A.~Bucciarelli, T.~Ehrhard, and G.~Manzonetto.
\newblock A relational semantics for parallelism and non-determinism in a
  functional setting.
\newblock {\em Ann. Pure Appl. Logic}, 163(7):918--934, 2012.

\bibitem{BucciarelliS08}
A.~Bucciarelli and A.~Salibra.
\newblock Graph lambda theories.
\newblock {\em Mathematical Structures in Computer Science}, 18(5):975--1004,
  2008.

\bibitem{CarraroES10}
A.~Carraro, T.~Ehrhard, and A.~Salibra.
\newblock Exponentials with infinite multiplicities.
\newblock In A.~Dawar and H.~Veith, editors, {\em Computer Science Logic, 24th
  International Workshop, {CSL} 2010. Proceedings}, volume 6247 of {\em Lecture
  Notes in Computer Science}, pages 170--184. Springer, 2010.

\bibitem{CarraroG14}
A.~Carraro and G.~Guerrieri.
\newblock A semantical and operational account of call-by-value solvability.
\newblock In Anca Muscholl, editor, {\em Foundations of Software Science and
  Computation Structures - 17th International Conference, {FOSSACS} 2014, Held
  as Part of the European Joint Conferences on Theory and Practice of Software,
  {ETAPS} 2014, Grenoble, France, April 5-13, 2014, Proceedings}, volume 8412
  of {\em Lecture Notes in Computer Science}, pages 103--118. Springer, 2014.

\bibitem{Church41}
A.~Church.
\newblock {\em The calculi of lambda conversion}.
\newblock Princeton Univ. Press, 1941.

\bibitem{CoppoDZ87}
M.~Coppo, M.~Dezani, and M.~Zacchi.
\newblock Type theories, normal forms and $\mathcal{D}_\infty$-lambda-models.
\newblock {\em Inf.\ Comput.}, 72(2):85--116, 1987.

\bibitem{CoppoDHL84}
M.~Coppo, M.~Dezani-Ciancaglini, F.~Honsell, and G.~Longo.
\newblock Extended type structures and filter lambda models.
\newblock In G.~Longo G.~Lolli and A.~Marcja, editors, {\em Logic Colloquium
  '82}, volume 112 of {\em Studies in Logic and the Foundations of
  Mathematics}, pages 241 -- 262. Elsevier, 1984.

\bibitem{CoppoDR78}
M.~Coppo, M.~Dezani{-}Ciancaglini, and S.~Ronchi~Della Rocca.
\newblock ({S}emi)-separability of finite sets of terms in {S}cott's
  {$\cD_\infty$}-models of the lambda-calculus.
\newblock In G.~Ausiello and C.~B{\"{o}}hm, editors, {\em Automata, Languages
  and Programming}, volume~62 of {\em Lecture Notes in Computer Science}, pages
  142--164. Springer, 1978.

\bibitem{deCarvalhoTh}
D.~de~Carvalho.
\newblock {\em S\'emantiques de la logique lin\'eaire et temps de calcul}.
\newblock PhD thesis, Universit\'e Aix-Marseille II, 2007.
\newblock Th\`ese de Doctorat.

\bibitem{deCarvalho09}
D.~de~Carvalho.
\newblock Execution time of $\lambda$-terms via denotational semantics and
  intersection types.
\newblock Draft available at \url{http://arxiv.org/abs/0905.4251}, 2009.

\bibitem{CarvalhoPF11}
D.~de~Carvalho, M.~Pagani, and L.~Tortora~de Falco.
\newblock A semantic measure of the execution time in linear logic.
\newblock {\em Th. Comp. Sci.}, 412(20):1884--1902, 2011.

\bibitem{deLiguoroP95}
U.~de'Liguoro and A.~Piperno.
\newblock Non deterministic extensions of untyped lambda-calculus.
\newblock {\em Inf. Comput.}, 122(2):149--177, 1995.

\bibitem{DezaniLP98}
M.~Dezani{-}Ciancaglini, U.~de'Liguoro, and A.~Piperno.
\newblock A filter model for concurrent lambda-calculus.
\newblock {\em {SIAM} J. Comput.}, 27(5):1376--1419, 1998.

\bibitem{Diaz-CaroMP13}
A.~D{\'{\i}}az{-}Caro, G.~Manzonetto, and M.~Pagani.
\newblock Call-by-value non-determinism in a linear logic type discipline.
\newblock In S.~N. Art{\"{e}}mov and A.~Nerode, editors, {\em Logical
  Foundations of Computer Science, {LFCS}. Proceedings}, volume 7734 of {\em
  Lecture Notes in Computer Science}, pages 164--178. Springer, 2013.

\bibitem{Ehrhard12}
T.~Ehrhard.
\newblock Collapsing non-idempotent intersection types.
\newblock In P.~C{\'{e}}gielski and A.~Durand, editors, {\em Computer Science
  Logic (CSL'12), 21st Annual Conference of the EACSL, {CSL} 2012}, volume~16
  of {\em LIPIcs}, pages 259--273. Schloss Dagstuhl - Leibniz-Zentrum fuer
  Informatik, 2012.

\bibitem{EhrhardG16}
T.~Ehrhard and G.~Guerrieri.
\newblock The bang calculus: an untyped lambda-calculus generalizing
  call-by-name and call-by-value.
\newblock In J.~Cheney and G.~Vidal, editors, {\em Proceedings of the 18th
  International Symposium on Principles and Practice of Declarative
  Programming}, pages 174--187. {ACM}, 2016.

\bibitem{lambdadiff}
T.~Ehrhard and L.~Regnier.
\newblock The differential $\lambda$-calculus.
\newblock {\em Th. Comp. Sci.}, 309(1):1--41, 2003.

\bibitem{bohmtaylor}
T.~Ehrhard and L.~Regnier.
\newblock B{\"o}hm trees, {K}rivine's machine and the {T}aylor expansion of
  lambda-terms.
\newblock In {\em CiE}, volume 3988 of {\em Lecture Notes in Computer Science},
  pages 186--197, 2006.

\bibitem{Engeler81}
E.~Engeler.
\newblock Algebras and combinators.
\newblock {\em Algebra Universalis}, 13(3):389--392, 1981.

\bibitem{jdm096}
M.~Fiore, N.~Gambino, M.~Hyland, and G.~Winskel.
\newblock The cartesian closed bicategory of generalised species of structures.
\newblock {\em Journal of the London Mathematical Society}, 77(1):203, 2008.

\bibitem{GianantonioFH99}
P.~Di Gianantonio, G.~Franco, and F.~Honsell.
\newblock Game semantics for untyped $\lambda \beta \eta$-calculus.
\newblock In {\em TLCA'99}, volume 1581 of {\em Lecture Notes in Computer
  Science}, pages 114--128. Springer, 1999.

\bibitem{Girard72}
J.-Y. Girard.
\newblock {\em Interpr\'etation {F}onctionnelle et {\'E}limination des
  {C}oupures de l'{A}rithm\'etique d'{O}rdre {S}up\'erieur}.
\newblock Th\`ese de doctorat, Universit\'e Paris 7, 1972.

\bibitem{Girard88}
J.-Y. Girard.
\newblock Normal functors, power series and $\lambda$-calculus.
\newblock {\em Ann. of Pure and App. Logic}, 37(2):129--177, 1988.

\bibitem{GouyTh}
X.~Gouy.
\newblock {\em {\'E}tude des th\'eories \'equationnelles et des propri\'et\'es
  alg\'ebriques des mod\`eles stables du $\lambda$-calcul}.
\newblock Th\`ese de doctorat, Universit\'e de Paris~7, 1995.

\bibitem{HonsellR92}
F.~Honsell and S.~Ronchi~Della Rocca.
\newblock An approximation theorem for topological lambda models and the
  topological incompleteness of lambda calculus.
\newblock {\em J. of Computer and System Sciences}, 45:49--75, 1992.

\bibitem{Hyland75errato}
J.M.E. Hyland.
\newblock A survey of some useful partial order relations on terms of the
  $\lambda$-calculus.
\newblock In {\em Lambda-Calculus and Comp. Sci. Th.}, volume~37 of {\em
  Lecture Notes in Computer Science}, pages 83--95. Springer, 1975.

\bibitem{Hyland76}
J.M.E. Hyland.
\newblock A syntactic characterization of the equality in some models for the
  $\lambda$-calculus.
\newblock {\em J. London Math. Soc. (2)}, 12(3):361--370, 1976.

\bibitem{Hyland10}
J.M.E. Hyland.
\newblock Some reasons for generalising domain theory.
\newblock {\em Mathematical Structures in Computer Science}, 20(2):239--265,
  2010.

\bibitem{HylandNPR06}
J.M.E. Hyland, M.~Nagayama, J.~Power, and G.~Rosolini.
\newblock A category theoretic formulation for {E}ngeler-style models of the
  untyped $\lambda$-calculus.
\newblock {\em Electr. Notes in TCS}, 161:43--57, 2006.

\bibitem{IntrigilaMP18}
B.~Intrigila, G.~Manzonetto, and A.~Polonsky.
\newblock Degrees of extensionality in the theory of {B}\"ohm trees and
  {S}all\'e's conjecture.
\newblock Submitted. CoRR abs/1802.07320.

\bibitem{IntrigilaMP17}
B.~Intrigila, G.~Manzonetto, and A.~Polonsky.
\newblock Refutation of {S}all{\'{e}}'s longstanding conjecture.
\newblock In D.~Miller, editor, {\em 2nd International Conference on Formal
  Structures for Computation and Deduction, {FSCD} 2017}, volume~84 of {\em
  LIPIcs}, pages 20:1--20:18. Schloss Dagstuhl - Leibniz-Zentrum fuer
  Informatik, 2017.

\bibitem{JacobsR97}
B.~Jacobs and J.~Rutten.
\newblock A tutorial on (co)algebras and (co)induction.
\newblock {\em EATCS Bulletin}, 62:62--222, 1997.

\bibitem{Kleene36}
S.~C. Kleene.
\newblock {\(\lambda\)}-definability and recursiveness.
\newblock {\em Duke Math. J.}, 2(2):340--353, 06 1936.

\bibitem{Koymans82}
C.P.J. Koymans.
\newblock Models of the lambda calculus.
\newblock {\em Information and Control}, 52(3):306--332, 1982.

\bibitem{KozenS17}
D.~Kozen and A.~Silva.
\newblock Practical coinduction.
\newblock {\em Mathematical Structures in Computer Science}, 27(7):1132--1152,
  2017.

\bibitem{Krivine93}
J.-L. Krivine.
\newblock {\em Lambda-calculus, types and models}.
\newblock Ellis Horwood, New York, 1993.
\newblock Translated from the ed. Masson, 1990, French original.

\bibitem{KrivineClass}
J.-L. Krivine.
\newblock Realizability in classical logic.
\newblock In P.-L. Curien, H.~Herbelin, J.-L. Krivine, and P.-A. Melli\`es,
  editors, {\em Interactive models of computation and program behaviour},
  number~27 in Panoramas et synth\`eses. Soci\'et\'e Math\'ematique de France,
  2009.

\bibitem{LairdMMP13}
J.~Laird, G.~Manzonetto, G.~McCusker, and M.~Pagani.
\newblock Weighted relational models of typed lambda-calculi.
\newblock In {\em 28th Annual {ACM/IEEE} Symposium on Logic in Computer
  Science, {LICS} 2013}, pages 301--310. {IEEE} Computer Society, 2013.

\bibitem{Lassen99}
S.~Lassen.
\newblock Bisimulation in untyped lambda calculus: B{\"{o}}hm trees and
  bisimulation up to context.
\newblock {\em Electr. Notes Theor. Comput. Sci.}, 20:346--374, 1999.

\bibitem{JJ}
J.-J. L\'evy.
\newblock Le lambda calcul - notes du cours, 2005.
\newblock In French.\\
  \url{http://pauillac.inria.fr/~levy/courses/X/M1/lambda/dea-spp/jjl.pdf}.

\bibitem{Levy06}
P.~B. Levy.
\newblock Call-by-push-value: Decomposing call-by-value and call-by-name.
\newblock {\em Higher-Order and Symbolic Computation}, 19(4):377--414, 2006.

\bibitem{Longo83}
G.~Longo.
\newblock Set-theoretical models of $\lambda $-calculus: theories, expansions,
  isomorphisms.
\newblock {\em Ann. Pure Appl. Logic}, 24(2):153--188, 1983.

\bibitem{LusinS04}
S.~Lusin and A.~Salibra.
\newblock The lattice of $\lambda$-theories.
\newblock {\em J. Log. Comput.}, 14(3):373--394, 2004.

\bibitem{ManzonettoTh}
G.~Manzonetto.
\newblock {\em Models and theories of lambda calculus}.
\newblock Th\`ese de doctorat, Univ. Ca'Foscari (Venice) and Univ. Paris
  Diderot (Paris 7), 2008.

\bibitem{Manzonetto09}
G.~Manzonetto.
\newblock A general class of models of $\mathcal{H}^{\star}$.
\newblock In {\em Mathematical Foundations of Computer Science 2009, 34th
  International Symposium, {MFCS} 2009. Proceedings}, volume 5734 of {\em
  Lecture Notes in Computer Science}, pages 574--586. Springer, 2009.

\bibitem{Manzonetto12}
G.~Manzonetto.
\newblock What is a categorical model of the differential and the resource
  $\lambda$-calculi?
\newblock {\em Mathematical Structures in Computer Science}, 22(3):451--520,
  2012.

\bibitem{ManzonettoR14}
G.~Manzonetto and D.~Ruoppolo.
\newblock Relational graph models, {T}aylor expansion and extensionality.
\newblock {\em Electr. Notes Theor. Comput. Sci.}, 308:245--272, 2014.

\bibitem{ManzonettoS06}
G.~Manzonetto and A.~Salibra.
\newblock Boolean algebras for lambda calculus.
\newblock In {\em 21th {IEEE} Symposium on Logic in Computer Science {(LICS}
  2006). Proceedings}, pages 317--326. {IEEE} Computer Society, 2006.

\bibitem{Morristh}
J.H. Morris.
\newblock {\em Lambda calculus models of programming languages}.
\newblock Phd, MIT, 1968.

\bibitem{PaganiR10}
M.~Pagani and S.~Ronchi~Della Rocca.
\newblock Linearity, non-determinism and solvability.
\newblock {\em Fundam. Inform.}, 103(1-4):173--202, 2010.

\bibitem{Paolini08}
L.~Paolini.
\newblock Parametric $\lambda$-theories.
\newblock {\em Theoretical Computer Science}, 398(1):51 -- 62, 2008.

\bibitem{PaoliniPR15}
L.~Paolini, M.~Piccolo, and S.~Ronchi Della~Rocca.
\newblock Essential and relational models.
\newblock {\em Mathematical Structures in Computer Science}, 27(5):626--650,
  2017.

\bibitem{Plotkin71}
G.D. Plotkin.
\newblock A set-theoretical definition of application.
\newblock Technical Report MIP-R-95, School of artificial intelligence, 1971.

\bibitem{Plotkin93}
G.D. Plotkin.
\newblock Set-theoretical and other elementary models of the
  $\lambda$-calculus.
\newblock {\em Theoretical Computer Science}, 121(1):351 -- 409, 1993.

\bibitem{Reynolds83}
J.C. Reynolds.
\newblock Types, abstraction and parametric polymorphism.
\newblock In {\em {IFIP} Congress}, pages 513--523, 1983.

\bibitem{Ronchi82}
S.~Ronchi~Della Rocca.
\newblock Characterization theorems for a filter lambda model.
\newblock {\em Information and Control}, 54(3):201--216, 1982.

\bibitem{RonchiP04}
S.~Ronchi~Della Rocca and L.~Paolini.
\newblock {\em The Parametric Lambda Calculus - {A} Metamodel for Computation}.
\newblock Texts in Theoretical Computer Science. An {EATCS} Series. Springer,
  2004.

\bibitem{phdruoppolo}
D.~Ruoppolo.
\newblock {\em Relational Graph Models and Morris's Observability:
  resource-sensitive semantic investigations on the untyped
  {$\lambda$}-calculus}.
\newblock Th\'ese de doctorat, Universit{\'e} Paris 13, 2016.

\bibitem{Salibra01}
A.~Salibra.
\newblock Nonmodularity results for lambda calculus.
\newblock {\em Fundam. Inform.}, 45(4):379--392, 2001.

\bibitem{SalvatiMGB12}
S.~Salvati, G.~Manzonetto, M.~Gehrke, and H.~Barendregt.
\newblock Loader and {U}rzyczyn are logically related.
\newblock In A.~Czumaj, K.~Mehlhorn, A.M. Pitts, and R.~Wattenhofer, editors,
  {\em Automata, Languages, and Programming - 39th International Colloquium,
  {ICALP} 2012. Proceedings, Part {II}}, volume 7392 of {\em Lecture Notes in
  Computer Science}, pages 364--376. Springer, 2012.

\bibitem{Scott72}
D.~Scott.
\newblock Continuous lattices.
\newblock In Lawvere, editor, {\em Toposes, Algebraic Geometry and Logic},
  volume 274 of {\em Lecture Notes in Math.}, pages 97--136. Springer, 1972.

\bibitem{Scott76}
D.S. Scott.
\newblock Data types as lattices.
\newblock {\em {SIAM} J. Comput.}, 5(3):522--587, 1976.

\bibitem{Selinger02}
P.~Selinger.
\newblock The lambda calculus is algebraic.
\newblock {\em J. Funct. Program.}, 12(6):549--566, 2002.

\bibitem{SeveriV16}
P.~Severi and F.{-}J. de~Vries.
\newblock The infinitary lambda calculus of the infinite eta {B\"{o}}hm trees.
\newblock {\em Mathematical Structures in Computer Science}, 27(5):681--733,
  2017.

\bibitem{Tait67}
W.~Tait.
\newblock Intensional interpretations of functionals of finite type {I}.
\newblock {\em J. Symb. Log.}, 32(2):198--212, 1967.

\bibitem{phdtranquilli}
P.~Tranquilli.
\newblock {\em Nets between {D}eterminism and {N}ondeterminism}.
\newblock Ph.{D}.\ thesis, Universit\`a Roma Tre/Universit\'e Paris Diderot
  (Paris 7), April 2009.

\bibitem{Turing37}
A.~M. Turing.
\newblock Computability and {\(\lambda\)}-definability.
\newblock {\em J. Symb. Log.}, 2(4):153--163, 1937.

\bibitem{Wadsworth76}
C.P. Wadsworth.
\newblock The relation between computational and denotational properties for
  {S}cott's $\mathscr{D}_{\infty}$-models of the lambda-calculus.
\newblock {\em {SIAM} J. Comput.}, 5(3):488--521, 1976.

\bibitem{Wadsworth78}
C.P. Wadsworth.
\newblock Approximate reduction and lambda calculus models.
\newblock {\em {SIAM} J. Comput.}, 7(3):337--356, 1978.

\end{thebibliography}
